\documentclass[prx,aps,amsfonts,eqsecnum,superscriptaddress,nofootinbib,longbibliography,notitlepage]{revtex4-1}

\usepackage{cmap}
\usepackage{graphicx}
\usepackage{xcolor}
\usepackage{subcaption}
\usepackage{enumerate}
\usepackage{rotating}
\usepackage{amsmath, amssymb, graphics, amsthm, isomath}
\usepackage{bbm}
\usepackage{physics}
\usepackage{verbatim}
\usepackage{tikz}
\usepackage{tikz-cd}
\usepackage{circuitikz}
\usetikzlibrary{patterns}
\usetikzlibrary{through,hobby}
\usetikzlibrary{decorations.markings}
\usetikzlibrary{fadings,shadings}
\usetikzlibrary{plotmarks}
\usetikzlibrary{positioning}
\usetikzlibrary{decorations,arrows}
\usetikzlibrary{decorations.pathreplacing}
\usetikzlibrary{chains, scopes, positioning, backgrounds, shapes, fit, shadows, calc, arrows.meta, decorations.pathreplacing}
\usepackage{blkarray}
\usepackage{mathrsfs}
\usepackage{algorithm}
\usepackage{algpseudocode}

\usepackage{xcolor}
\definecolor{mycitecolor}{rgb}{0.0, 0.45, 0.85}   

\usepackage{simpler-wick}

\newcommand{\PSL}{\operatorname{PSL}}

\newcommand{\U}{\operatorname{U}}

\newcommand{\diaggroup}{G^n_{\textnormal{diag}}}
\newcommand{\diagX}{H}

\newcommand{\tinyhub}{%
\begin{tikzpicture}[baseline=-0.6ex, line cap=round, line join=round, x=1ex, y=1ex]
  \draw[thick] (-3,0)--(3,0);
  \draw[thick] (0,-3)--(0,3);
  \draw[thick] (-2,-2)--(2,2);
  \node[circle,draw,fill=black,minimum size=2.5ex,inner sep=0pt] at (-1.25,0) {};
\end{tikzpicture}%
}

\newcommand{\tinytwoleg}{%
\begin{tikzpicture}[baseline=-0.6ex, line cap=round, line join=round, x=1ex, y=1ex]
  \draw[thick] (-3,0)--(3,0);
  \draw[thick] (0,0)--(0,3);
  \node[circle,draw,fill=white,minimum size=2.5ex,inner sep=0pt] at (-1.25,0) {};
\end{tikzpicture}%
}

\newcommand{\Xrighticon}{%
\begin{tikzpicture}[baseline=-0.6ex, x=1ex, y=1ex, line cap=round, line join=round]
   \tikzstyle{sergio}=[rectangle,draw=none]
\draw[thick] (1,0) -- (1,-4);
 \begin{scope}
    \draw[thick, fill=white] (-2,-2) -- (0,2) -- (4,2) -- (2,-2) -- cycle;
  \end{scope}
  \draw[thick] (1,0) -- (1,4);
 \path (1,0) node [style=sergio]{\footnotesize$\overrightarrow{X}_g$};
\end{tikzpicture}%
}

\newcommand{\Xlefticon}{%
\begin{tikzpicture}[baseline=-0.6ex, x=1ex, y=1ex, line cap=round, line join=round]
  \tikzstyle{sergio}=[rectangle,draw=none]
\draw[thick] (1,0) -- (1,-4);
 \begin{scope}
    \draw[thick, fill=white] (-2,-2) -- (0,2) -- (4,2) -- (2,-2) -- cycle;
  \end{scope}
  \draw[thick] (1,0) -- (1,4);
 \path (1,0) node [style=sergio]{\footnotesize$\overleftarrow{X}_g$};
\end{tikzpicture}%
}

\newcommand{\Xstabilizer}{%
\begin{tikzpicture}[baseline=-0.6ex, x=1ex, y=1ex, line cap=round, line join=round]
  \tikzstyle{sergio}=[rectangle,draw=none]
\draw[thick] (1,0) -- (1,-4);
\draw[thick] (1,0) -- (1,-4);
 \begin{scope}
    \draw[thick, fill=white] (-2,-2) -- (0,2) -- (4,2) -- (2,-2) -- cycle;
  \end{scope}
  \draw[thick] (1,0) -- (1,4);
 \path (1,0) node [style=sergio]{\footnotesize$\overleftarrow{X}_g$};
 \begin{scope}[xshift=30]
 \tikzstyle{sergio}=[rectangle,draw=none]
\draw[thick] (1,0) -- (1,-4);
\draw[thick] (1,0) -- (1,-4);
 \begin{scope}
    \draw[thick, fill=white] (-2,-2) -- (0,2) -- (4,2) -- (2,-2) -- cycle;
  \end{scope}
  \draw[thick] (1,0) -- (1,4);
 \path (1,0) node [style=sergio]{\footnotesize$\overrightarrow{X}_g$};
 \end{scope}
\begin{scope}[xshift=60]
 \tikzstyle{sergio}=[rectangle,draw=none]
\draw[thick] (1,0) -- (1,-4);
\draw[thick] (1,0) -- (1,-4);
 \begin{scope}
    \draw[thick, fill=white] (-2,-2) -- (0,2) -- (4,2) -- (2,-2) -- cycle;
  \end{scope}
  \draw[thick] (1,0) -- (1,4);
 \path (1,0) node [style=sergio]{\footnotesize$\overrightarrow{X}_g$};
 \end{scope}
 \begin{scope}[xshift=45, yshift=30]
 \tikzstyle{sergio}=[rectangle,draw=none]
\draw[thick] (1,0) -- (1,-4);
\draw[thick] (1,0) -- (1,-4);
 \begin{scope}
    \draw[thick, fill=white] (-2,-2) -- (0,2) -- (4,2) -- (2,-2) -- cycle;
  \end{scope}
  \draw[thick] (1,0) -- (1,4);
 \path (1,0) node [style=sergio]{\footnotesize$\overleftarrow{X}_g$};
 \end{scope}
\begin{scope}[xshift=15, yshift=-30]
 \tikzstyle{sergio}=[rectangle,draw=none]
\draw[thick] (1,0) -- (1,-4);
\draw[thick] (1,0) -- (1,-4);
 \begin{scope}
    \draw[thick, fill=white] (-2,-2) -- (0,2) -- (4,2) -- (2,-2) -- cycle;
  \end{scope}
  \draw[thick] (1,0) -- (1,4);
 \path (1,0) node [style=sergio]{\footnotesize$\overrightarrow{X}_g$};
 \end{scope}
 \end{tikzpicture}%
}

\newcommand{\Zicon}{%
\begin{tikzpicture}[baseline=-0.6ex, x=1ex, y=1ex, line cap=round, line join=round]
   \tikzstyle{sergio}=[rectangle,draw=none]
\draw[thick] (1,0) -- (1,-4);
\draw[thick] (1,0) -- (-3,0);
\draw[thick] (1,0) -- (5,0);
 \begin{scope}
    \draw[thick, fill=white] (-2,-2) -- (0,2) -- (4,2) -- (2,-2) -- cycle;
  \end{scope}
  \draw[thick] (1,0) -- (1,4);
 \path (1,0) node [style=sergio]{\footnotesize$Z_\Gamma$};
\end{tikzpicture}%
}

\newcommand{\Zstabilizer}{%
\begin{tikzpicture}[baseline=-0.6ex, x=1ex, y=1ex, line cap=round, line join=round]
   \tikzstyle{sergio}=[rectangle,draw=none]
\draw[thick] (1,0) -- (1,-4);
\draw[thick] (1,0) -- (-3,0);
\draw[thick] (1,0) -- (5,0);
 \begin{scope}
    \draw[thick, fill=white] (-2,-2) -- (0,2) -- (4,2) -- (2,-2) -- cycle;
  \end{scope}
  \draw[thick] (1,0) -- (1,4);
 \path (1,0) node [style=sergio]{\footnotesize$Z_\Gamma^\dag$};
\begin{scope}[xshift=35]
\draw[thick] (1,0) -- (1,-4);
\draw[thick] (1,0) -- (-3,0);
\draw[thick] (1,0) -- (5,0);
 \begin{scope}
    \draw[thick, fill=white] (-2,-2) -- (0,2) -- (4,2) -- (2,-2) -- cycle;
  \end{scope}
  \draw[thick] (1,0) -- (1,4);
 \path (1,0) node [style=sergio]{\footnotesize$Z_\Gamma$};
\end{scope}
\begin{scope}[xshift=70]
\draw[thick] (1,0) -- (1,-4);
\draw[thick] (1,0) -- (-3,0);
\draw[thick] (1,0) -- (5,0);
 \begin{scope}
    \draw[thick, fill=white] (-2,-2) -- (0,2) -- (4,2) -- (2,-2) -- cycle;
  \end{scope}
  \draw[thick] (1,0) -- (1,4);
 \path (1,0) node [style=sergio]{\footnotesize$Z_\Gamma$};
\end{scope}
\draw[thick] (-3,0) -- (-4.5,-3) -- (20,-3) -- (21.5,0);
\end{tikzpicture}%
}

\usepackage[
  colorlinks=true,
  linkcolor=mycitecolor,
  citecolor=mycitecolor,
  urlcolor=mycitecolor,
  hyperindex=true,
  linktocpage=true
]{hyperref}

\usepackage[capitalise,compress]{cleveref}
\usepackage{tcolorbox}

\newtcolorbox{shadedtheorem}{
  colback=gray!15,    
  colframe=white,     
  boxrule=0pt,        
  arc=0pt,            
  left=5pt,           
  right=5pt,
  top=5pt,
  bottom=5pt
}

\newcommand\Ad{\operatorname{Ad}}
\newcommand\Hom{\operatorname{Hom}}

\newcommand\SU{\operatorname{SU}}
\newcommand\C{\mathbb C}
\newcommand\Z{\mathbb{Z}}
\newcommand\F{\mathbb{F}}
\newcommand{\e}{\mathrm{e}}
\newcommand{\ii}{\mathrm{i}}

\renewcommand\equiv{:=}
\renewcommand\epsilon{\varepsilon}

\newtheorem{thm}{Theorem}
\numberwithin{thm}{section}
\newtheorem{cor}[thm]{Corollary}
\newtheorem{lem}[thm]{Lemma}

\newtheorem{prop}[thm]{Proposition}
\newtheorem{eg}[thm]{Example}

\theoremstyle{definition}

\theoremstyle{definition}
\newtheorem{defn}[thm]{Definition}

\theoremstyle{definition}
\newtheorem{rmk}[thm]{Remark}

\renewcommand{\thesection}{\arabic{section}}
\renewcommand{\thesubsection}{\thesection.\arabic{subsection}}
\renewcommand{\thesubsubsection}{\thesubsection.\arabic{subsubsection}}

\makeatletter
\renewcommand{\p@subsection}{}
\renewcommand{\p@subsubsection}{}
\makeatother

\usepackage{mathtools}
\tikzstyle{densely dashed}= [dash pattern=on 4pt off 3pt]

\usepackage{dsfont}

\begin{document}

\title{Calderbank–Shor–Steane codes on group-valued qudits}

\author{Ben T. McDonough}
\email{ben.mcdonough@colorado.edu}
\affiliation{Department of Physics and Center for Theory of Quantum Matter, University of Colorado, Boulder, CO 80309, USA}

\author{Jian-Hao Zhang}
\email{Sergio.Zhang@colorado.edu}
\affiliation{Department of Physics and Center for Theory of Quantum Matter, University of Colorado, Boulder, CO 80309, USA}

\author{Victor V. Albert}
\affiliation{Joint Center for Quantum Information and Computer Science,
NIST/University of Maryland, College Park, MD 20742, USA}

\author{Andrew Lucas}
\email{andrew.j.lucas@colorado.edu}
\affiliation{Department of Physics and Center for Theory of Quantum Matter, University of Colorado, Boulder, CO 80309, USA}

\begin{abstract}
Calderbank–Shor–Steane (CSS) codes are a versatile quantum
error-correcting family built out of commuting $X$- and $Z$-type checks.  
We introduce CSS-like codes on $G$-valued qudits for any finite group
$G$ that reduce to qubit CSS codes for \(G = \mathbb{Z}_2\) yet generalize the Kitaev quantum double model for general groups.  
The $X$-checks of our group-CSS codes correspond to left and/or right
multiplication by group elements, while $Z$-checks project onto solutions to
group word equations.
We describe quantum-double models on oriented two-dimensional CW complexes (which need not cellulate a manifold) and prove that,
when $G$ is non-Abelian and simple, every $G$-covariant group-CSS code
with suitably upper-bounded $Z$-check weight and lower-bounded $Z$-distance
reduces to a CW quantum double.  
We describe the codespace and logical operators of CW quantum doubles via the same intuition used to obtain logical structure of surface codes.
We obtain distance bounds for codes on non-Abelian simple groups from the graph underlying the CW complex, and construct intrinsically non-Abelian code families with asymptotically optimal rate and distances.
Adding ``ghost vertices" to the CW complex generalizes quantum double models with defects and rough boundary conditions whose logical structure can be understood without reference to non-Abelian anyons or defects.  
Several non-invertible symmetry-protected topological states, both with ordinary and higher-form symmetries, are the unique codewords of simply-connected CW quantum doubles with a single ghost vertex. 
\end{abstract}

\maketitle
\def\thefootnote{**}\footnotetext{B.T.M. and J.-H.Z. contributed equally to this work.} \def\thefootnote{\arabic{footnote}}

\twocolumngrid      
\tableofcontents
\onecolumngrid      

\pagebreak

\section{Introduction}
\setcounter{footnote}{0}
The past three decades have seen the development of a systematic theory of quantum error-correcting codes \cite{albuquerque2009topological, Breuckmann_2016}.  
From its origins to the present day, this theory continues to be intertwined with the theory of topological quantum matter.

Two of the earliest quantum codes---the nine-qubit Shor code \cite{ShorCode} and the Kitaev surface/toric code \cite{Kitaev_1997}---are special cases of the quantum double model \cite{Kitaev_2003, Kitaev_2006, Cui:2019lvb}.
This model, in turn, is based on placing a topological quantum field theory (TQFT)---Dijkgraaf-Witten theory \cite{dijkgraaf1990topological} ---on a triangulation of a two-dimensional manifold (a torus for the toric code, and the projective plane for the Shor code \cite{freedman2002z2}).  
Looking back, we can see this TQFT as effectively the first quantum code.

Dijkgraaf-Witten theory is the gauge theory of flat $G$-valued connections on a manifold
given a finite group $G$.   
The Hamiltonian terms (or ``checks") of its associated code admit very elegant interpretations: ``$Z$-type checks" enforce flatness of the connection, while ``$X$-type checks" enforce gauge-invariance (each codeword is a sum over all gauge-equivalent configurations).   
The topological nature of the TQFT is equivalent to the fact that the codewords of the quantum double are classified by topology (as well as $G$).  

Our modern homological understanding of Calderbank-Shor-Steane (CSS) codes \cite{CSS1, CSS2}, a popular and useful family of stabilizer codes \cite{gottesman1997,calderbank1997quantum} where each check is either an \(X\)-type or a \(Z\)-type Pauli operator, clearly has its roots in this setting (as was already understood by Kitaev \cite{kitaev1997quantum}).  
Indeed, one interprets CSS codes as defined by a certain three-term chain complex---a natural mathematical generalization of 
the assocation of checks to 2D triangulations in quantum double codes.  
Since code parameters on manifolds are severely limited \cite{Bravyi_2009, Bravyi_2010, Delfosse_2013,baspin2022connectivity,baspin2025improved}, it has been important to look for sophisticated chain complexes that cannot be interpreted as cellulations of 2D manifolds in order to obtain more powerful quantum codes, such as asymptotically good codes \cite{panteleev2022asymptotically,leverrier2022quantum,dinur2023good}.  

Many modern codes now have superior code parameters to the older quantum double codes, but they struggle when it comes to implementing non-Clifford logical gates (necessary to perform universal quantum computation) in a fault-tolerant and geometrically local manner.
In order to circumvent stabilizer-code no-go theorems \cite{Bravyi_2013, Pastawski_2015} associated with such unitary gates, one has to either introduce measurements (see, e.g., \cite{gidney2024magic,tiedinknots,williamson2026fast}) or go beyond the stabilizer formalism.
There are currently at least two active efforts aimed at the latter direction.

The first effort substitutes Pauli operators for Clifford operators \cite{cupsgatescohomology,  Lin:2024uhb, tiedinknots, Zhu:2026vec, vedhika_twisted, ni2015non, Yoshida_2016, Webster2022xpstabiliser, Hsin_2025}.
This maintains unitarity of the generalized ``stabilizer'' group and still defines the codespace to be the joint \(+1\)-eigenvalue eigenspace.
However, it comes at the price of losing commutativity outside of the codespace and an unclear syndrome structure in the general case.
Many such constructions are non-geometrically local generalizations of \textit{twisted} quantum doubles \cite{Hu_2013, Hu_2017, CuiTwisted} restricted to Abelian \(G\).

The second effort wraps a stabilizer-like structure around ordinary (untwisted) quantum doubles for non-Abelian \(G\) \cite{albert2021spinchainsdefectsquantum, Fechisin_2025}.
The resulting \(Z\)-type checks become non-unitary matrix-product operators \cite{Fechisin_2025} (cf. \cite{Schuch_2010}), while the non-Clifford \(X\)-type checks remain unitary but no longer commute outside of the codespace.
The syndrome structure is sufficient to decode in principle, albeit syndrome extraction becomes substantially more complex and \(G\)-dependent \cite{verresen2022, bravyi2022, Tantivasadakarn_2024}.
The benefit of all this is a natural way to implement fault-tolerant non-Clifford gates \cite{Mochon_2004}, possible even for the simplest non-Abelian groups such as dihedral groups (as shown recently in Ref.~\cite{warman2026}).
One ``modern way" to understand some of these features is in terms of the quantum double models' \emph{non-invertible symmetries} \cite{schafernameki2023ictplecturesnoninvertiblegeneralized, shao2024whatsundonetasilectures, Luo_2024, thorngren2019fusioncategorysymmetryi, thorngren2021fusioncategorysymmetryii, Choi_2022, Bhardwaj_2022, zhang2023anomalies11dcategoricalsymmetries}: logical operators on the codespace that are most naturally expressed as non-unitary operations.  

One may hope to combine the best of both worlds: the non-invertible symmetries of the quantum double and related models (originally defined with geometric locality) and the good code parameters of quantum expander CSS codes (not local in any finite spatial dimension).   
However, we are not aware of any attempt to define or classify generic \(n\)-body (or \textit{block}) quantum codes with group qudits or cross-reference to existing classical literature.  
This paper is a first attempt.

\section{Summary of results}
\subsection{Group CSS codes and rigidity theorem}
We define and characterize a general class of ``group CSS'' codes that includes quantum doubles and reduces to the usual CSS codes when the group $G=\mathbb{Z}_2$.
The $Z$-type checks of our codes correspond to group word equations that should be satisfied on every codeword, allowing to make use of results in combinatorial group theory \cite{cohen1989combinatorial}.
The $X$-type checks correspond to left and/or right multiplication of qudits by specific motifs of group elements.  
Special cases are summarized as follows:

\begin{enumerate}[1.]
    \item Every group CSS code for Abelian $G$ is equivalent to a group GKP code \cite{Albert_2020}---a distinct generalization of the Abelian setting aimed at the monolithic (single-group-qudit) setting rather than the block (\(n\)-body) setting. 
    In the cyclic case ($G=\mathbb{Z}_m$), the group CSS code is equivalent to an ordinary CSS code defined by a three-term chain complex (Proposition \ref{prop:groupGKPabelian}).
    
    \item For non-cyclic Abelian $G$, we do not find a canonical ``duality" between $X$-checks and $Z$-checks of ordinary CSS codes that is required to be compatible with our construction (Remark~\ref{rmk:XZdual}).
    We can still include such codes in our formalism if we ``unblock'' the group-qudit into a tensor product and embed each factor into a larger cyclic qudit of the same size. 

    \item 
    For non-Abelian simple $G$ and a simple covariance condition, every group CSS code family with $Z$-check weight smaller than an \(\mathrm O(1)\) constant and $Z$-distance greater than the $Z$-check weight is a quantum double model on a \textit{2D CW complex}---a graph with faces glued onto a subset of its loops (Theorem \ref{thm:classify}).  This $\mathrm O(1)$ constant depends on the specific group $G$.
\end{enumerate}

The relationships between our group-CSS construction and other common types of quantum error correction codes are summarized in Fig.~\ref{fig:summary}.

\begin{figure}[t]
    \centering
    \def\svgwidth{0.8\linewidth}
    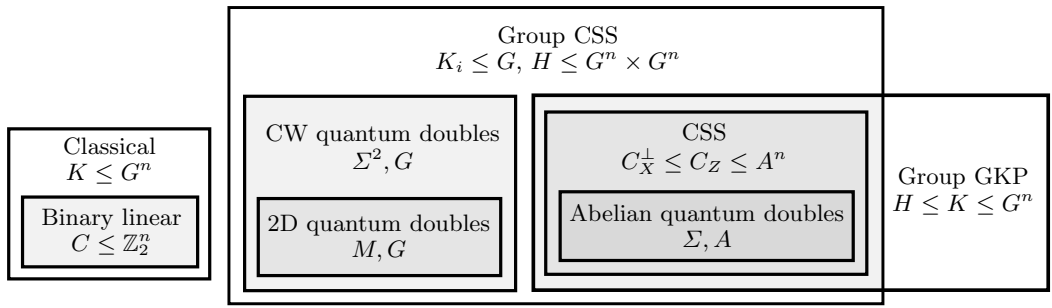
    \caption{
      We illustrate relationships between classical and quantum block codes on \(n\) subsystems defined using a group \(G\) (i.e., group-valued coordinates for the classical case, and group-valued qudits for the quantum case).
      Here, $\leq$ denotes a subgroup and $\subseteq$ denotes a subset. 
      Group CSS codes are specified by subgroups $K_i \leq G$ defining \(Z\)-type constraints along with a subgroup \(H \leq G^{n} \times G^{n}\) defining (left- and right-) \(X\)-type operators. 
      Ordianry (Abelian) CSS codes fall under the group CSS framework if we allow for local blocking of sites and expansion of local dimension. 
      We study a generalization of quantum doubles on 2D manifolds \(M\) to quantum doubles on oriented 2D CW complexes $\Sigma^2$. 
      While quantum doubles for non-Abelian \(G\) utilize only the \(\leq \)2D structures of such complexes, restricting to an Abelian group $G = A$ allows the quantum double to utilize the features of an arbitrary (not necessarily 2D) CW complex $\Sigma$. 
      We identify certain $G$ for which, under mild additional assumptions, all
      $G$-CSS codes are quantum double codes on a 2D CW complex.
      }
    \label{fig:summary}
\end{figure}

\subsection{Code bounds}
The reduction of many group CSS codes to CW quantum doubles per Theorem \ref{thm:classify} demonstrates a surprising rigidity to group CSS codes.
This result implies that there is no purely ``geometric" or ``homological" specification of
stabilizers for a non-Abelian code that can be used to produce good quantum
codes, as was true for Abelian codes. A ``geometry" which relies on a specific group structure would be needed to overcome these bounds.

Nevertheless, since CW complexes go beyond cellulations of 2D manifolds, we can construct constant-rate quantum double codes with provably better distance than any code defined on a 2D manifold, such as hyperbolic surface codes \cite{Breuckmann_2016}.  
We derive general bounds on CW quantum doubles via bounds on their underlying graphs (Theorem~\ref{thm:bounds}) and construct a family of intrinsically non-Abelian group-qudit codes with linear rate, linear \(X\)-distance, and logarithmic \(Z\)-distance (all as functions of \(n\); see Theorem \ref{thm:optimal}).

\subsection{Logical structure}
We show that the codespace and logical operators of CW quantum doubles are natural generalizations of well-known results for surface codes and quantum double models on manifolds.
Namely, we extend (to the CW quantum double setting) the intuition that the logical structure of surface codes comes from non-contractible loops and any defects between boundaries.
This allows us to characterize the logical structure using only the ``topological'' properties of the CW complex and \textit{without} relying on the structure of any non-Abelian excitations \cite{QD_boundary, Cong1, Cong2, Cong3} or defects \cite{Barkeshli_2019}.

Contributions to the logical structure due to non-contractible loops can be extracted from the fundamental group of the CW complex \(\Sigma\). We introduce \textit{ghost vertices} to unify the description of boundaries and holes and obtain topological formulas for the codespace in such models using the complex \(\Sigma^{\prime}\) (Prop.~\ref{prop:codespace_rough}), which is a quotient of the original complex by the set of ghost vertices. Our logical operator formulation places transversal logical gates \cite{Ellison_2026, huang2026hybridlatticesurgerynonclifford}, ribbon operators of quantum doubles, and non-invertible symmetries of other models (Prop.~\ref{prop:Zlogical}) on the same footing.

\subsection{Group-qudit topological orders}
The relaxation of geometric locality of our formalism enables us to incorporate $G$-symmetric ground states of geometrically local lattice models, such as symmetry-protected topological (SPT) states or spontaneous symmetry broken (SSB) states \cite{albert2021spinchainsdefectsquantum,Fechisin_2025}, into our framework.  
Namely, these states can be interpreted as quantum double codes on CW complexes  whose locality is inequivalent to the original model.  
Certain non-trivial facts become readily apparent in this perspective: for example, a global $G^r$ symmetry in any $G$-qudit code is spontaneously broken in the vacuum if $r>1$. 
We also present higher-form non-invertible SPT states, which are also readily explained within the group-CSS framework.

The rest of the paper is organized as follows.  Section \ref{sec:general} defines group CSS codes formally and makes a few simple and general observations about such codes.  Section \ref{sec:QD} analyzes quantum double models on CW complexes as a special case of group CSS codes, well-defined for arbitrary finite $G$.   Section \ref{sec:matter} demonstrates how to interpret SPT/SSB states in this formalism.  Section \ref{sec:rigid} explains the rigidity of group CSS codes for certain groups $G$ and rules out any simple families of expander group CSS codes for generic $G$.  In Section \ref{sec:outlook} we suggest some directions for future research and potential alternatives to the group CSS framework with less rigidity.

\section{Group qudit codes}\label{sec:general}
In this section, we set up a general formalism for studying group qudit codes.  First we will review some historical constructions and observe that they fall short of including the quantum double codes.  We then present ``group CSS" codes to remedy that issue and describe how such codes can be thought of as codes on a coset space.

\subsection{Classical group codes}

A (classical) binary linear code \(C\) is defined as a subset of the group \(G = \mathbb{Z}_2^n\) of \(n\)-bit strings that is closed under binary addition and multiplication \cite{macwilliams1977theory}.
This subset forms a subgroup of \(\mathbb{Z}_2^n\), where the ``group multiplication'' operation is binary addition.
Binary multiplication is unnecessary to define the code since one can only multiply by zero or one, and since the all-zeroes string is always a codeword due to the code being a group.
Binary codes protect against bit-flip errors that flip a few bits at a time.
As long as such bit flips do not bring one codeword closer (in Hamming distance) to another one, they can be corrected.
This structure extends naturally to classical block codes whose coordinates are non-Abelian group elements \cite{slepian2005permutation,gumm1985new,kschischang2002block,forney2002geometrically,loeliger1991signal,forney2002hamming}.
In that case, the message alphabet consists of strings \(\mathbf{g} \in G^n\), where \(G\) is any finite group.
A classical group code is a subgroup of \(G^n\).
We use \(\leq\) to define subgroups, which need not be normal, and \(\subseteq\) to denote subsets.

\begin{defn}[Classical group bit code]\label{def:classicalgroupcode}
    Let $G^n$ denote the space of all possible strings of length $n$, where each coordinate is an element of \(G\).
    The codespace of a classical group code is a subgroup $K \le G^n$.
\end{defn}
 \begin{rmk}\label{rmk:abeliancodelinear}
    This is a very natural definition when \(G\) is Abelian.  
    In that case, $G$ can be defined as a module over $\mathbb{Z}$, and for any $n_1,\ldots, n_p \in \mathbb{Z}$, the solutions to the equation $n_1 g_{i_1} + \cdots +n_p g_{i_p}=0$ form a subgroup of $G^n$.   
    We may therefore define low-density parity-check (LDPC) codes by taking the intersections of many subgroups, each of which comes from enforcing a group word equation as above.  
    The nice qualities of LDPC codes, especially over the fields $\mathbb{Z}_p$, can then be ported into the group code setting.
    For example, using the same ``local'' subgroup for all checks and laying out string coordinates on a regular graph yields a large class of LDPC codes called the Tanner codes \cite{tanner2003recursive}.
\end{rmk}

The group has natural actions on itself by left-multiplication, \(\mathbf{g} \to \mathbf{h}\mathbf{g}\) for any words \(\mathbf{g},\mathbf{h}\in G^n\), and right-multiplication, \(\mathbf{g} \to \mathbf{g}\mathbf{h}^{-1}\).
To define a group code, one first defines a metric on the group analogous to the Hamming weight in the case of $G = \Z_2$, with multiple metrics possible for certain groups \cite{barg1993dawn,barg2010codes,como2005ensembles}.
In the setting of this work, dominant ``bit-flip'' errors are those \(\mathbf{h}\) which have only a few non-identity coordinates, corresponding to a small Hamming weight a.k.a. \textit{support}.

\subsection{Group GKP codes}

In the quantum setting, the code must store superpositions of codewords, and the ``quantum alphabet'' (a.k.a. physical Hilbert space) is the space of complex-valued (wave-)functions on the classical alphabet.
\begin{defn}
Let $G$ be a finite group.  A group qudit code has a physical Hilbert space 
\begin{equation}
\mathbb{C}[G^{n}]:=\mathrm{span}\lbrace|\mathbf{g}\rangle=|g_{1},g_{2},\cdots g_{n}\rangle\rbrace_{\mathbf{g}\in G^{n}}\,.
\end{equation}
\end{defn}

In this setting, one has to protect quantum information not only from the ``bit-flip'' (a.k.a. \(X\)-type) errors inherited from the classical setting, but also from ``phase errors'' (a.k.a. \(Z\)-type) errors which alter the relative phase between superpositions.

Classical binary linear codes can be naturally extended to a simple and tractable class of quantum codes called the CSS codes.
In this case, one binary linear code, \(C_X\) and \(C_Z\), respectively, is picked to protect against each error type.
The two codes have to satisfy the compatibility relation \(C_X^\perp \leq C_Z \leq \mathbb{Z}_2^n\), where \(C_X^\perp\) is the code dual to \(C_X\) \cite{macwilliams1977theory}.
The CSS compatibility relation warrants a natural generalization of Definition \ref{def:classicalgroupcode} to quantum codes, which goes by the name of a group Gottesman-Kitaev-Preskill (group GKP) code \cite{Albert_2020}.

\begin{defn}[Group GKP code]\label{def:GKP}
    The sequence of subgroups $H \le K\le G^n$ defines a group GKP code.  
    The codespace
    is labeled by left cosets: 
    \begin{equation}
      \mathcal{C}=\mathrm{span}\left\lbrace \sum_{h\in H}|hk\rangle\right\rbrace _{k\in K}=\mathbb{C}[H\setminus K].
    \end{equation}
In other words, the codewords (basis states for the encoded subspace) are those cosets of \(H\) in \(G^n\) that are also cosets of \(H\) in \(K\).
\end{defn}

As Remark \ref{rmk:abeliancodelinear} suggests, thinking of classical and group GKP block codes over non-Abelian $G$ is primarily useful when the subgroup $K$ can be described in terms of local constraints, as is possible for Abelian $G$.
Unfortunately, 
the subgroups $K\le G^n$ for certain finite groups $G$ can be extraordinarily restrictive under reasonable assumptions. 
This was already noticed for classical codes \cite{forney2002hamming} and is manifest in our Lemma \ref{lem:equivalentthm}.
Another drawback is that group GKP codes do not capture many interesting group-qudit codes, as we will show precisely in Rmk.~\ref{rmk:weirddimension}.
As stated in the introduction, we would like our family of group CSS codes to include quantum-double codes, which have generalized stabilizers that are different from those of group GKP codes.
Therefore, we look for a different kind of group-based quantum code.

\begin{figure}[t]
    \centering
\begin{tikzpicture}[scale=0.8]

  \draw (-4,-3.0) rectangle (4,2.5);

  \def\r{2}
  \draw (-1.6,0) circle (\r);   
  \draw ( 1.6,0) circle (\r);   
  \draw (0,-0.8)  circle (\r);  

  \node at (-2.2,1) {$Z^{K_1} = 1$};
  \node at ( 2.2,1) {$Z^{K_2} = 1$};
  \node at (0,-2.1) {$Z^{K_3} = 1$};

  \node at (0,-0.1) {$\mathcal{C}_Z$};

  \node[anchor=west] at (-3.7,2.1) {$\mathbf{g} \in G^n$};

\end{tikzpicture}
\caption{
    Generalized \(Z\)-checks of group CSS codes define the admissible group elements \(\mathbf{g}\in G^n\) whose computational basis states \(|\mathbf{g}\rangle\) are used to construct the codewords.
    Here, the square region depicts all such labels, and the three circular regions of the Venn diagram denote the subsets satisfying the constraints \(Z^{K_i} = 1\) for each ``local'' subgroup \(K_{i} \leq G\).
    Their intersection forms the set \(\mathcal{C}_Z\) of admissible basis labels for the code.
    All regions reduce to subgroups of \(G^n\) for Abelian \(G\), but \textit{neither} the three circular regions \textit{nor} their intersection are subgroups for non-Abelian codes.
    }
    \label{fig:z-checks}
\end{figure}
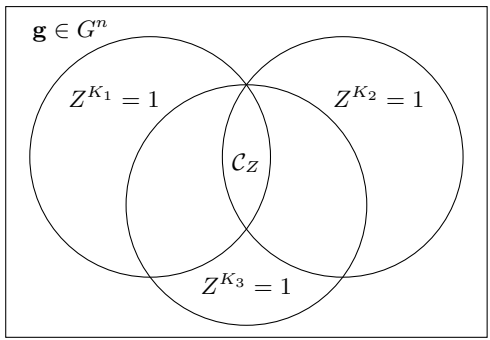

\subsection{Group CSS codes}

Our generalization of CSS codes to ``group CSS" codes is meant to encompass most existing group-valued codes, with the notable exception of twisted quantum doubles \cite{Hu_2013, CuiTwisted}, which will be the subject of a future work.
As with CSS codes, the construction comes with two types of ingredients, one for each generalized check set.
The \(Z\)-type checks are best thought of as a set of ``local'' conditions on allowed group-valued basis states \(|\mathbf{g}\rangle\).

\begin{defn}[Generalized $Z$-checks]\label{def:Zcheck}
    Let $K\le  G$ be a subgroup and $\mathbb{I}(\cdots)$ be the indicator function which is 1 if its argument is true and 0 if false.  The $Z$-checks of a group qudit CSS code are of the form  \begin{equation}
        Z^K_{q_1\cdots q_j}|\mathbf{g}\rangle := \mathbb{I}(g_{q_1}^{a_1}\cdots g_{q_j}^{a_j} \in K)|\mathbf{g}\rangle~, \label{eq:generalZcheck}
    \end{equation}
    where $q_1, \cdots, q_j\in\lbrace 1,\ldots,n\rbrace $ are qudits (possibly repeated) and $a_1,\ldots, a_j \in \lbrace \pm 1\rbrace $ are their \textit{orientations}.
We suppress the orientations in the definition on the left-hand side to avoid clutter.
The $Z$-check \textit{weight} 
is defined as the largest number qudits $j$, counting multiplicities, that show up in any $Z$-check. 
\end{defn}

There is a difference between the weight, as defined above, and the \textit{support} of a \(Z\)-check (i.e., the number of qudits the check acts on nontrivially).
For example, the check that tests whether \(g_1^3 g_2 \in K\) has a support of 2 (i.e., it is a two-qudit check) but has weight 4.
The weight relates the check to the concept of a group word in combinatorial group theory---an important feature of our formulation (see Sec.~\ref{sec:rigid}).
The support is useful for characterizing distance and is upper-bounded by the weight. 

We refer to \(K \leq G\) as a \textit{local subgroup} because it is intended to constrain an \(O(1)\)-fraction of the length-\(n\)  basis-state labels \(\mathbf{g}\), in spirit of the constraints of Abelian LDPC codes (see Remark~\ref{rmk:abeliancodelinear}).
The price we pay is that the set of \(\mathbf{g}\) satisfying a \(Z\)-check's constraint
does not generally form a subgroup of $G^{n}$. 
For instance, taking the trivial local group $K = 1$, the product of two group words $(g_1, \dots, g_p), (h_1, \dots, h_p) \in
G^n$ satisfying $g_1\dots g_p  = h_1\dots h_p = 1$ does not necessarily satisfy $(g_1h_1)\dots (g_ph_p) = 1$.
This property is recovered in the Abelian case, and the resulting subgroups of \(G^n\) can be used instead of the local subgroups \(K \leq G\) to define the \(Z\)-checks (see Appendix \ref{app:abelian}).

The set of \(Z\)-checks, each constructed using a (potentially different) local subgroup \(K\), defines the admissible set \(\mathcal{C}_Z\) of basis states \(|\mathbf g\rangle\) used to construct the code (see Fig.~\ref{fig:z-checks}).
This set is also not a subgroup of \(G^n\) in general, and it is only guaranteed to be a subgroup when \(G\) is Abelian.

Our \(X\)-checks will consist of the conventional group multiplication operators, which permute the group-valued basis states according to the group's multiplication table.
There are two types of \(X\) checks for non-Abelian groups because multiplication from the left and right do not always agree when $G$ is non-Abelian.

\begin{defn}[Generalized $X$-checks]\label{def:Xcheck}
    For a single $G$-valued qudit, define the unitary operators \begin{subequations}
        \begin{align}
            \overrightarrow{X}^h |g\rangle &:= |hg\rangle , \\
            \overleftarrow{X}^h |g\rangle &:= |gh^{-1}\rangle.
        \end{align}
         \end{subequations}
        We consider generalized $X$-type checks to be products of these operators: for $\mathbf{g},\mathbf{h}\in G^n$,  \begin{equation}
            X^{\mathbf{g},\mathbf{h}} := \prod_{j=1}^n \overrightarrow{X}^{g_j}_j\overleftarrow{X}^{h_j}_j. \label{eq:generalXcheck}
        \end{equation}
The $X$-check weight 
is identical to the check support, i.e., it is the largest number of distinct qudits 
that show up in any $X$-check.
\end{defn}

In general, left and right multipliers commute with each other, but not amongst themselves.
\begin{subequations} \label{eq:Xoperatormultiplication} \begin{align}
\overrightarrow{X}^{g}\overrightarrow{X}^{h}\left(\overrightarrow{X}^{g}\right)^{-1}\left(\overrightarrow{X}^{h}\right)^{-1}&=\overrightarrow{X}^{ghg^{-1}h^{-1}},\\\overleftarrow{X}^{g}\overleftarrow{X}^{h}\left(\overleftarrow{X}^{g}\right)^{-1}\left(\overleftarrow{X}^{h}\right)^{-1}&=\overleftarrow{X}^{ghg^{-1}h^{-1}},\\\left[\overrightarrow{X}^{g},\overleftarrow{X}^{h}\right]&=0~.
\end{align}\end{subequations}
We depict the non-commutativity of these operators in the first two lines above using group commutators because it shows how difficult it is to maintain the minimum weight of these checks for block codes made out of non-Abelian groups.
For example, two-qudit checks associated with group elements \((g,g)\) and \((g,h)\) are both weight two, but their group commutator yields a weight-one check associated with \((1,ghg^{-1}h^{-1})\).
Naively picking blocks of \(X\)-checks of some fixed minimal weight generically generates checks of lower weight.  
This observation was also made in the classical community by Forney \cite{forney2002hamming} and plays a central role in our main result, Thm.~\ref{thm:classify}.

Group CSS codes are defined using a set of \(Z\)- and \(X\)-checks, $\mathcal{S}_Z$ and $\mathcal{S}_X$, respectively.
The \(Z\)-checks form a set of constraints on admissible group words that are defined by their respective local subgroups \(K \leq G\), while the unitary \(X\)-checks generate the ``global'' subgroup,
\begin{align}\label{eq:x-check-group}
     \left\langle X^{\mathbf{g},\mathbf{h}}\in\mathcal{S}_{X}\right\rangle & \cong \diagX \leq G^n \times G^n &\text{(\ensuremath{X}-check group)}\,. 
\end{align}

\begin{defn}[Group CSS code]\label{def:CSS}
Let $\mathcal S  = \mathcal S_X \cup \mathcal S_Z$, where $\mathcal S_X$ contains only $X$-type checks (Def.~\ref{def:Xcheck}), and where $\mathcal S_Z$ contains only $Z$-type checks for one or more subgroups \(K \leq G\) (Def.~\ref{def:Zcheck}), and require that every check in $\mathcal S_X$ preserves the mutual \(+1\)-eigenvalue eigenspace of all the checks in $\mathcal S_Z$.
A \textit{group CSS code} is a joint \(+1\)-eigenvalue eigenspace of all elements of $\mathcal S$.
Its codespace is the subspace \begin{equation}
    \mathcal{C} := \lbrace |\psi\rangle \in \mathbb{C}[G^n] : S|\psi\rangle = |\psi\rangle \; \forall S\in \mathcal{S}_X\cup \mathcal{S}_Z\rbrace.
\end{equation}
The code \textit{log-dimension} $k$ (roughly, the number of logical $G$-qudits)\footnote{The log-dimension $k$ is only guaranteed to be an integer if $G=\mathbb{Z}_p^l$, with $p$ a prime number.
For \(\mathbb{Z}_q\)-qudits with composite \(q\) and for non-Abelian groups, the codespace can decompose into a tensor product of logical qudits of various dimensions; see Remark \ref{rmk:weirddimension}.
} is 
\begin{equation}\label{eq:rate_definition}
        k := \log_{|G|}\dim(\mathcal{C}).
\end{equation}
The code \(Z\)-distance and \(X\)-distance are defined analogously to qubit codes.
We say that the $X$-distance $d_X >m$ if for any $X$-type operator $\mathcal{O}_X$ of the form \eqref{eq:generalXcheck} acting on $m$ qudits, the following (Knill-Laflamme \cite{knill1997theory}) condition is obeyed: \begin{equation}\label{eq:KL}
        \langle \psi_\alpha |\mathcal{O}_X|\psi_\beta \rangle = c\cdot \delta_{\alpha\beta},
    \end{equation} 
    where \(|\psi_{\alpha,\beta}\rangle\) are logical states, and where the complex coefficient $c$ depends only on \(\mathcal O_X\).   
    The $Z$-distance $d_Z>m$ if an analogous statement holds for any $Z$-type operator $\mathcal{O}_Z$ of the form \eqref{eq:generalZcheck}.
\end{defn}

Group GKP codes and group CSS codes intersect, but neither is contained in the other.
Group GKP codes are defined using only one ``global'' group \(K\) for the \(Z\)-checks, and only one type of group multiplication operators for the \(X\)-checks.
Given a nested sequence of groups $H\le K\le G^n$, the corresponding group-GKP checks are
     \begin{equation}
\mathcal S_X = \left\lbrace \overrightarrow{X}^{\mathbf{h}}~|~\mathbf{h}\in H\right\rbrace\quad\quad\text{and}\quad\quad\mathcal{S}_{Z}=\left\lbrace \mathbb I (\hat{\vb g} \in K) | \vb g \in G^n\right\rbrace \,.
        \end{equation}
The \(Z\)-check here is of a different type than the ones we introduced in Def.~\ref{def:Zcheck} and merely tests whether a given \(\mathbf{g}\) is in the ``global'' subgroup \(K \leq G^n\), \(\mathbb{I}(\hat{\mathbf{g}}\in K)|\mathbf{g}\rangle = \mathbb{I}(\mathbf{g}\in K) |\mathbf{g}\rangle\).
Setting \(n=1\) guarantees this constraint can be expressed in terms of group CSS checks, yielding codes that are both group GKP and group CSS.

On one hand, there exist group GKP codes that are not group CSS codes while preserving the number of physical qudits because it is not always possible to express the ``global'' \(Z\)-check group as a set of local-subgroup \(Z\)-checks 
On the other hand, there exist group CSS codes that are not group GKP codes since both left- and right-multiplication operators are allowed to form the \(X\)-check group \(\diagX\).
Because of this constraint, generic $X$-logical operators of a group GKP code are permutations of cosets $H\setminus K$, and as such, they are naturally captured by unitary operators corresponding to right multiplication \cite{Albert_2020}.  
In contrast, group CSS codes such as quantum doubles can have non-invertible symmetries (Remark \ref{rmk:noninvertible}) when $G$ is non-Abelian, so they cannot be group GKP codes.  
Group CSS codes can also have different logical structure
(see Remark \ref{rmk:weirddimension}).

Specializing further, we see how qubit CSS codes fit into the picture above. 
Here, we can take the group $G=\mathbb{Z}_2=\lbrace 0,1\rbrace$.  The Pauli matrix $X$ corresponds to $\overrightarrow{X}^1 = \overleftarrow{X}^1 $, so $X$-type checks such as $X_1X_2X_3X_4$ are a simple and special case of Definition \ref{def:Xcheck}.   As for $Z$-type checks, notice that we can re-write \begin{equation}
    \frac{1+Z_1Z_2Z_3Z_4}{2} |g_1g_2g_3g_4\rangle = \mathbb{I}(g_1+g_2+g_3+g_4 = 0)|g_1g_2g_3g_4\rangle. 
\end{equation}
Here we have written + for the group operation to avoid any potential confusion,
since $\mathbb{Z}_2$ is a ring as well as a group.  In other words, the
$Z$-checks of all CSS codes can precisely be understood as fixing the group word
$|\mathbf{g}\rangle$ to lie in a specific subgroup of $\mathbb{Z}_2^n$, just as
in Definition \ref{def:Zcheck} (see also  Remark \ref{rmk:abeliancodelinear}).  
More generally, group CSS codes over cyclic $G$ are equivalent to ordinary CSS codes over the same group, and Abelian CSS codes are equivalent to group CSS codes up to a local isometry (see Appendix \ref{app:abelian}).

\subsection{The double coset formalism}
\label{sec:doublecoset}
The codespace of CSS codes can be understood in terms of equivalence classes of qubit basis states that satisfy \(Z\)-check constraints and that are related to each other by $X$-checks.
Similarly, group-GKP codewords are uniform superpositions of elements in \(K\) that are related to each other by elements of \(H\).
The group CSS construction generalizes this convenient property, allowing for a similar analysis of the codespace and resulting bounds on the weight of logical operators. 
In the group GKP case, code states are conveniently labeled by
  cosets in $H \backslash K$, or equivalently, by a subset of cosets in the right-coset space
  $H \backslash G$. 
Treating the underlying group space itself as a coset space, we can express group CSS codewords in terms of equivalence classes of \(\mathbf g\) that are related by \(X\)-checks and that satisfy the constraints of the \(Z\)-checks.
These will form a double coset space, built on top of the underlying physical coset space.
This allows us to treat left and right \(X\)-check multiplication on the same footing.

\begin{defn}[Doubled Hilbert space]\label{def:doublecoset}
  We embed $\mathbb{C}[G^n]$ into a subspace of $\mathbb{C}[G^n]\otimes \mathbb{C}[G^n]$ as  follows: 
  \begin{equation}
|\mathbf{g}\rangle\quad\rightarrow\quad|\mathbf{g}\diaggroup\rangle\equiv \sum_{\mathbf{h}\in G^{n}}|\mathbf{g}\mathbf{h},\mathbf{h}\rangle\,,
    \label{eq:doublecosetstates}
  \end{equation}
where the \textit{diagonal subgroup} is
  $    \diaggroup \equiv \langle  (\vb g, \vb g) \rangle_{\vb g
        \in G^n}$.
The physical Hilbert space is related by projection onto the coset
space 
\begin{equation}
G^n \times G^n
    \to (G^n \times G^n)/ \diaggroup \cong G^n~,
\end{equation}
(note that they are isomorphic as sets, not groups) and it is spanned by the coset states \(|\mathbf{g}\diaggroup\rangle\) for all \(\mathbf{g} \in G^n\).

\end{defn}

We can define the $X$-checks [Eq.~\eqref{eq:generalXcheck}] to act on this enlarged space solely from one side, ``turning around'' the left-pointing arrow in the definition of \(\overleftarrow{X}\):
    \begin{equation}\label{eq:doubledXidentification}
        X^{\mathbf{g},\mathbf{l}}|\mathbf{k},\mathbf{h}\rangle=|\mathbf{g}\mathbf{k},\mathbf{l}\mathbf{h}\rangle\,,
    \end{equation}
for any \(\mathbf{g},\mathbf{h},\mathbf{k},\mathbf{l}\in G^n\).
This reduces to the previously defined action of these \(X\)-checks when they act on the coset states,
\begin{equation}
X^{\mathbf{g},\mathbf{l}}|\mathbf{k}\diaggroup\rangle=\sum_{\mathbf{h}\in G^{n}}|\mathbf{g}\mathbf{k}\mathbf{h},\mathbf{l}\mathbf{h}\rangle=\sum_{\mathbf{h}\in G^{n}}|\mathbf{g}\mathbf{k}\mathbf{l}^{-1}\mathbf{h},\mathbf{h}\rangle=|\mathbf{g}\mathbf{k}\mathbf{l}^{-1}\diaggroup\rangle\,.
\end{equation}
Above, the change \(\mathbf{h} \to \mathbf{l}^{-1}\mathbf{h}\) merely re-shuffles the terms in the sum.

The group \(\diagX\) of \(X\)-checks
relates physical basis states to each other.
The set of all basis states that are related to a fixed basis state \(|\mathbf{g} \diaggroup\rangle\) is called the orbit of \(|\mathbf{g} \diaggroup\rangle\).

\begin{defn}[Orbit of $X$-checks]\label{def:orbit}
    Basis states $|\mathbf{g}\diaggroup\rangle$ and $|\mathbf{g}^\prime\diaggroup\rangle$ are in the same orbit of the $X$-checks if $\sigma|\mathbf{g}\diaggroup\rangle = |\mathbf{g}^\prime\diaggroup\rangle$ for some $\sigma\in \diagX$. 
    This implies an equivalence relation, \(\mathbf{g} \sim \mathbf{g}^\prime\), for any two elements in the same orbit.
\end{defn}

Group CSS codes can then be constructed by first defining an admissible subset of physical basis states \(\mathcal{C}_Z\) (determined by the \(Z\)-checks), partitioning this subset into orbits under the \(X\)-check group, and uniformly superposing each orbit into its corresponding codeword.
Since the physical states themselves form a coset space, the codespace is the \textit{double} coset space \(\diagX \backslash G^{2n} / \diaggroup\).

Admissible physical basis states satisfy the \(Z\)-check conditions, which can be defined on the physical coset space as follows.
Each coset state \(|\mathbf{g} \diaggroup\rangle\) is a sum of elements \(|\mathbf{k} = \mathbf{g}\mathbf{h},\mathbf{l} = \mathbf{h}\rangle\), with the physical basis label obtainable from the relation \(\mathbf{k}\mathbf{l}^{-1}=\mathbf{g}\).
Its coordinates have to satisfy the \(Z\)-check relations in order to make it to the admissible set.
A basis for the codespace can then be written as
\begin{equation}\label{eq:basis-codewords}
    |[\mathbf{g}]\rangle\equiv \sum_{\sigma\in \diagX }\sigma|\mathbf{g}\diaggroup\rangle\qquad\qquad\text{(group CSS codewords)}
\end{equation}
for reach representative $\mathbf{g}$ of orbits 
\([\mathbf{g}] \in H \backslash \mathcal{C}_Z\).
Presenting this rigorously, we have the following proposition.

\begin{prop}
  \label{prop:codespace}
    The codespace $\mathcal{C}$ of a group qudit CSS code is
    $\mathcal{C}=\mathbb{C}[\mathcal{K}]$, where $\mathcal{K} \subseteq  \diagX  \backslash G^{2n}/   G^n_{\mathrm{diag}}$ 
    is the subset of double cosets such that any representative $(\mathbf{k},\mathbf{l})$ satisfies 
    $\prod_{i=1}^{j}(\mathbf{k}\mathbf{l}^{-1})_{q_{i}}^{a_{i}}\in K$ for each $Z$-check $Z^K_{q_1,
      \dots, q_j}$.
\end{prop}

\begin{proof}
    This result will follow from finding a basis for
    $\mathcal C_X$ (the mutual $+1$-eigenvalue eigenspace of the operators in
    $\mathcal S_X$) which is labeled by
    elements of the double coset space $ \diagX \setminus G^{2n}/ \diaggroup$.
    Suppose that we have a generic state
    \begin{equation}
    |\psi\rangle=\sum_{\mathbf{g}\in G^{n}}\psi_{\mathbf{g}}|\mathbf{g}\diaggroup\rangle\in\mathcal{C}_{X}~,
    \end{equation}
    where $\psi_{\mathbf{g}}$ are some complex coefficients.
    Any $X$-check $\sigma \in H$ merely permutes the basis states.
    Since this state is stabilized by the \(X\)-checks, 
    the coefficients corresponding to the orbit of \(\mathbf{g}\) are equal:
    \begin{equation}
    0=\bra{\vb g\diaggroup}\sigma-1\ket{\psi}=\psi_{\sigma^{-1}(\vb g)}-\psi_{\vb g},
    \end{equation}
    i.e., $\psi_{\vb g} = \psi_{\vb h}$ if $\vb g \sim \vb h$. 
    A basis for this space is the set of indicator functions
    on double-cosets in $ \diagX  \backslash G^{2n} / \diaggroup$. Since the $X$-checks preserve the codespace of the $Z$-checks $\mathcal C_Z$ (Def.~\ref{def:CSS}), the orbit states [Eq.~\eqref{eq:basis-codewords}] span the codespace $\mathcal C$ as claimed.
\end{proof}

\section{Quantum doubles on 2D CW complexes} \label{sec:QD}
Quantum double (QD) models are largely studied on cellulations of 2D orientable surfaces in the many-body and quantum-information communities.
Even their generalizations, including 3D quantum triple models~\cite{Moradi_2015, Delcamp_2017} and higher-dimensional non-twisted Dijkgraaf-Witten theories~\cite{dijkgraaf1990topological}, utilize only the 2D ``skeletons'' of their underlying higher-dimensional manifolds.
However, it has long been known~\cite{huebschmann1999extended,oeckl2003generalized} that quantum doubles can, in principle, be constructed on a more general class of non-manifolds called 2D CW complexes.
Such complexes go beyond 2D manifolds by allowing each edge to border arbitrarily many faces, thereby encompassing many interesting non-planar structures such as presentation complexes (see Fig.~\ref{fig:cw_example}).
Despite this greater generality, a thorough characterization of quantum-double properties (ground-state degeneracy, logical operators, boundary effects, etc.) on general \textit{non-manifold} complexes appears missing.

CW complexes can be understood as graphs with 2D plaquettes continuously glued to loops in the graph. 
While Abelian quantum doubles can readily be associated to higher-degree
complexes---as in the $(2,2)$ 4d toric code \cite{4dtoriccode}---the non-Abelian nature of the gauge group
$G$ forbids a similar association between higher-dimensional cells and the
checks of a non-Abelian double. 
We will see this in two ways. First, for a generic group $G$, every qudit
can only be involved in two $X$-checks---one that acts by left multiplication
and one that acts by right multiplication---just as an edge is connected to at
most two vertices.
Next, while the fundamental
group $\pi_1(X)$ of a topological space $X$ can be non-Abelian, the higher
homotopy groups $\pi_{k}(X)$ for $k > 1$ are always Abelian, so we cannot associate a non-Abelian ``holonomy'' to
higher-dimensional cells in the same manner.\footnote{The framework of
  higher-form gauge theory provides a \emph{partial} way around this---c.f. Refs.~\cite{Bullivant_2017, Pfeiffer_2003, Baez_2010}. We will comment on this further in the outlook.} 
  We
make this correspondence precise with our main classification result in
Thm.~\ref{thm:classify}, where we show that, without leveraging the group structure of $G$, such as its normal subgroups, all group-CSS codes on $G$ (under modest assumptions) are
quantum doubles on a 2D CW complex.

In this section, we illustrate the association of a 2D CW complex to a quantum
double. We show that
code parameters can be improved over standard quantum doubles by looking for
codes on complexes which are not manifolds. 
We introduce a new
treatment of boundaries and holes through an object called a ``ghost
vertex'' (following the terminology of Ref.~\cite{CuiTwisted}), which 
provides a conceptually simple understanding of the code space and the CW
analogues of the logical ``ribbon'' operators.
In particular, the codespace of a quantum double will be shown to  depend only on the topology
of the complex, straightforwardly generalizing the case for a 2-manifold \cite{Cui:2019lvb, cuinotes}, including those with boundary \cite{QD_boundary, Cong1, Cong2, Cong3}.  
The picture of ghost vertices is, surprisingly, extremely helpful for thinking about the connections between phases of matter and quantum double codes (see Section~\ref{sec:matter}).

\subsection{Construction}
\label{sec:construction}
A 2D CW complex is not locally 2D in the same sense as a manifold; rather, it is
an arbitrary graph with 2D faces, or plaquettes, glued continuously to loops in
the graph. From a physicist's perspective, we can intuitively think that a 2D CW complex is
more general than a 2D manifold in a key way: it can have 
junctions (see the left panel of Fig.~\ref{fig:cw_example} for an example) where the complex is not locally Euclidean. This is made precise with the following definition.

\begin{defn}[CW complex \cite{hatcher}]
   Define a $k$-cell $e^k$ to be a topological space homeomorphic to the open
   $k$-disk (a.k.a \(k\)-ball). A 2-dimensional CW complex $X$ is the union of spaces $X_0 \subset X_1 \subset X_2$
, where $X_k = X_{k-1} \cup \{e^k_\alpha\} / \sim$, the attaching map $\varphi_{\alpha}^k:\partial e^k_\alpha \to X_{k-1}$ is continuous, and we define the equivalence relation $x \sim \varphi_\alpha^k(x)$ for each $x\in X_{k-1}$. 
Choose an orientation for each cell; this choice is arbitrary and will not meaningfully affect our definitions.
\end{defn}

The CW construction is a straightforward extension of the standard construction due to Ref.~\cite{Kitaev_2003}. 
We utilize the double-coset formalism from Sec.~\ref{sec:doublecoset} to define these codes.
In that formalism, we associate the \(X\)-check \(X^{\mathbf{k},\mathbf{h}}=\overrightarrow{X}^{\mathbf{k}}\overleftarrow{X}^{\mathbf{h}}\) with a \(2n\)-dimensional vector of group elements 
\(\mathbf g = (\mathbf{k},\mathbf{h}) \in G^{2n}\) such that \(g_{i} = k_i\) and \(g_{i+n} = h_i\) for \(i \in \{1,\cdots ,n\}\).
That way, each group-qudit corresponds to two coordinates of such a vector, one for each type of multiplier.

\begin{defn}[Smooth quantum double code]\label{def:generalizedQDsmooth}
Consider a group CSS code where the $X$-checks take the form
\begin{equation}
X_{B_\alpha}^g \equiv \prod_{q \in B_\alpha}X^g_q
\end{equation}
in the double coset picture (Def.~\ref{def:doublecoset}) for all $g \in G$ and for subsets $B_\alpha$ which partition $\{1, \cdots, 2n\}$. 
Further suppose the $Z$-checks take the form $Z^{\{1\}}_{q_0, \dots, q_k}$, where $q_{i}, q_{i+1 \ (\text{mod } k)}$ are both acted upon by some $X$-check $X^g_{B_{\alpha_i}}$. The orientation $a_i$ (Def.~\ref{def:Zcheck}) is $-1$ if $q_{i} \in B_{\alpha_i}$ and $+1$ if $q_{i}+n \in B_{\alpha_i}$.
\end{defn}

The motivation for the construction above is that it constitutes---in a way we will later make precise---the most general relationships such that the corresponding double-like ground space projectors
\begin{align}\label{eq:QDstabilizers}
A_{B_{\alpha}} &\equiv \frac{1}{|G|}\sum_{g \in G}X^g_{B_\alpha} & B_{q_0, \dots, q_k} &\equiv Z^{\{1\}}_{q_0, \dots, q_k}
\end{align}
commute when $G$ is non-Abelian. The first condition in Def.~\ref{def:generalizedQDsmooth} ensures that each qudit belongs to two $X$-checks; one that acts by right-translation, and one that acts by left-translation. When $G$ is non-Abelian, it follows from Eq.~\eqref{eq:Xoperatormultiplication} that this bivalence is the only way to write down a set of commuting projectors $\{A_{\alpha}\}$. 
We will see that this is ultimately the obstruction to finding good codes from such models (Cor.~\ref{cor:dZ_simple_covariant}). The second relationship ensures that the $X$-checks modify the $Z$-checks trivially: If $q_i \in B_\alpha$, then
\begin{equation}
X^h_{B_\alpha}: g_{q_0} \dots g_{q_{i-1}}g_{q_i} \dots g_{q_k} \mapsto g_{q_0} \dots (g_{q_{i-1}}h^{-1})(hg_{q_i}) \dots g_{q_k}.
\end{equation}
This is proven in the next proposition.

\begin{figure}[t]
    \centering
    \def\svgwidth{0.3\linewidth}
    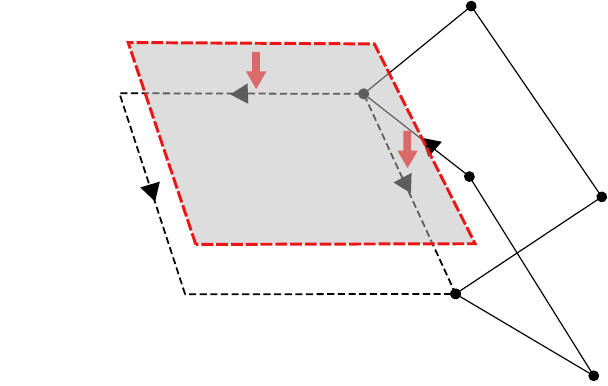
    \hspace{2cm}
    \def\svgwidth{0.3\linewidth}
    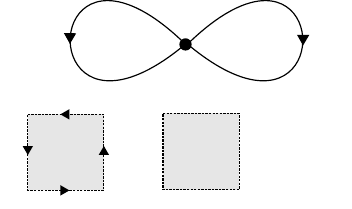
    \caption{Figure illustrating quantum double codes on 2D CW complexes.
      Zero-cells are represented by dots, 1-cells by oriented edges, and 2-cells
      are shown shaded in gray. \textbf{Left}: The gluing map $\varphi^2:\partial e^2 \to \Sigma^1$ is defined by mapping the red dotted line continuously to the black dotted line. Group-valued degrees of freedom are attached to the edges of the graph, with the blue kets denoting the group words associated to the edges in the dotted loop.
    To associate a QD code to this complex, the two-cell is associated to a $Z$-check, which enforces the constraint $ab^{-1}c^{-1}d = 1$.
    An $X$-check is associated to each vertex acting on all adjacent qudits.
    \textbf{Right:} A 2D CW complex $\Sigma$ obtained by gluing 2-cells to the wedge sum
    of two circles labeled by $r, s$ as indicated by the orientations has the
    fundamental group $\pi_1(\Sigma) = \mathrm D_8 = \langle  r,s \ | \ r^4 = 1, s^2=1, srs = r^{-1}
    \rangle$. This is called a \textit{presentation complex}, where the circles in the
    wedge sum are generators and the 2-cells are relations. This illustrates a
    general procedure for engineering a CW complex $\Sigma$ with $\pi_1(\Sigma)
    = G$ for any group $G$, which controls the codespace, as implied by
    Eq.~\eqref{eq:codespace_rough}. These complexes are often not manifolds.
    }
    \label{fig:cw_example}
\end{figure}

\begin{prop}
  \label{prop:QDonCW}
  Every smooth quantum double code defines a 2D CW complex $\Sigma = (V, E, F)$ and vice-versa.
  In other words, the vertex set is $V = \{B_\alpha\}_{\alpha}$ with \(X\)-checks
  \begin{equation}\label{eq:Xg_shorthand}
      X_{v}^{g}:=X_{l_{1},\dots,l_{p},r_{1},\dots,r_{q}}^{g}:=\overrightarrow{X}_{l_{1}}^{g}\cdots\overrightarrow{X}_{l_{p}}^{g}\overleftarrow{X}_{r_{1}}^{g}\cdots\overleftarrow{X}_{r_{q}}^{g}~,
\end{equation}
    the edge set $E$ is the set of physical qubits, and each $Z$-check $Z^{\{1\}}_{q_0, \dots, q_k}$ defines a plaquette $p \in F$ with $\partial p = \{q_0, \dots, q_k\}$. In particular, this association guarantees that the stabilizers in Eq.~\eqref{eq:QDstabilizers} commute.
\end{prop}

\begin{proof}
  For each edge $q_i$, we define a gluing map $\varphi_q^1:\partial q \to V$ via $\varphi_q^1(0) = B_\alpha$ if $q_{i}+n \in B_{\alpha}$ and $\varphi_q^1(1) = B_\beta$ if $q_i \in B_\beta$. Note that a single qudit may be connected twice to the same vertex. For each $Z$-check $Z^{\{1\}}_{q_0, \ldots, q_k}$ associated to a plaquette $p$ with $\partial p =\{q_0, \dots, q_k\}$, since $q_i$ and $q_{i+1 \text{ (mod $k$)}}$ appear in the same $X$-check, the path $\gamma = q_0 \to q_1 \to \ldots \to q_k \to q_0$ defines a continuous loop in the graph $V \cup E$; therefore it corresponds to a continuous map $\varphi_p^2:\partial p \to V \cup E$. This shows that a smooth QD code defines a 2D CW complex. 
  
  These associations are also reversible; given a 2D CW complex $\Sigma = (V, E, F)$ with oriented cells, we can associate a qubit to each edge $q_i \in E$. To each vertex $v$, let $q_i \in B_v^L$ if $q_i$ is an outgoing edge from $v$, and $q_{i+n} \in B_v^R$ if $q_i$ is an incoming edge of $v$. Then $B_{v} \equiv B_v^L \sqcup B_v^R$ defines an $X$-check (in the double coset picture). Lastly, to each plaquette $p$ with $\partial p = \{q_0, \ldots, q_k\}$, we associate a $Z$-check $Z^{\{1\}}_{q_0, \dots, q_k}$, where the $i^\text{th}$ orientation $a_i$ is $\pm 1$ if the orientation given to $q_i$ by the vertices is identical (opposite) to the orientation of $q_i$ in $\partial p$. Since $\partial p$ is a continuous loop, $q_i$ and $q_{i+1}$ always belong to a shared $X$-check $X_{B_v}^g$. If $q_i$ is an outgoing edge of $v$, then $q_i \in B_v^L$, in which case $\partial p$ assigns the opposite orientation to $q_i$ and $a_i = -1$, with the reverse if $q_{i+i} \in B_v^R$. 
  This shows that the constructed code is a smooth QD code (Def.~\ref{def:generalizedQDsmooth}).

  Now we show that the projectors $A_{B_v}, B_{q_0, \ldots, q_k}$ [Eq.~\eqref{eq:QDstabilizers}] commute. The $Z$-checks commute since they are diagonal in the group basis. An edge is connected to at most two vertices, and since the vertices assign the edge opposite orientations, each qudit has exactly one $X$-check with left or right multiplication respectively, which commute. 
  Finally, suppose that a $Z$-check $Z_{q_1, \dots, q_k}$ overlaps
  with an $X$-check $X_{B_v}$.
  Because the edges $q_1, \dots, q_k$ form a closed loop, the checks must overlap
  on pairs of edges $(q_i, q_{i+1\text{ (mod $k$)}})$. Assume first that $i \neq k$ and $q_i,
  q_{i+1}$ are both positively oriented with respect to the loop. Then $A_{B_v}$
  acts on these edges by sums of $\overrightarrow X_{q_i}^g \overleftarrow
  X_{q_{i+1}}^g \ket{q_i,q_{i+1}} = \ket{q_i g,g^{-1}q_{i+1}}$,
  but this leaves the product $q_1\dots q_k$ invariant. Since changing the orientation
  of either edge changes the orientation with respect to $\partial p$,
  this argument holds regardless of the orientations of the edges. This
  leaves the case where $i = k$. In this case, $\overleftarrow X_{q_i}^g
  \overrightarrow X_{q_{i+1}}^g\ket{q_1, \dots, q_{k}} = \ket{gq_1, \dots,
    q_{k}g^{-1}}$, which changes the product from $q_1 \dots q_k \mapsto gq_1
  \dots q_kg^{-1}$. Since $gq_1 \dots q_kg^{-1} = 1$ iff $q_1\dots q_k = 1$ for
  any $g \in G$, this shows that the $Z$-checks and $X$-checks commute.
\end{proof}

The proof above has a very simple gauge-theoretic interpretation; the $Z$-checks enforce a trivial holonomy around homotopically trivial loops, and the $X$-checks leave these holonomies invariant. Throughout the rest of the paper, we will find it useful to employ the language of gauge theory, which has been successful for describing the properties of error correcting codes \cite{ungauging, Yoshida_2015, Cui:2019lvb}, as summarized in the following definition.

\begin{defn}[Gauge theory terminology]\label{def:gauge-theory}
    The configuration $\mathbf{g}\in G^n$ is a
    \textit{$G$-valued connection} on $\Sigma$. The $X$-checks generate
    \textit{gauge transformations} on the basis states:   two states $\ket{\vb g}, \ket{\vb g'}$ will be
    called \textit{gauge equivalent} if $\ket{\vb g} \sim \ket{\vb g'}$
    [Def.~\ref{def:orbit}], and the equivalence class $[\vb g]$ will be called the \textit{gauge class} of $\vb g$. 
    Identifying each gauge class with a representative or a subset of representatives will be called \textit{gauge-fixing}, and symmetrizing over the gauge transformations at vertex
    \(v\) will be called \textit{gauging} the $G$-symmetry at \(v\). The product of all $X$-checks $\prod_v X_v^g$ will be called a \textit{large gauge transformation}.
    Given a closed loop $\gamma = (q_1, \dots, q_m)$, 
    the group word $g_{q_1}\cdots g_{q_m}$ is referred to as the \textit{holonomy} of the
    connection around $\gamma$, and will be denoted $\vb g_\gamma$.  The connection is called \textit{flat} if the holonomy
    around each $\partial p$ is the identity in $G$.
\end{defn}

\begin{rmk}
\label{rmk:oriented_surface_duality}
Note that when the CW complex is a cellulation of a 2-manifold there is a notion of a dual lattice, and
exchanging the $X$-checks and $Z$-checks maps a smooth QD code to a corresponding code on the dual lattice. 
However, for a general CW complex the $X$-checks and $Z$-checks play fundamentally different roles; each 1-cell is connected to at most two 0-cells, but a 1-cell can belong to arbitrarily many 2-cells. 
In this case, the duality transformation no longer produces a valid group CSS code.
\end{rmk}

Note that the construction above specifies a set of commuting 
projectors
for any group $G$.
We will show that this property is unique to quantum doubles on 2D CW
complexes: to design any other non-Abelian group-CSS code, we must know
something about the structure of $G$. We will prove this result in Thm.~\ref{thm:classify}.

\subsection{Ghost vertices}
We can push the association between quantum doubles and CW complexes to further
encompass \textit{rough boundaries}---locations where $X$-checks are not measured---and generalizations thereof.
Different types of boundaries are important for the practical usage of the toric code because planar geometries are more compatible with 2D architectures.
Such boundaries have been classified for general 2D quantum doubles~\cite{QD_boundary, Cong1, Cong2, Cong3}.
In this section, we introduce the framework of ``ghost vertices'' (following
Ref.~\cite{CuiTwisted}) to handle quantum doubles with boundaries.
As an immediate consequence, our framework implies that boundaries and holes of
quantum doubles admit transversal logical gates at the boundary, as observed in
Refs.~\cite{Ellison_2026, huang2026hybridlatticesurgerynonclifford}.
We will also use our
framework to extend the connection between SPTs and error correcting codes \cite{LDPCphysics} to non-Abelian global symmetries in the case of
the non-Abelian cluster state introduced in Ref.~\cite{Fechisin_2025}.

\begin{defn}[Ghost vertices and rough quantum doubles] \label{def:roughQD}
A quantum double associated to a CW complex will be called \textit{rough} if
some of the $X$-checks $A_v$ are removed from the stabilizer group, i.e., some
collection of vertices $\{v_1, \dots, v_k\}$ are not associated to $X$-checks.
The removed vertices are named \textit{ghost vertices}.
\end{defn}

While there is no longer a unique CW complex associated to a rough quantum double because ghost vertices can be arbitrarily identified together, we can specify a minimal identification which canonically associates a CW complex and a rough quantum double.

\begin{prop}
  A rough quantum double can be associated to a unique CW complex with a set of
  ghost vertices, and a CW complex with specified ghost vertices can be uniquely
  associated to a rough quantum double.
  \label{prop:roughQDonCW}
\end{prop}

\begin{proof}
First, the map from CW complexes to smooth quantum doubles in
Prop.~\ref{prop:QDonCW} extends to a map from CW complexes with a
specified set of ghost vertices $Y$ to rough doubles by removing every check
corresponding to a ghost vertex.
In the other direction, given a rough quantum double, for each qudit which is
not left- (right-) multiplied by any check, we add an additional check to the
stabilizer set which acts only on this qudit by left- (right-) multiplication. We
use this new set of $X$-checks to construct the 1-skeleton of the complex
$\Sigma$ as done in Prop.~\ref{prop:QDonCW}, and the vertices
associated to the new $X$-checks the set of ghost vertices $Y$. Now the corresponding
gluing rules for a 2-cell $p$ may not be continuous. By construction, this
discontinuity can only occur if $q_i, q_{i+1} \in \partial p$ are both connected
to ghost vertices $v_i, v_{i+1} \in Y$. Therefore we identify $v_i$ with $v_{i+1}$
and repeat for each discontinuity in each 2-cell, as depicted in Fig.~\ref{fig:identify_ghost_cells}. Since the identification is
independent of the order, this clearly produces a unique CW complex associated
with the rough double having the maximum number of ghost vertices.
\end{proof}

\begin{figure}
\centering
\hspace{1.5cm}
\def\svgwidth{.4\textwidth}
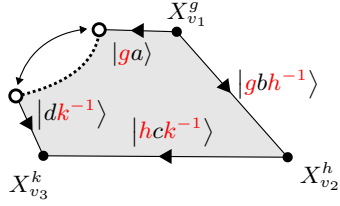
\caption{
  Associating a rough double to a CW complex. The open circles show ghost checks.
 The $X$-checks associated to each vertex are shown by their action on the
  qudits $\ket{a}, \ket{b}, \ket{c}, \ket{d}$ in red.
  The gray shaded region shows a $Z$-check $\mathbb I(a^{-1}bcd^{-1} = 1)$
  with a discontinuity
  illustrated by the dotted line. Once ghost checks are identified to a ghost vertex as shown by
  the double-headed arrow, the gluing rule for the $Z$-check becomes continuous. 
 }
 \label{fig:identify_ghost_cells}
\end{figure}

As we will see shortly in Prop.~\ref{prop:codespace_rough}, we can always
understand the codespace of a rough quantum double by gluing all ghost vertices
together.
However, it is often important to use the minimal identification in the previous proposition.
For example, in Sec.~\ref{sec:matter}, we will see that each ghost vertex
separately hosts a $G$-symmetry corresponding to a transversal logical operator
at each rough boundary of a quantum double \cite{huang2024fermionicquantumcriticalitylens, Ellison_2026}. 

\subsection{Codespace}

Ghost vertices are sufficient to describe the codespace of quantum doubles on CW complexes \(\Sigma\). 
In a nutshell, Proposition \ref{prop:codespace_rough} will show that the codespace contains two pieces:  one dependent on the fundamental group of $\Sigma$, and one dependent on the number of ghost vertices.   One elegant way to understand this codespace, which we will see later in Proposition \ref{prop:Zlogical}, is that $Z$-logical operators correspond to either the holonomy of loops around non-contractible cycles, or to the integral of the connection along paths that stretch between two different ghost 0-cells.   This is our way of formalizing and generalizing the standard observation that there are two kinds of logical operators on quantum double models, one corresponding to non-contractible cycles on a closed surface, and the other associated to lines stretching between defects.

\begin{prop}
The codespace \(\mathcal C\) of a rough quantum double code on a connected CW complex $\Sigma$ with ghost vertex set $Y$ is
\begin{equation}\label{eq:codespace_rough}
    \mathcal{C} = \mathbb{C}[\mathrm{Hom}(\pi_1(\Sigma/Y),G)] = \mathbb{C}[\mathrm{Hom}(\pi_1(\Sigma),G)] \otimes \mathbb{C}[G^{|Y|-1}]~,
\end{equation} 
where $\pi_1(\Sigma/Y)$ is the fundamental group of $\Sigma$ with all ghost vertices identified as a single point, where \(\Hom(\pi_1,G)\) is the set of all homomorphisms from \(\pi_1\) into \(G\), and where \(\mathbb{C}[\Hom(\pi_1,G)]\) is the vector space of formal linear combinations of such maps.
\label{prop:codespace_rough}
\end{prop}

The proof straightforwardly generalizes the proof found in Ref.~\cite{Cui:2019lvb} to include CW complexes and rough quantum doubles.
It will be useful to recall the gauge-theory terminology from Def.~\ref{def:gauge-theory} and periodically refer to Fig.~\ref{fig:codespace_ghost} while following this proof.

\begin{proof}
We will first identify a set of edges which can always be gauge-fixed to the identity.
Let $T$ be a spanning tree rooted at one of the ghost vertices $v_1 \in Y$. 
Because $T$ is a tree, we can decompose the vertex set as a disjoint union $V(T) =
\bigsqcup_{i=1}^N V(T_i)$, where $T_i$ is a tree rooted at $v_i$ for each $v_i
\in Y$ (an example is shown in Fig.~\ref{fig:codespace_ghost}).
As in Prop.~\ref{prop:codespace}, the mutual $+1$-eigenspace of all \(X\)-checks $\mathcal C_X$ is spanned by the states $\ket{[\vb g]}$,
defined as even superpositions of $\ket{\vb h}$ for all $\vb h$ which are
gauge-equivalent to $\vb g$. 
The vertex operators act transitively on the edges in each $T_i$, so each gauge
class has a distinguished representative such that each edge of $T_i$ is gauge-fixed to 1 for
each $i$. Moreover, since each edge connected to the ghost vertices are only acted
upon by either right or left multiplication by the $X$-checks, this representative is \textit{unique}.

Next, we define the CW complex 
\begin{equation}\label{eq:quotient}
    \Sigma' \equiv
\Sigma/Y\qquad\qquad\text{(quotient complex)}    
\end{equation}
arising from identifying all ghost vertices.
The entire ghost-vertex set \(Y\) becomes a single vertex in this new quotient complex.
Let $\Sigma'^1$ denote the 1-skeleton of $\Sigma'$, and let $\pi_1(\Sigma'^1, Y)$ be the fundamental group of $\Sigma'^1$ with basepoint $Y$.
Since $T' = \bigsqcup_{i=1}T_i$ 
is a spanning tree of $\Sigma'$ rooted at $Y$, by Van Kampen's theorem \cite{hatcher}, the fundamental group of $\Sigma'^1$
is free and generated by the unique loops passing through the edges of $\Sigma'^1 - T'$.
Therefore, each gauge class $[\vb g]$ can be associated to a map
$\Phi_{[\vb g]} \in \Hom(\pi_1(\Sigma'^1, Y), G)$ in the following way.

Let $q_1,
\dots, q_N$ be the edges in $\Sigma' - T'$. Each edge $q_i$ corresponds to a
loop $\gamma_i$ which passes through $q_i$ and is otherwise contained in $T'$, 
which together generate $\pi_1(\Sigma'^1, Y)$. 
The unique representative $\vb h$
of $[\vb g]$ with all edges in $T'$ gauge-fixed to $1$
specifies an assignment $h_1, \dots, h_N$ to the edges $q_1, \dots, q_N$, so we put
$\Phi_{\vb h}(\gamma_i) = h_i$. This is clearly a homomorphism.
Furthermore, since $h_i$ is the holonomy of $\vb h$ around $\gamma_i$, and since the
application of $X$-checks leaves all holonomies based at $Y$ 
unchanged, $\Phi_{\vb h} = \Phi_{\vb g}$, so $\Phi_{[\vb g]}$ is well-defined
for a gauge-class $[\vb g]$. 

Now, we need to introduce the 2-cells. Since $X^g_{v}$ and $Z_{\partial p}^{\{1\}}$ commute,
the code space $\mathcal C$ is spanned by $\ket{[\vb g]}$ such that $Z_{\partial p}\ket{\vb
  g} = \ket{\vb g}$. Again by
Van Kampen's theorem, the
inclusion $i: \Sigma'^1 \to \Sigma'$ induces a surjective map $i_\ast:
\pi_1(\Sigma'^1, Y) \to \pi_1(\Sigma', Y)$ with kernel $[\varphi_{p}^2]$, where
$\varphi_{p}^2:S^1 \to \Sigma'^1$ is the gluing map associated with face \(p\). Therefore
\begin{equation}
  \Hom(\pi_1(\Sigma', Y),G) = \{\theta \in \Hom(\pi_1(\Sigma'^1, Y), G): \theta([\varphi_p]) = 1 \ \forall  \ p\}~.
\end{equation}
Since $\ket{[\mathbf g]} \in \mathcal C$ if and only if $\Phi_{\mathbf g}([ \varphi_p^2 ]) =
1$, this shows that $[\vb g] \mapsto \Phi_{[\vb g]}$
exhibits a bijection between a basis for $\mathcal C$
and $\Hom(\pi_1(\Sigma', Y), G)$.

To complete the proof, we need the following simple fact from algebraic topology: 
\begin{lem}
    Let $|Y|=R$.  Then $\pi_1(\Sigma/Y) \cong \pi_1(\Sigma)* \mathbb{Z}^{*(R-1)}$, where $*$ denotes the free product of groups \cite{hatcher}.
\end{lem}
\begin{proof} 
    Select an arbitrary distinguished ghost vertex $y_1 \in Y$. 
    Let $\Psi = \Sigma \cup Y^\prime$, where $Y^\prime$ is a 1D CW complex of $R-1$ edges $e_2,\ldots, e_R$ with boundaries $\partial e_j=\lbrace y_1,y_j\rbrace$ for each vertex \(y_j \in Y\).
    Since $Y'$ is connected and contractible, it follows that $\pi_1(\Sigma/Y) = \pi_1(\Psi/Y') \cong \pi_1(\Psi)$.
    Now take a spanning tree $T\subset \Sigma$ that connects $y_1,\ldots, y_R$ and contract a neighborhood of $T$ to the single point $y_1$.  
    These operations reveal that $\Psi$ is homotopy equivalent to $\Sigma \vee (\mathrm{S}^1)^{\vee (R-1)}$.  By van Kampen's Theorem, $\pi_1(\Sigma \vee (\mathrm{S}^1)^{\vee (R-1)}) = \pi_1(\Sigma)*\mathbb{Z}^{*(R-1)}$.
\end{proof}
To complete the proof using this fact,
we now note that $\mathrm{Hom}(\pi_1(\Sigma)* \mathbb{Z}^{*(R-1)},G)= \mathrm{Hom}(\pi_1(\Sigma),G)\times G^{R-1}$ because each additional free generator can be sent to an arbitrary element of $G$.  
\end{proof}

As the above proposition shows, one can either think of the additional logicals afforded
by multiple ghost vertices as living on strings stretching between two distinct
ghost vertices or as non-contractible loops on the complex with these two
vertices identified.

As an example, we consider a punctured grid as shown in
Fig.~\ref{fig:codespace_ghost}. The complex is already topologically non-trivial
due to the puncture: $\pi_1(\Sigma) \cong \Z$ is generated by a single loop winding around the hole.  When looking at $\mathrm{Hom}(\mathbb{Z},G) \cong G$, we can equivalent ask for the $G$-valued holonomy around any loop that winds around the hole.  We can take the generator of $\pi_1(\Sigma)$
to pass through a $b$ edge (colored in red in Figure \ref{fig:codespace_ghost}) exactly once. 
Adding the ghost vertices results in another non-trivial
loop, and we can see that $\pi_1(\Sigma/Y) \cong \Z^2$,
where the additional
generator corresponds to the edges labeled $a$. Thus the codespace is $\mathcal
C \cong \C G^2$, which agrees with Prop.~\ref{prop:codespace_rough}.

\begin{figure}[t]
  \centering
  \def\svgwidth{.3\textwidth}
\begingroup%
  \makeatletter%
  \providecommand\color[2][]{%
    \errmessage{(Inkscape) Color is used for the text in Inkscape, but the package 'color.sty' is not loaded}%
    \renewcommand\color[2][]{}%
  }%
  \providecommand\transparent[1]{%
    \errmessage{(Inkscape) Transparency is used (non-zero) for the text in Inkscape, but the package 'transparent.sty' is not loaded}%
    \renewcommand\transparent[1]{}%
  }%
  \providecommand\rotatebox[2]{#2}%
  \newcommand*\fsize{\dimexpr\f@size pt\relax}%
  \newcommand*\lineheight[1]{\fontsize{\fsize}{#1\fsize}\selectfont}%
  \ifx\svgwidth\undefined%
    \setlength{\unitlength}{289.62722658bp}%
    \ifx\svgscale\undefined%
      \relax%
    \else%
      \setlength{\unitlength}{\unitlength * \real{\svgscale}}%
    \fi%
  \else%
    \setlength{\unitlength}{\svgwidth}%
  \fi%
  \global\let\svgwidth\undefined%
  \global\let\svgscale\undefined%
  \makeatother%
  \begin{picture}(1,0.93780711)%
    \lineheight{1}%
    \setlength\tabcolsep{0pt}%
    \put(0.72242282,0.90548978){\color[rgb]{0,0,0}\makebox(0,0)[lt]{\lineheight{1.10000002}\smash{\begin{tabular}[t]{l}$a$\end{tabular}}}}%
    \put(0,0){\includegraphics[width=\unitlength,page=1]{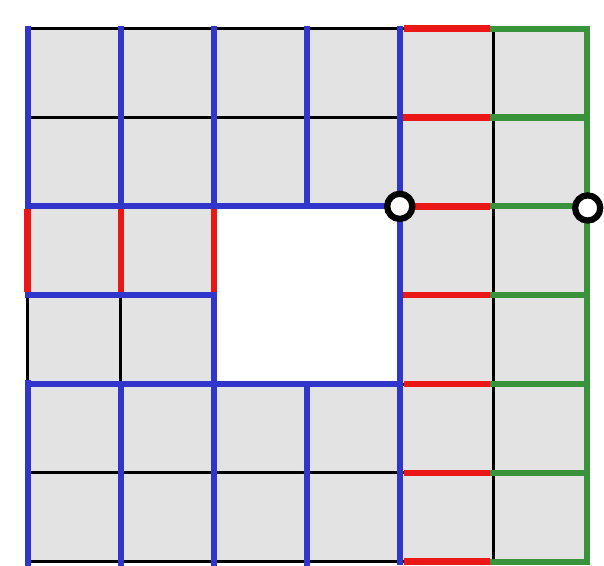}}%
    \put(-0.00215441,0.50297441){\color[rgb]{0,0,0}\makebox(0,0)[lt]{\lineheight{1.10000002}\smash{\begin{tabular}[t]{l}$b$\end{tabular}}}}%
  \end{picture}%
\endgroup%

  \caption{
    Example of codespace construction for a rough quantum double. Each box
    shaded gray hosts a $Z$-check, and all vertices except for the open circles
    host $X$-checks. The green edges correspond to $T_1$, a tree
    rooted at the first ghost check, and the blue edges correspond to $T_2$, a
    tree rooted at the second ghost check. We can see that $T = T_1 \sqcup T_2$ is a
    spanning tree for $\Sigma/Y$. The black edges are trivial once $T$ is fixed
    to the identity, and the red edges can be non-trivial. The loop corresponding to the edges labeled
    $a$ connects the two ghost vertices and is nullhomotopic in $\Sigma$ but
    non-trivial in $\Sigma/Y$. The edges labeled $b$ correspond to the loop
    wrapping around the puncture.
  }
  \label{fig:codespace_ghost}
\end{figure}

Returning to smooth quantum doubles (i.e., those without ghost checks), the second factor in the codespace structure result in Eq.~\eqref{eq:codespace_rough} disappears.
However, 
adding the last ghost check back into the stabilizer group is equivalent to adding a symmetry under large gauge transformations, i.e. conjugating all of the edges by $g$ for each $g \in G$.
More specifically, 
given $\phi \in \mathrm{Hom}(\pi_1(\Sigma),G)$, a codeword of a
smooth quantum double corresponds to a choice of $\phi$ up to conjugation:  $\phi \sim \phi^g$, where $\phi^g \in \mathrm{Hom}(\pi_1(\Sigma),G)$ is defined by $\phi^g(\sigma) \equiv g \phi(\sigma)g^{-1}$. 
This yields the following corollary.

\begin{cor}\label{cor:codespacesmooth}
  Given a smooth quantum double code on a 2D CW complex $\Sigma$,
  \begin{equation}
    \mathcal C  \cong \C[\Hom(\pi_1(\Sigma), G)/\Ad_G]~,
    \label{eq:codespacesmooth}
\end{equation} 
    where the quotient is by \(\Ad_G\), the \textit{adjoint action} of \(G\) (by conjugation).
\end{cor}

\begin{proof}
This is proven explicitly in \cite{Cui:2019lvb} when $\Sigma$ triangulates an orientable two-dimensional manifold, but their technique generalizes in an obvious way to any CW complex.  Here, we will instead leverage the proof of Proposition \ref{prop:codespace_rough}.

First, we remove a single $X$-check at some vertex $v_0$. By Prop.~\ref{prop:codespace_rough}, the codespace is
\begin{align}
\mathcal C \cong \C [\Hom(\pi_1(\Sigma, v_0), G)].
\end{align}
Now note that the basis vectors in $\mathcal{C}$ can be labeled by the holonomy of all non-contractible loops that begin and end at $v_0$.  We now include $X_{v_0}^g$ back into the stabilizer group.
Its effect is to conjugate all holonomies based at $v_0$ by $g$. Therefore, a connection
$\vb g$ no longer has a unique representative $\vb h$ which is gauge-fixed on
the spanning tree, but is only unique up to global conjugation by any $g \in G$.
Therefore $[\vb g]$ corresponds uniquely to a map $\Phi_{[\vb g]} \in \Hom(\pi_1(\Sigma,
v_0), G)/\Ad_G$, proving the claim.
\end{proof}

\begin{rmk}\label{rmk:weirddimension}
    Previously, we had commented that quantum double codes are not group GKP
    codes.   Now that we have reviewed the codespace of quantum double codes, it
    is easy to see why.  The logical dimension of the codespace of a group GKP
    code is $|H|/|K|$.   Since $|H|$ must divide $|G^n|$ in a group GKP code
    (Definition \ref{def:GKP}), we deduce that $\mathrm{dim}(\mathcal{C})$ also
    divides $|G|^n$.  This property does not hold for a quantum double code.  As
    a simple example, consider the symmetry group of a square, the dihedral
    group $G=\mathrm{D}_8$, placed on a cellulation of a two-dimensional torus
    with $\pi_1(\mathrm{T}^2)=\mathbb{Z}^2$. We can use
    Eq.~\ref{eq:codespacesmooth} to compute the dimension of the codespace. We
    see that $\Hom(\Z^2, G)$ corresponds to the number of central pairs in $G$,
    i.e. $a,b \in G$ such that $b \in Z(a)$, where $Z(a)$ is the set of
    elements commuting with $a$. Under conjugation by $g \in G$, we have $Z(a^g) =
    gZ(a)g^{-1}$,  which implies that $\Hom(\Z^2, G)/\Ad_{G}$ is equal to pairs
    $(C_a, Z^a_b)$, where $C_a$ is a conjugacy class in $G$ with representative
    $a$ and $Z^a_b$ is a conjugacy class in $Z(a)$. This is exactly equivalent to
    the standard formula in Ref.~\cite{Cong3}, because the number of conjugacy
    classes in $Z(a)$ is the same as the number of irreps of $Z(a)$.
    In the case $G = \mathrm D_8$, we find $\mathrm{dim}(\mathcal{C}) = 22$, which does not divide $8^n$ for any $n$.
\end{rmk}

\subsection{Quantum doubles with general boundaries}
The operation of quantum double codes can be challenging due to the requirements
on the topology of the quantum processor used to implement the code. Frequently,
it is natural to realize logical non-Clifford
operators or transversal logical operators on the boundary of a 2D geometry \cite{Ellison_2026,
  huang2026hybridlatticesurgerynonclifford}. Quantum double boundaries have been
previously studied for performing universal quantum computation \cite{Cong3} and
for the algebraic description of their anyonic excitations \cite{QD_boundary, Cong1,
  Cong2, Cong3}. In this section, we give a transparent description of the
codespace of such models.

The boundaries of quantum doubles are classified by
a subgroup $H \leq G$ and a cocycle $\phi \in H^2(H, \U(1))$. The group-CSS
formalism only encapsulates the trivial cocycle $\phi = 1$, corresponding to
Type I boundaries in Ref.~\cite{QD_boundary}, so this cocycle will not play a role in our formalism. Given a 1d boundary $\gamma \subseteq
\Sigma$, the boundary Hamiltonian is defined by 
\begin{align}
  \mathcal H^H_{\gamma} = -\sum_{v \in \gamma}\frac{1}{|H|}\sum_{h \in H}X^h_v -\sum_{e \in \gamma}Z^H_{e},\label{eq:Hbdy}
\end{align}
where the sum over $v \in \gamma$ runs over the vertices of the boundary and $e
\in \gamma$ are the edges of the boundary \cite{albert2021spinchainsdefectsquantum, QD_boundary}.
This restricts the boundary edges and
breaks down the gauge symmetry on the boundary vertices from $G$ to $H$.

This boundary description is an interpolation between the commonly used ``rough'' and
``smooth'' boundaries of the toric code when $G = \Z_2$ \cite{bravyiboundary}. To see this, we
introduce an $H$-valued ghost vertex, which is defined by taking a valid $X$-check
of a quantum double and restricting it to the subgroup $H$, with the ghost vertex previously introduced
corresponding to the trivial subgroup $H = \{1\}$. A rough boundary is defined
by removing all $X$-checks at the boundary, corresponding to $H = \{1\}$, and a
smooth boundary is defined by cutting the lattice along an edge, corresponding
to $H = G$. This is illustrated in Fig.~\ref{fig:rough_smooth_bdy}.
In the left panel, the orange dots do not correspond to $H$-valued ghost
checks because the removed checks, $X^g_v$ for $g \notin H$, fail to commute
with the $Z$-checks on the boundary. However, if the boundary is
contractible, then the boundary edges may be gauge-fixed to identity, and we can replace the boundary with a true ghost vertex. This is shown in the lower-left panel and is an example of the ``minimal identification" from Prop.~\ref{prop:roughQDonCW}.
We will
use this description to understand the codespace of codes with mixed boundaries.

\begin{figure}[t]
  \centering
  \def\svgwidth{.65\textwidth}
\begingroup%
  \makeatletter%
  \providecommand\color[2][]{%
    \errmessage{(Inkscape) Color is used for the text in Inkscape, but the package 'color.sty' is not loaded}%
    \renewcommand\color[2][]{}%
  }%
  \providecommand\transparent[1]{%
    \errmessage{(Inkscape) Transparency is used (non-zero) for the text in Inkscape, but the package 'transparent.sty' is not loaded}%
    \renewcommand\transparent[1]{}%
  }%
  \providecommand\rotatebox[2]{#2}%
  \newcommand*\fsize{\dimexpr\f@size pt\relax}%
  \newcommand*\lineheight[1]{\fontsize{\fsize}{#1\fsize}\selectfont}%
  \ifx\svgwidth\undefined%
    \setlength{\unitlength}{1175.64930833bp}%
    \ifx\svgscale\undefined%
      \relax%
    \else%
      \setlength{\unitlength}{\unitlength * \real{\svgscale}}%
    \fi%
  \else%
    \setlength{\unitlength}{\svgwidth}%
  \fi%
  \global\let\svgwidth\undefined%
  \global\let\svgscale\undefined%
  \makeatother%
  \begin{picture}(1,0.33844426)%
    \lineheight{1}%
    \setlength\tabcolsep{0pt}%
    \put(0,0){\includegraphics[width=\unitlength,page=1]{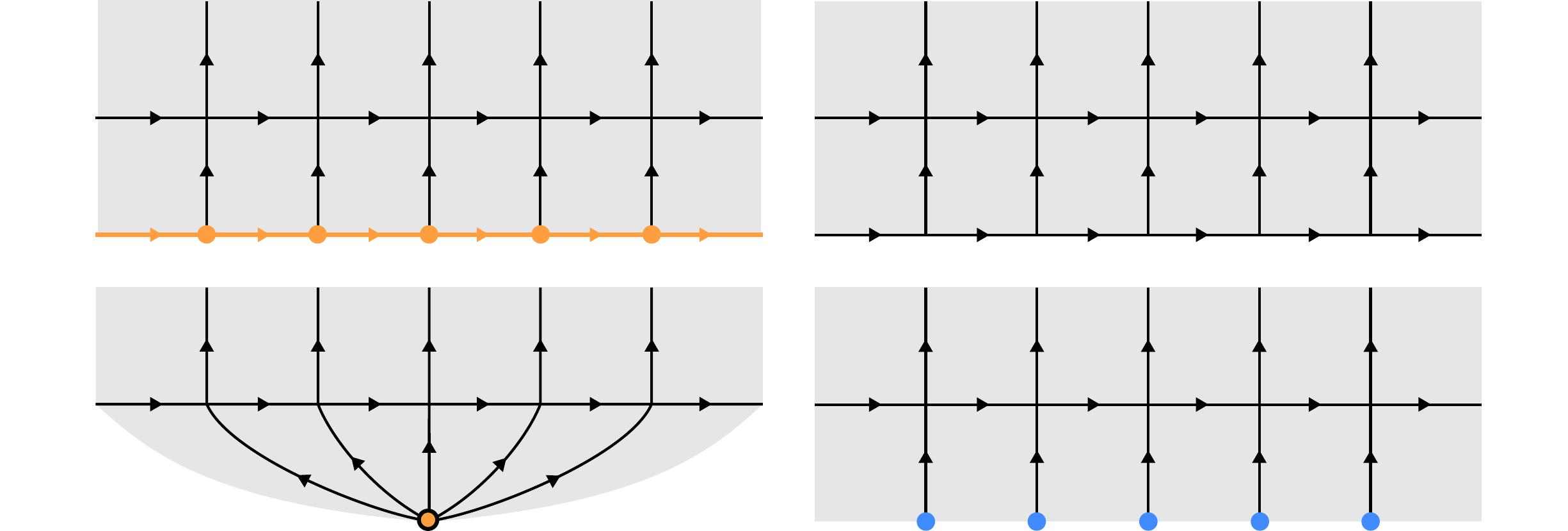}}%
    \put(0.95028474,0.19003833){\color[rgb]{0,0,0}\makebox(0,0)[lt]{\lineheight{1.10000002}\smash{\begin{tabular}[t]{l}$H = G$\end{tabular}}}}%
    \put(0.01574951,0.20439998){\color[rgb]{0,0,0}\makebox(0,0)[lt]{\lineheight{1.10000002}\smash{\begin{tabular}[t]{l}$\Sigma$\end{tabular}}}}%
    \put(-0.00053075,0.12395338){\color[rgb]{0,0,0}\makebox(0,0)[lt]{\lineheight{1.10000002}\smash{\begin{tabular}[t]{l}$\Sigma/ \gamma$\end{tabular}}}}%
    \put(0.94929102,0.14578779){\color[rgb]{0,0,0}\makebox(0,0)[lt]{\lineheight{1.10000002}\smash{\begin{tabular}[t]{l}$H = \{1\}$\end{tabular}}}}%
    \put(0,0){\includegraphics[width=\unitlength,page=2]{figures/rough_smooth_bdy_svg-tex.pdf}}%
  \end{picture}%
\endgroup%

  \caption{Our boundary construction and relationship to commonly studied
    boundary types. On the left, we show a boundary $\gamma$ on a complex
    $\Sigma$ classified by a subgroup
    $H$. The orange dots correspond to $X$-checks restricted to $H$ on those vertices.
    The orange lines correspond to checks $Z^H_{q}$. Gray shaded plaquettes
    $p$ correspond to $Z$-checks $Z^{\{1\}}_{p}$. If the boundary is
    contractible, then we can obtain a topologically equivalent complex
    $\Sigma' = \Sigma/Y$, where $Y$ represents the boundary and becomes a ghost
    cell in the double corresponding to $\Sigma'$, as shown by a closed orange circle.
    Blue dots correspond to vertices
    with no associated checks, i.e. $H = \{1\}$.
    On the right from top to bottom are shown smooth and
    rough boundaries respectively. Since the boundary edges in the bottom panel are
    fixed to $H = 1$, we freely remove the boundary edges, hence the term ``rough''.}
  \label{fig:rough_smooth_bdy}
\end{figure}

The first model we will treat is an open surface with $2R$ alternating rough and
smooth boundaries valued in $H_1, \dots, H_R$, as
shown in Fig.~\ref{fig: open boundary QD}. We can use the method shown in Fig.~\ref{fig:rough_smooth_bdy}
to contract each of the $R$ rough boundaries to an $H_i$-valued ghost vertex. This
manifests the $\prod_{i= 1}^R G/H_i$ boundary
symmetry of such a model, corresponding to the ghost checks removed from the
stabilizer group. In the language of quantum error correction, this implies that such non-Abelian doubles host
transversal logical operators at these rough boundaries, an observation
previously made in Refs.~\cite{huang2026hybridlatticesurgerynonclifford, Ellison_2026}. 

Consider first the case where $H_i = \{1\}$ for each $i$. Similar to the previous
discussion, we glue all rough boundaries to a single ghost vertex and obtain a
new CW complex $\Sigma'$, which is homotopy equivalent to a disk with $R$ points
on its boundary identified, which is in turn homotopy equivalent to the union of
$R-1$ circles joined at a point:  $(\mathrm{S}^1)^{\vee (R-1)}$.  Therefore by
Prop.~\ref{prop:codespace_rough}, the codespace $\mathcal C$ is given by
\begin{align}
\mathcal C \cong \C[\mathrm{Hom}(\pi_1(\Sigma'),G)]=\C[G^{R-1}],
\end{align}
i.e., we can encode $R-1$ logical $G$-qudits in the ground-state subspace of the quantum double model on the square lattice with $R$ rough boundaries.

\begin{figure}[t]
  \centering
  \def\svgwidth{0.7\textwidth}
  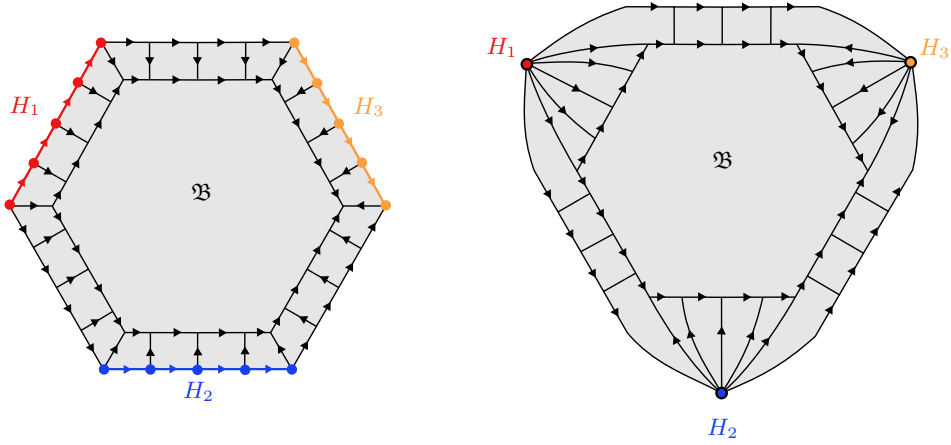
  \caption{
    Quantum double model with $R = 3$ smooth boundaries and rough boundaries
    corresponding to subgroups $H_1, H_2, H_3 \leq G$ indicated by red, blue,
    and orange. The checks on each boundary are given by Eq.~\eqref{eq:Hbdy}.
    In the right panel, we identify each boundary with a single ghost vertex to
    form a rough quantum double [Def.~\ref{def:roughQD}]. This manifests the
    $\prod_i G/H_i$ symmetry of the boundary, or equivalently, the existence of
    transversal logical operators corresponding to the removed checks at the
    ghost vertex.
    }
\label{fig: open boundary QD}
\end{figure}

Next, we can introduce a formulation of holes in the language of ghost vertices.
There are different conventions for defining quantum double boundaries and
holes. Consider a lattice with bulk $\mathfrak B$ and a contractible region
$\mathfrak h$ called a hole, labeled with a subgroup $H \leq G$, as shown in
Fig.~\ref{fig:hole}. Given a region $\Omega$, consider the commuting Hamiltonian
\begin{align}
  \mathcal H_{\Omega}^H = -\sum_{v \in V(\Omega)}\frac{1}{|H|}\sum_{h \in H}X^k_v - \sum_{e \in E(\Omega)}Z^H_e - \sum_{p \in F(\Omega)}Z^{\{1\}}_p
\end{align}
where $V(\Omega), E(\Omega)$, and $F(\Omega)$ are the vertices, edges, and
plaquettes of $\Omega$. The ground space of this Hamiltonian corresponds
simply to restricting the topological order within $\Omega$ from $G$ to $H$.
Then we define the Hamiltonian of the system to be
\begin{align}
  \mathcal H = \mathcal H^G_{\mathfrak B} + \sum_i \mathcal H^{H_i}_{\mathfrak h_i},
\end{align}
where $\mathfrak h_i$ is a hole classified by the subgroup $H_i \leq G$.

As we depict graphically in Fig.~\ref{fig:hole}, we can represent such a hole by
taking a quotient of the complex by $\mathfrak h$ and replacing it with a single $H$-ghost vertex. Because $\mathfrak h$ is assumed to
be contractible, we can gauge-fix $\mathfrak h$ to the identity. Then
the only remaining degree of freedom is multiplying each of the bulk edges
attached to $\mathfrak h$ simultaneously by $g$ up to its coset in $H$, which is exactly a
$H$-valued ghost check, or a transversal logical operator acting on the boundary of $\mathfrak h$. 

The importance of using a hole (left panel of Fig.~\ref{fig:hole}) over simply defining a ghost defect (right panel of Fig.~\ref{fig:hole}) is that the former is manifestly LDPC, whereas the latter is not if the size of the hole grows with $n$, which is required to have $d_X$ also grow with $n$.

\begin{figure}[t]
  \centering
  \def\svgwidth{.55\textwidth}
\begingroup%
  \makeatletter%
  \providecommand\color[2][]{%
    \errmessage{(Inkscape) Color is used for the text in Inkscape, but the package 'color.sty' is not loaded}%
    \renewcommand\color[2][]{}%
  }%
  \providecommand\transparent[1]{%
    \errmessage{(Inkscape) Transparency is used (non-zero) for the text in Inkscape, but the package 'transparent.sty' is not loaded}%
    \renewcommand\transparent[1]{}%
  }%
  \providecommand\rotatebox[2]{#2}%
  \newcommand*\fsize{\dimexpr\f@size pt\relax}%
  \newcommand*\lineheight[1]{\fontsize{\fsize}{#1\fsize}\selectfont}%
  \ifx\svgwidth\undefined%
    \setlength{\unitlength}{1087.03574834bp}%
    \ifx\svgscale\undefined%
      \relax%
    \else%
      \setlength{\unitlength}{\unitlength * \real{\svgscale}}%
    \fi%
  \else%
    \setlength{\unitlength}{\svgwidth}%
  \fi%
  \global\let\svgwidth\undefined%
  \global\let\svgscale\undefined%
  \makeatother%
  \begin{picture}(1,0.48548254)%
    \lineheight{1}%
    \setlength\tabcolsep{0pt}%
    \put(0,0){\includegraphics[width=\unitlength,page=1]{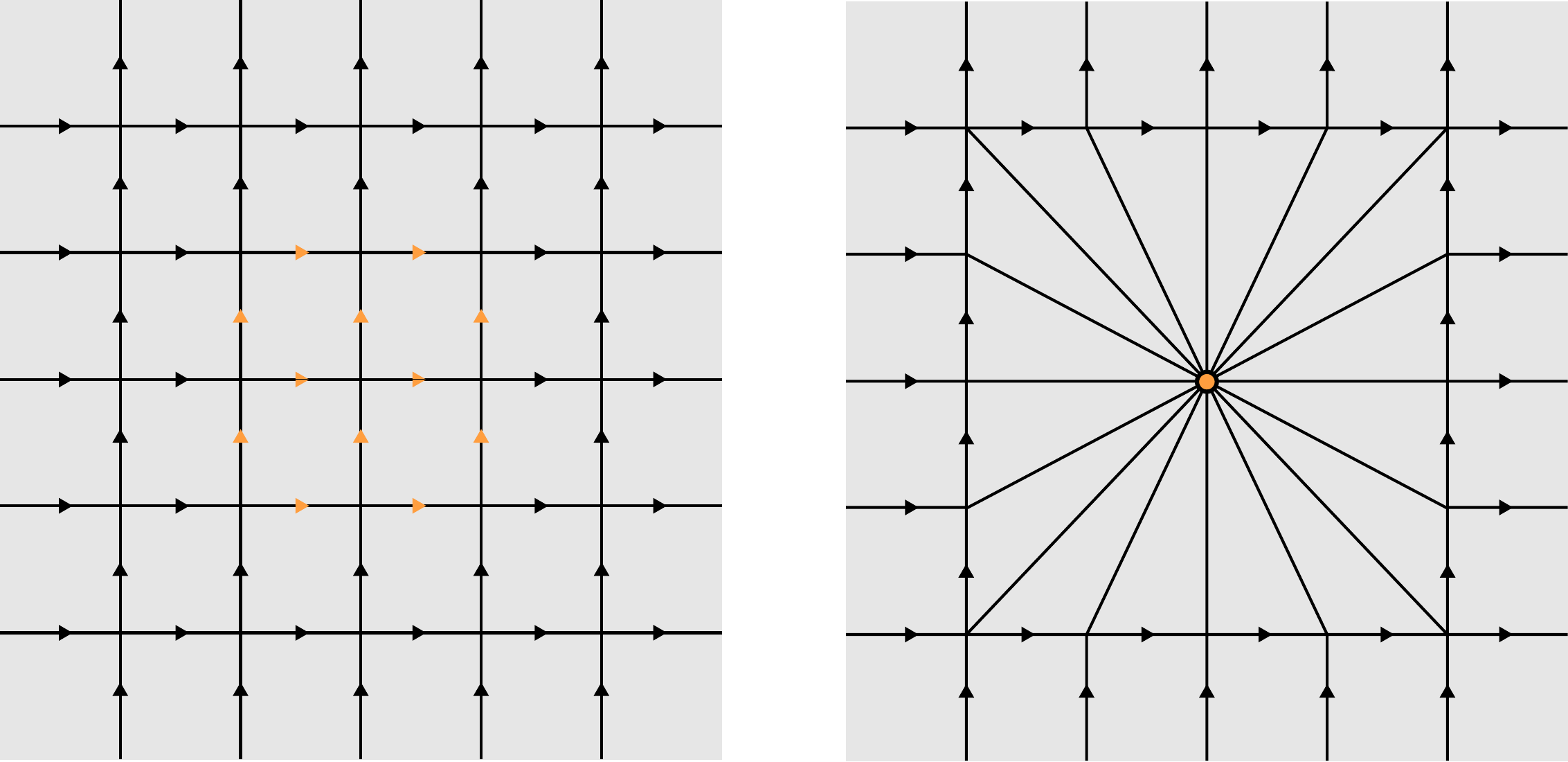}}%
    \put(0.19001048,0.27555786){\color[rgb]{0,0,0}\makebox(0,0)[t]{\lineheight{1.10000002}\smash{\begin{tabular}[t]{c}\(\mathfrak h\)\end{tabular}}}}%
    \put(0.19042632,0.35657194){\color[rgb]{0,0,0}\makebox(0,0)[t]{\lineheight{1.10000002}\smash{\begin{tabular}[t]{c}\(\mathfrak B\)\end{tabular}}}}%
    \put(0,0){\includegraphics[width=\unitlength,page=2]{figures/hole_svg-tex.pdf}}%
  \end{picture}%
\endgroup%

  \caption{Our formulation of rough holes in a quantum double using ghost vertices.
    On the left is depicted a complex $\Sigma = \mathfrak B \sqcup \mathfrak h$ with vertices corresponding to
    $X$-checks, edges corresponding to qudits, and shaded plaquettes
    corresponding to $Z$-checks. The orange edges and vertices correspond to the
    hole $\mathfrak h$, where the edges are restricted to lie in $H$ and the
    gauge transformations at the vertices are also restricted to $H$. On the
    right, taking a quotient of $\Sigma$ with $\mathfrak h$, the entire hole can
    be represented with a single ghost vertex.
  }
  \label{fig:hole}
\end{figure}

To comment on the relationship between this theory of holes and boundaries in
relation to others found in the literature; our boundary construction is the
same as those found in Refs.~\cite{albert2021spinchainsdefectsquantum, QD_boundary}. We do
not adopt the construction of Ref.~\cite{Cong2} for holes, because 
these holes support small logical operators connecting checks at the
boundary and a logical degree of freedom encircling the boundary.

Now that we have seen two examples of interest---open boundaries and
holes---which can be represented by a QD on a CW complex with $H$-valued ghost checks, we
will prove formulas for the code space and logical operators of such models.
Using the gauging interpretation, it is straightforward to generalize our
results from the previous section when $H$ was the trivial subgroup.

\begin{cor}\label{cor:intermediate}
    Given a quantum double code where ghost vertex $y_i\in Y$ is associated with subgroup $H_i$, the codespace is \begin{equation} \label{eq:intermediatecodespace}
        \mathcal{C}= \mathbb{C}[\mathrm{Hom}(\pi_1(\Sigma/Y),G)/(H_1 \times \cdots H_R) ]
    \end{equation}
    where $R=|Y|$, and $H_1\times \cdots H_R$ acts on the codespace as follows (without loss of generality): Given data $(\phi, g_{12},\ldots, g_{1R})\in \mathrm{Hom}(\Sigma, G)\times G^{R-1} = \mathrm{Hom}(\Sigma/Y, G)$, which we can interpret as the holonomy of non-contractible loops on $\Sigma$ rooted at $y_1$ ($\phi$) together with the group element traced along a path from $y_1$ to $y_i$ ($g_{1i}$ for $i=2,\ldots, R$), $\mathbf{h}\in H_1\times \cdots \times H_R$ acts as \begin{equation}\label{eq:codespacegeneral}
        \mathbf{h} \cdot (\phi, g_{12},\ldots, g_{1R}) = (\mathrm{Ad}_{h_1}\circ \phi, h_1 g_{12} h_2^{-1},\ldots, h_1 g_{1R}h_r^{-1}).
    \end{equation}
\end{cor}
\begin{proof}
    The proof mirrors that of Corollary \ref{cor:codespacesmooth}, which handles the case $r=1$ and $H_1=G$.

    While it might be unclear whether the formula [Eq.~\eqref{eq:codespacegeneral}]
is well defined because of the choice of the ``first'' ghost vertex from which the
logical loops originate, it turns out that there
is no issue with this. To see this, consider Fig.~\ref{fig: open boundary QD}.
Let $\gamma_1, \gamma_2, \gamma_3$ be strings connecting boundaries $1$ and $2$,
$2$ and $3$, and $1$ and $3$ respectively. Due to the plaquette constraints, we
have $g_{\gamma_3} = g_{\gamma_2}g_{\gamma_1}$. Under the equivalence relation,
\begin{equation}
  \vb k \cdot g_{\gamma_1}g_{\gamma_2} = k_1g_{\gamma_1}k_2^{-1}k_2g_{\gamma_2}k_3^{-1} = k_1g_{\gamma_3}k_3^{-1} = \vb k \cdot g_{\gamma_3}.
\end{equation}
This shows that the equivalence relation is well-defined on homotopy classes of
loops in $\Sigma'$.
\end{proof}

We are not aware of an elegant formula for \eqref{eq:intermediatecodespace} in general.   If however one of the ghost vertices has $H_1=1$, then we note the elegant simplification \begin{equation}
    \mathcal{C}= \mathbb{C}[\mathrm{Hom}(\pi_1(\Sigma/Y),G)/(H_1 \times \cdots H_r) ] =\mathbb{C}[\mathrm{Hom}(\pi_1(\Sigma),G)]\otimes \mathbb{C}[G/H_2]\otimes \cdots \otimes \mathbb{C}[G/H_r]. 
\end{equation}
This becomes $\mathcal C \cong \bigotimes_{i=1}^R\C[G/H_i]$ for the surface with
$R$ rough boundaries [Fig.~\ref{fig: open boundary QD}], because in this case
$\pi_1(\Sigma)$ is trivial.

We note that if we take $G=\mathbb{Z}_2$, the codespace can also be interpreted in terms of the ``Majorana defects'' between rough and smooth boundaries. In a rough quantum double with $R$ rough boundaries and $R$ smooth boundaries, there is a Majorana zero mode at each junction between a rough boundary and a smooth boundary. Therefore, there are $2R$ Majorana zero modes ($\gamma_1,\cdots,\gamma_{2R}$) on the boundary, which is equivalent to $R$ physical (complex) fermions. Furthermore, there is a fermion parity even constraint, namely \cite{Sarkar_2024, bravyi_coherent_errors}
\begin{align}
\ii^R\gamma_1\cdots\gamma_{2R}=1,
\end{align}
which leaves only $R-1$ independent physical fermions on the boundary. As a consequence, the codespace dimension is $\mathrm{dim}(\mathcal{C}) = \mathbb{Z}_2^{R-1}$, in agreement with our formula.

\subsection{Logical operators}
Now we will define the logical operators for rough and smooth QD codes on
non-manifold CW complexes. Going forward, we will restrict our attention to the
$H = \{1\}$ case for simplicity when there are ghost vertices, but our formalism
straightforwardly generalizes.

It is convenient to first define the logical operators for a rough quantum
double. 
In a smooth double, invariance under large gauge transformations (see Def.~\ref{def:gauge-theory}) makes
the logical operators take a more complicated form. 
Unlike
Abelian CSS codes, the logical operators cannot be separated into purely
$X$-type and purely $Z$-type in the sense of Defs.~\ref{def:Xcheck} and \ref{def:Zcheck}.

\begin{prop}[Logical operators of a rough quantum double]\label{prop:Zlogical}
  Let $\Sigma$ be a 2D CW complex with a non-empty set of ghost vertices $Y$.
  The union of trees $T' = \bigsqcup_i T_i$ from Prop.~\ref{prop:codespace_rough}
  defines a presentation of $\pi_1(\Sigma/Y)$ with the edges in $\Sigma - T'$
  as generators $\mathcal G$ and the 2-cells as relations.
  A complete set
  of $Z$-logical operators is given by\footnote{In gauge theory, one typically thinks of Wilson loops, i.e. characters of group representations evaluated on holonomies, as being a set of gauge-invariant observables. Indeed, for Abelian groups and for common groups such as $\SU(2)$, they are a complete set of observables. In \cite{Cui:2019lvb}, it is shown that an \textit{outer class automorphism}, i.e. $\phi: G \to G$ such that $\phi(g) \sim g$ for all $g$ but $\phi(g) \ne hgh^{-1}$ for any $h\in G$, provides a transformation between gauge-invariant states that is not detectable by Wilson loops. This is proven for excited states in contractible complexes, however we can see that the modification $\ket{[\mathbf g]} \to \ket{[\phi(g_1), \dots, \phi(g_n)]}$ also acts non-trivially on the ground space in a non-contractible complex. Our construction of logical $Z$ operators avoids this ambiguity.
  }
  \begin{equation} \label{eq:Zlogicalrough}
    Z^{\mathrm{L}}_{\phi, \gamma}\ket{\vb g} \equiv \mathbb I(\vb g_\gamma = \phi([\gamma]))\ket{\vb g}
  \end{equation}
  for each $\phi \in
  \Hom(\pi_1(\Sigma/Y, Y), G)$ and $\gamma \in \mathcal G$.

  Given $\phi_1, \phi_2 \in
  \Hom(\pi_1(\Sigma/Y, Y), G)$ and $\gamma \in \mathcal G$ associated to an edge
  $q_\gamma \in \Sigma - T'$, a complete set of logical operators for the code is given by
  \begin{align}
    X_{\phi_{2}\phi_{1}}^{\mathrm{L}}=\prod_{\gamma\in\mathcal{G}}\overrightarrow{X}_{q_\gamma}^{\left(\phi_{2}(\gamma)\phi_{1}(\gamma)^{-1}\right)^{\mathbf{g}_{\sigma(\gamma)}}}Z_{\phi_{1},\gamma}^{\mathrm{L}}
      \label{eq:logicals}
  \end{align}
  where $\sigma(\gamma)$ is a path connecting the terminal vertex of $q_\gamma$ to $Y$, and where $h^g:=ghg^{-1}$.
\end{prop}

Note that the checks removed at the ghost vertices are included as  logical operators in Eq.~\ref{eq:logicals}, so it is only ever beneficial to create ghost vertices at very high degree vertices.
Importantly, the ghost checks also create logical operators on paths that are not non-contractible loops on $\Sigma$, but are on $\Sigma/Y$.  More transparently, we find $Z$-logical operators defined along paths between any two distinct ghost vertices in $\Sigma$, in addition to those that wind around non-contractible loops on $\Sigma$.
With respect to intersection pairing between $X$- and $Z$-logicals, we can view the new $Z$-logical operators as the loops pairing non-trivially with the removed $X$-checks.
This situation is
illustrated in Fig.~\ref{fig:logicals_ghostcells}.

\begin{figure}
\centering
\def\svgwidth{.55\textwidth}
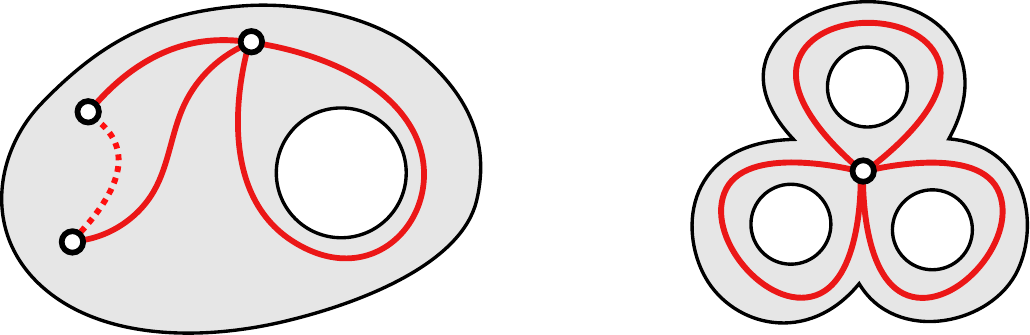
\caption{Illustration of $Z$-logical operators on a topologically non-trivial
  surface with ghost vertices. The underlying CW complex is some cellulation of the
  surfaces. (a) The 2D surface $\Sigma$, shown in light blue, has
  one topologically non-trivial loop winding around the puncture which generates
$\pi_1(\Sigma) \cong \Z$, as illustrated with a red line. The open circles show
 the set of ghost vertices $Y$. The product
of qudits along a loop connecting any two ghost vertices is a $Z$-logical,
shown by the solid red lines, which depict the 3 independent logical operators (the dotted red line is also a logical operator but does not contain additional information). (b) The quotient complex $\Sigma/Y$
deformation retracts onto a wedge-sum of three circles corresponding to the
logical operators shown in panel (a).}
\label{fig:logicals_ghostcells}
\end{figure}

\begin{proof}

  First, we treat the $Z$-logicals. The operator $Z_{\phi, \gamma}$ checks if the holonomy of the loop  
  $\gamma$ is equal to $\phi(\gamma)$, and from our construction of the code
  space,  it is clear that this is a complete set of $Z$ operators. We only need
  to confirm that these operators  are well-defined and commute with all \(X\)-checks.
  Since $\ket{[\vb g]} \in \mathcal C$ means that $\vb g$ represents a flat
  connection (Def.~\ref{def:gauge-theory}), the holonomy $g_{\gamma}$ depends only on the homotopy class $[\gamma]$.
  Since applying any vertex operator $X_{v}^{g}$ does not change the
  holonomy, $Z_{\phi, \gamma}$ also does not depend on the representative of
  $[\vb g]$. This also implies that $Z_{\phi, \gamma}$ commutes with $X_{v}^g$.

  Now we will construct the remaining logicals. Let $\Sigma' \equiv
  \Sigma/Y$~\eqref{eq:quotient}.
  Given any state $\ket{[\vb g]}$, we will first
  gauge-fix this state to identity on $T'$ as in Prop.~\ref{prop:codespace_rough}.
  By construction, $\Sigma'/T'$ is a bouquet graph of circles
  $\bigvee_{\alpha}\textnormal{S}^1_{\alpha}$ with attached 2-cells. 
  The fundamental group $\pi_1(\Sigma')$ has the group
  presentation $\pi_1(\Sigma') = \langle  ~F~|~R~
  \rangle$, where $F$ is the alphabet of the circles in the bouquet, and where $R$ is a set of relations generated by the attaching maps.\footnote{While this might not be the optimal presentation of $\pi_1$, in general there is no canonical presentation of a group, and even discerning whether two finite presentations represent isomorphic groups is an undecidable problem.}
  Suppose that $\phi_1, \phi_2:\pi_1(\Sigma', Y) \to G$ are two homomorphisms.
  Left-multiplication on the edge $q \in \Sigma - T'$ representing $\gamma$ by
  $\phi_2([\gamma])\phi_1([\gamma])^{-1}$ maps $\ket{\phi_1}$ to $\ket{\phi_2}$,
  but this operation does not commute with $X$-checks, so we consider
  re-adding the checks sequentially, beginning with the single checks that act by left-multiplication on $q$, i.e., associated to the outgoing vertex $v_q$.
The edge $q$ corresponds to a loop $\gamma$. Symmetrizing over the checks $X^h_{v_q}$ yields the map
  \begin{equation}
  \ket{q', q} = \ket{1,\phi_1(\gamma)} \mapsto \frac{1}{|G|}\sum_k\ket{k^{-1}, k\phi_1(\gamma)},
\end{equation}
where $q'$ is any outgoing edge from $q$ in $T'$.
We assume WLOG that both edges are positively oriented from the vertex.
Then the logical operator $X^L_{\phi_2\phi_1}$ maps
\begin{equation}
  \sum_k\ket{k^{-1}, k\phi_1(\gamma)} \mapsto \sum_k\ket{k^{-1}, kh\phi_1(\gamma)} = \sum_{k}\ket{k^{-1}, h^{k}k\phi_1(\gamma)},
\end{equation}
where $h = \phi_2(\gamma)\phi_1(\gamma)^{-1}$, and where $k$ is any edge contained in
$T'$ and connected to $v_q$. Continuing this process sequentially, we obtain $k$ as the oriented product of
edges in a path $\sigma(q)$ through $T'$ connecting the vertex $v_q$ to $Y$. We
repeat this for each $q$. Define the resulting operator $X_{\phi_2\phi_1}$,
which explicitly acts on each edge $q \in \Sigma'-T'$ by left-multiplication by
$(\phi_2(\gamma)\overline \phi_1(\gamma))^{g_{\sigma(q)}}$.

We have constructed $X_{\phi_2\phi_1}$ such that $X_{\phi_2\phi_1}\ket{\phi_1}= \ket{\phi_2}$, but unless $\pi_1$
is a free group, $X_{\phi_2\phi_1}$ will in general not preserve the codespace when acting on other states. Therefore,
we consider the family of logical operators $X_{\phi_2\phi_1}^{\mathrm{L}} =
X_{\phi_2\phi_1}\prod_{\gamma \in \mathcal G}Z^{\mathrm L}_{\phi_1, \gamma}$. This forms a complete basis
for logical operators on the code space of the QD code.
\end{proof}

We note that a systematic way of constructing the logical operators of rough quantum double codes is proposed as aforementioned. It is not necessarily the most convenient way to construct the logical operators. In fact, for rough quantum double codes, we can construct a transversal $X$-logical operator at any ghost vertex \cite{huang2026hybridlatticesurgerynonclifford, Ellison_2026}: consider the product of $X$-type operators among all edges connecting to the pertinent ghost vertex; in fact, this operator is exactly the missing $X$ check at the ghost vertex, which commutes with all existing $Z$ checks. Furthermore, this operator is not in the stabilizer group, hence it must be a logical $X$ operator on the rough boundary. 

\begin{cor}[Logical operators of a smooth double]
If $\Sigma$ is the CW complex corresponding to a smooth quantum double,
$\mathcal G$ is a generating set of $\pi_1(\Sigma, v_0)$ for some basepoint
$v_0$, then let $Z^{\mathrm L}_{\phi}, X^{\mathrm L}_{\phi_2\phi_1}$ be the
logical operators of the double with a ghost vertex at $v_0$.

A complete set of $Z$-logicals of the smooth double are then given by
\begin{equation}
\widetilde Z_{\phi}^{\mathrm {L}} \equiv \frac{1}{|G|}\sum_{g \in G}\prod_{\gamma \in \mathcal G}Z_{\phi_1^g, \gamma}^{\mathrm L}, \label{eq:Zlogicalsmooth}
\end{equation}
and more generally, a complete set of logical operators is
\begin{equation}
\widetilde X_{\phi_2\phi_1}^{\mathrm{L}} \equiv \frac{1}{|G|}\sum_{g \in G}X^{\mathrm L}_{\phi_2^g \phi_1^g}
\label{eq:logicalssmooth}
\end{equation}
\end{cor}

\begin{proof}
  After removing a single $X$-check from the set of $X$-stabilizers to create a ghost vertex, we can apply Prop.~\ref{prop:Zlogical} to construct the set of logical operators. 
  Note that when there is only one ghost vertex, there are no additional non-contractible loops on the quotient complex.
  Adding in the single removed ghost check is equivalent to symmetrizing these logicals over the large gauge transformations, i.e. conjugation of all holonomies by $g$ for each $g \in G$.
   The operators constructed this way are are well-defined on the code
  space and commute with the checks, and so they are logical operators.
\end{proof}

Later on, we will use a systolic argument to bound the $Z$-distance $d_Z$ of quantum double codes. We will find that while $d_X$ can be very large, even $\mathrm O(n)$, the $Z$-distance is fundamentally restricted by graph theory, as the following corollary shows.

\begin{cor}\label{cor:dZloop}
    Every quantum double code is equivalent to one drawn on a CW complex $\Sigma$ (possibly with ghost vertices $Y$), where 
    each non-contractible loop supports a $Z$-logical operator. In such a complex, $d_Z$ is the length of the smallest non-contractible loop in the CW complex $\Sigma/Y$.
\end{cor}

\begin{proof}
    An operator of the form of Eqs.~\eqref{eq:Zlogicalrough} or \eqref{eq:Zlogicalsmooth} restricted to a single non-contractible loop violates the Knill-Laflamme condition [Eq.~\eqref{eq:KL}] and is therefore a $Z$-logical operator, unless the holonomy around this non-contractible loop is forced to be trivial (a trivial holonomy solution where $g_j=1$ always is part of a codeword, by construction).  If it is forced to be trivial, then the codespace matches that of a quantum double code where we add a 2-cell glued around this non-contractible loop, turning it into a contractable loop.   This means that there exists a CW complex such that every non-contractible loop admits a logical operator.\footnote{A simple illustration of this phenomenon is to consider a code with a single qudit ($n=1$) defined on 
    a group $G$ such that a prime $p$ does not divide the order of $G$, with a 2-cell glued to enforce holonomy $g^p=1$. 
    The loop $g$ is non-contractible on the CW complex, however on $G$ the only solution to the constraint is $g=1$.  We can therefore freely add another 2-cell enforcing $g=1$, making the $g$ loop contractible.}
    
    In contrast, any operator acting on $m<d_Z$ qudits must obey Eq.~\eqref{eq:KL}: using $X$-checks (gauge transformations) we can arbitrarily change the simultaneous state of all $m$ qudits, as we can always place these $m$ edges in the CW complex into a spanning tree.  Combining these two observations together we deduce that $d_Z$ is exactly the length of the smallest non-contractible loop.
\end{proof}

As an example, we will show in detail how Eq.~\ref{eq:logicalssmooth} reproduces the closed ribbon
operators on the torus quantum double.
\begin{figure}
  \centering
  \def\svgwidth{.35\textwidth}
  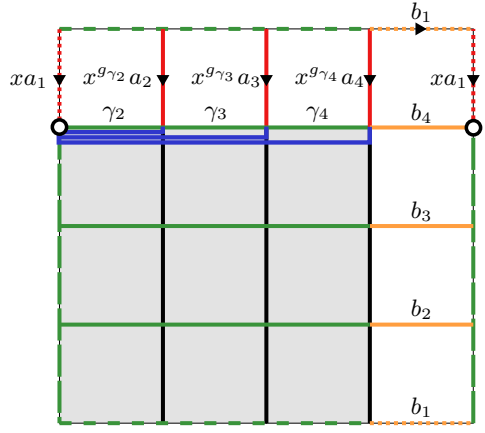
  \caption{
  Example showing the derivation of the standard logical ribbon operators on a cellulation of a torus.
  The green edges show the spanning tree, with the open circle at the
    root $v_0$. The dotted lines show the edges that are identified through the
    periodic boundaries. 
    After gauge-fixing $T$ to $\ket{\vb 1}$, the $Z$-checks also fix additional edges to 1, which are shown in black. There is a complete set of logical operators which is not supported on the edges in $T$ or the edges colored black, and this region is shaded gray.
    The paths $\gamma_{2,3,4}$ (in blue) label paths through $T$ based at $v_0$ from each of
    the edges which are left-multiplied by $x^{g_\gamma} \neq e$. The labels $a_1, \dots, a_4$ and $b_1, \dots,
    b_4$ label the generators of the fundamental group.}
  \label{fig:torus}
\end{figure}
The fundamental group of the cellulation of the torus in Fig.~\ref{fig:torus} has the presentation
\begin{align}
  \langle a_1, \dots, a_4, b_1, \dots, b_4 \ | \ a_na_{n+1}^{-1}, b_nb_{n+1}^{-1}, [a_1,b_1] \rangle \cong \mathbb Z^2
\end{align}
In this case, it is obvious from the presentation that $\pi_1(\Sigma)$ has two
generators, $a$ and $b$, with $a = a_1=a_2=a_3=a_4$ and $b = b_1=b_2=b_3=b_4$,
and the fundamental group is isomorphic to $\mathbb Z^2$. Therefore,
$\Hom(\pi_1(\Sigma), G)$ is in bijection with pairs $(a, b)$ such that $b \in
Z(a)$, where $Z(x)$ is the centralizer of $x$.
In particular, $a$ can be multiplied by $x$ for any $x \in Z(b)$. On the
complex with all $X$-checks removed, the logical operator
\begin{equation}
X_{a \mapsto xa} \equiv \sum_{(a, b), b \in Z(x)}X^{L}_{\phi_{(xa, b)} \phi_{(a,b)}}~,
\end{equation}
is represented by
\begin{equation}
  \sum_{z \in Z(x)}\prod_i\overrightarrow X_{a_i}(x)\mathbb{I}(b_4 = z)~.
\end{equation}
Gauging the $X$-checks gives the operator
\begin{align}
  \frac{1}{|G|}\sum_{k \in G}\sum_{z \in Z(x^k)}\prod_i \overrightarrow X_{a_i}(x^{g_{\gamma_i}k})\delta(g_{\sigma}^k = z)
  = \frac{|C(x)|}{|G|}\sum_{c \in C(x)}\sum_{z \in Z(c)}\prod_i \overrightarrow X_{a_i}(c^{g_{\gamma_i}})\delta(g_{\sigma} = z) = F^{(C(x), \rm{triv})}~,
\end{align}
where $\sigma$ is the loop based at $v_0$ passing through $b_4$, or
equivalently, the loop encircling the upper handle of the torus, and $C(x)$ is the conjugacy class of $x$. The product runs over the edges labeled $a_i$ in Fig.~\ref{fig:torus}. This is the standard closed magnetic anyon ribbon operator $F^{(C(x), \rm{triv})}$ (transforming in the trivial representation of $Z(x)$) \cite{Kitaev_2003, SWSSB_nonabelian}, 
and other representations would be obtained by choosing different weights in the sum over $b \in Z(x)$. This also illustrates how the ribbon operators of a general CW quantum double can be much more complicated than those for a cellulation of an orientable surface; the fundamental group of a genus-$g$ orientable surface $M$ has the presentation 
\begin{equation}\label{eq:genusg}
\langle a_1, b_1, \dots, a_g, b_g \ | \ [a_1, b_1] \dots [a_g, b_g]\rangle,
\end{equation}

In the condensed matter literature, the ribbon operators, which are non-unitary, are often referred to as non-invertible symmetries. When $G$ is Abelian, the ribbon operators are invertible, corresponding to strings of generalized $X$ and $Z$ operators. Thus, the presence of non-invertible symmetries is another feature which distinguishes quantum doubles on non-Abelian $G$.

\begin{rmk}[Non-invertible symmetries]\label{rmk:noninvertible}
    For a smooth quantum double code on CW complex $\Sigma$, pick any non-contractible loop $\gamma=g_{i_1}\cdots g_{i_p}$.  Analogously to \eqref{eq:Zlogicalsmooth}, we can define $Z$-logical operators violating the Knill-Laflamme condition~\eqref{eq:KL} associated to holonomy around $\gamma$: 
    \begin{equation}
   Z_{\gamma}^{\Gamma}|\mathbf{g}\rangle:=\mathrm{Tr}\left[Z^{\Gamma}(g_{i_{1}})\cdots Z^{\Gamma}(g_{i_{p}})\right]|\mathbf{g}\rangle~,
    \end{equation} 
    which form a $\mathrm{Rep}(G)$ fusion category \cite{Fechisin_2025} \begin{equation} \label{eq:fusioncategory}
        Z^\Gamma_\gamma Z^{\Gamma^\prime}_{\gamma}=\sum_{\Gamma^{\prime\prime}\in \mathrm{Rep}(G)} N^{\Gamma^{\prime\prime}}_{\Gamma \Gamma^\prime } Z^{\Gamma^{\prime\prime}}_\gamma
    \end{equation}
    with coefficients $N^{\Gamma^{\prime\prime}}_{\Gamma\Gamma^\prime}$ denoting the multiplicity of irrep $\Gamma^{\prime\prime}$ in the composite irrep $\Gamma \otimes \Gamma^\prime$.  
    It is manifest from the form of \eqref{eq:fusioncategory} that the $Z^\Gamma_\gamma$ are not unitary operators when $G$ is non-Abelian, so they are called non-invertible symmetries in the physics literature \cite{schafernameki2023ictplecturesnoninvertiblegeneralized, shao2024whatsundonetasilectures, Luo_2024, thorngren2019fusioncategorysymmetryi, thorngren2021fusioncategorysymmetryii, Choi_2022, Bhardwaj_2022, zhang2023anomalies11dcategoricalsymmetries}.   
    
     There are (at least) two reasons why these non-invertible symmetries are interesting.  Firstly, each term in the Hamiltonian $H$ commutes with $Z^\Gamma_\gamma$, so the symmetry is enforced locally.  Secondly, there is generically no other family of unitary operators acting on as few qudits as $Z^\Gamma_\gamma$ that represents such a symmetry.  It is therefore more natural to talk about the non-invertible symmetry $Z^\Gamma_\gamma$ than to artificially build unitary operators that take one codeword into another.
\end{rmk}

\subsection{Rate and distance bounds on CW quantum doubles}

Of particular interest in coding theory are (families of) asymptotically good codes where the \textit{code rate} $k/n$ and relative distances, $d_X/n$ and $d_Z/n$, remain separated from zero as $n\rightarrow\infty$.  
Abelian CSS codes of this kind have been found \cite{panteleev2022asymptotically,leverrier2022quantum,dinur2023good}, and so by definition group qudit codes can also be asymptotically good.  
We are interested in designing group CSS codes on \(n\) group-qudits (see Def.~\ref{def:CSS}) with effective number of logical qudits $k$ (a.k.a. code log-dimension), and with distances $d_X$ and $d_Z$ as large as possible, for more generic groups $G$.
Here, we derive results on optimal code parameters for quantum double codes on CW complexes that we developed in Sec.~\ref{sec:QD}.  
In Cor.~\ref{cor:dZloop}, we remarked that there is a connection between the $Z$-distance of a QD code on a CW complex and the systole, or shortest-length cycle in the graph. To illustrate how this constrains $d_Z$, we introduce the following Moore bound, proven in Ref.~\cite{alon}.

 \begin{thm}[Moore bound \cite{alon}] \label{thm:alon}
        Given a graph with $N$ vertices, where every vertex has degree $\ge 2$, if the average degree $K$ obeys $K>2$, then the size of the smallest cycle in the graph (the girth $g$) obeys \begin{equation}
            g < \frac{2}{\log (K-1)}\log \left(1+\frac{K-2}{K}N \right)+2. \label{eq:alon}
        \end{equation}
    \end{thm}
While this is proven in Ref.~\cite{alon} for arbitrary graphs, we will give a short illustrative proof for regular graphs here.
\begin{proof}
Consider an $K$-regular graph with $K > 2$ and girth $g$. The key observation is that given any vertex $v$, a ball $B$ of radius $\lceil g/2 \rceil - 1$ around $v$ is acyclic, so $B$ is a $K$-regular tree. Therefore, the volume of this ball is
\begin{align}
N > |B| = K\sum_{i=1}^{g/2-2}(K-1)^i = K\frac{(K-1)^{g/2-1}-1}{K-2},
\end{align}
which followed by elementary counting and summing the exponential series.
Inverting this inequality for $g$ gives the desired bound.
\end{proof}

It has long been known that there are graphs which saturate this bound, called Moore graphs \cite{erdossachs}. For a particularly simple and explicit construction, see Ref.~\cite{margulis}.
Note that the way we define the code rate [Eq.~\eqref{eq:rate_definition}] is $k = \log_{|G|}\dim(\mathcal C)$, so for this discussion we will assume that $|G|$ is not allowed to vary with $n$.
With this caveat in mind, we can state the following results.

\begin{thm}\label{thm:bounds}
    A generalized quantum double code with $n$ physical $G$-qudits, log-dimension $k = \mathrm O(n)$, and distances \(d_Z\) and \(d_X\) satisfies
    \begin{equation}
        d_Z \le 2+\frac{2\log (n)}{\log(1+2(k-1)/n)}.
        \label{Eq: dZ bound}
    \end{equation}
    \label{thm:dZ bound}
    In the limit of large $n$, this bound becomes asymptotically
    \begin{align}
    kd_Z \le \mathrm{\Theta}(n\log(n)).
    \end{align}
\end{thm}

\begin{proof}

Let $\Sigma$ be the CW complex associated to the code as in Cor.~\ref{cor:dZloop}, such that every non-contractible loop in $\Sigma$ supports a logical-$Z$ operator.
We construct a subcomplex $\Sigma_*$ by taking the union of a spanning tree with all of the edges in its complement 
which take non-trivial values in some gauge class $\ket{[\vb g]}$~\eqref{eq:basis-codewords} after gauge-fixing the edges of the spanning tree (e.g. by removing all the blue edges in Fig.~\ref{fig:torus}). Since none of the removed edges represent a non-contractible loop, $\Sigma$ and $\Sigma_*$ have the same girth $g$, and by Cor.~\ref{cor:dZloop}, $d_Z = g$.
Both complexes $\Sigma, \Sigma_*$ contain a spanning tree of $\Sigma$, so they have the same vertex set.  On the other hand, clearly any loop in $\Sigma_*$ can only pass through a vertex connected to a single edge $e$ by traversing the path $\cdots ee^{-1}\cdots  $, which does not increase $d_Z$, so we may remove all vertices of degree 1 and restrict to a subcomplex $\widetilde\Sigma_*$ in which every vertex has at least degree 2. Let $V$ be the number of vertices in $\widetilde\Sigma_*$, $E$ the number of edges, and $L = E - V +1$ be the number of edges not in the original spanning tree (i.e. the number of loops).

    Now we may apply Theorem \ref{thm:alon} as soon as we bound the average degree of $\widetilde\Sigma_*$:
    \begin{equation}
        K = \frac{1}{V}\sum_{v\in \widetilde\Sigma_*} K_v = \frac{2 E}{V} = 2 + 2\frac{L-1}{V}, \label{eq:Kbound}
    \end{equation}
    where and $K_v \ge 2$ denotes the degree of each vertex in $\widetilde\Sigma_*$.  
    We have used the fact that $\sum_{v \in \widetilde \Sigma_*}K_v = 2 E$, which is valid for a general graph as can be seen inductively.

    Now, by Prop.~\ref{prop:codespace_rough}, we have
    \begin{align}
    \dim \mathcal C \leq |\Hom(\pi_1(\Sigma), G)|  \leq |G|^{L},
    \end{align}
    where the first inequality arises from considering the quotient by $\Ad_G$ in the smooth case, and the last inequality arises from the fact that different loops may be homotopic to each other. From this, we obtain $k \leq L$.
    Moreover, from $K\ge 2$ \eqref{eq:Kbound}, we obtain $V \le E \le n$, where the last inequality comes from the fact that a subset of the $n$ qubits are edges on $\widetilde\Sigma_*$.  Hence  \begin{equation}
        K \ge  2 + \frac{2(\lceil k\rceil -1)}{V} \ge 2+\frac{2(k -1)}{n}~,
    \end{equation} 
    and the code distance obeys \begin{equation}
        d_Z = g< 2+\frac{2\log (L)}{\log (1+\frac{2(L-1)}{V})} < 2 + \frac{2\log(n)}{\log(1+\frac{2(k-1)}{n})}~. \label{eq:dZboundinproof} 
    \end{equation}
    where we used $L < n$ and $(L-1)/V = K/2-1 \geq \frac{(k-1)}{n}$.
\end{proof}

Let us comment briefly on the optimality of our bound. Asymptotically, our bound reduces to
\begin{align}
kd_Z \lesssim n\log(n) .
\end{align}
When $k = \mathrm \Omega(n)$, then this bound becomes tight. However, in the case of surface codes and $k < \mathrm O(n)$, this is not as tight as the bound in Ref.~\cite{Delfosse_2013},
\begin{align}
kd_Z^2 \lesssim (\log k)^2 n.
\label{eq:Zboundsurface}
\end{align}
This does appear to be a fundamental limit of the Moore bound argument which leads to logarithmic sub-optimality in the bound. Take the toric code on an $\ell \times \ell$ torus for example. Then $V = \ell^2$ and $L \leq 2\ell$, as shown in Fig.~\ref{fig:torus}. We can choose any $L \leq 2 \ell$, and our bound becomes, asymptotically,
\begin{align}
d_Z \lesssim \frac{V}{L}\log(L) \lesssim \ell\log(\ell) \sim \sqrt{n}\log(n).
\end{align}
We can see that no valid choice of $L$ will compensate for the $\log(n)$, so our bound is best used for non-manifold CW complexes. However, we note that the bound from Ref.~\cite{Delfosse_2013} similarly implies
\begin{align}
kd_X^2 \lesssim (\log k)^2 n
\label{eq:Xboundsurface}
\end{align}
for any QD code on a 2D surface, and this is optimal for constant-rate codes because it is saturated by hyperbolic surface codes \cite{Breuckmann_2016, Breuckmann_2017}.
This is ultimately due to the duality between the systole and the co-systole. In the non-manifold setting, this duality is lost, and $d_X$ can be much larger. To illustrate this, we show the optimality of our bounds in the constant-rate case:

\begin{thm}[Optimal code parameters] \label{thm:optimal}
    There exist quantum double code families (Definition \ref{def:generalizedQDsmooth}) on $n$ qudits with 
    \begin{subequations}\label{eq:goodquantumdoubleproperties}
        \begin{align}
            k &= \mathrm{\Theta}(n), \\
            d_X &= \mathrm{\Theta}(n), \\
            d_Z &= \mathrm{\Theta}(\log n).
        \end{align}
    \end{subequations}
\end{thm}

\begin{proof}
A constructive proof is given in Appendix \ref{app:optimal}.  The construction relies on a careful ``lift" of a classical code which is asymptotically good ($k, d_X = \mathrm{\Theta}(n)$) to a CW complex whose 1D subcomplex has a large girth.  Qubits are assigned in a specific way to ensure that $d_Z\sim \log n$.
We note that the construction does \emph{not} give an LDPC code -- typical $Z$-checks have $\mathrm{\Theta}(\log n)$ weight.
\end{proof}

 We note that the scaling $d_X = \mathrm{\Theta}(n)$ far exceeds what is possible with a surface code.

 \section{Phases of quantum matter as CW quantum doubles}\label{sec:matter}

The past decade has seen a fruitful exchange of ideas between the theory of quantum codes and the classification of phases of gapped quantum matter \cite{Dennis_2002, Kitaev_2003, Levin_2005, Bombin_2010, Bravyi_2010p, Chen_2010, Haah_2011, Terhal_2015, Yoshida_2015, Wen_2017, Barkeshli_2019, lavasani2024stabilityklocalquantumphases, 7x71-8j7k, Yin_2025}.   Indeed, many interesting phases of quantum matter in two dimensions have already been analyzed via their relationships to quantum double codes on two-dimensional surfaces \cite{Kitaev_2003, Bombin_2008, Cui:2019lvb}. Some examples from the literature that we will not discuss explicitly include 2D models~\cite{SWSSB_nonabelian,Brell_2015}, fracton-like derivatives \cite{tantivasadakarn2021non}, the 3D quantum triple models \cite{Moradi_2015, Delcamp_2017}, more general (untwisted) Dijkgraaf-Witten gauge theories \cite{dijkgraaf1990topological}.  Still, these constructions utilize only the 0D, 1D, and 2D structures (i.e., 2D skeletons) of higher-dimensional cellulations.
Due to the generality of our Definition~\ref{def:CSS}, these models are closely related to our group CSS construction.

In this section, we will point out that a number of \emph{other} phases of quantum matter with group qudits can also be understood as quantum double codes on CW complexes with ghost vertices.  We emphasize that this relationship is not restricted to two-dimensional quantum phases---our first example will be of a one-dimensional phase.

First, we introduce an alternative, but equivalent, set of \(Z\)-checks to define physics-based codes. 
This more common set of checks can be constructed using the group's irreducible representations, or irreps, \(\Gamma \in \text{Rep} (G)\). Define the corresponding operators
\begin{equation}
        Z^\Gamma_{i, \alpha\beta}|\mathbf{g}\rangle := \Gamma_{\alpha\beta}(g_i)|\mathbf{g}\rangle,
\end{equation}
    where $\Gamma_{\alpha\beta}(g_i)$ is (component-wise) the $d_\Gamma \times
    d_\Gamma$ matrix representing group element $g_i$ at site \(i\). These operators form a basis for functions on the group. 
The group CSS \(Z^K\) checks can be expressed in terms of \(Z^{\Gamma}\) checks as
\begin{subequations}
\begin{align}
Z_{q_{1}\cdots q_{j}}^{K}&=\sum_{k\in K}\sum_{\mathbf{g}\in G^{n}}\mathbb{I}(g_{q_{1}}\cdots g_{q_{j}}=k)|\mathbf{g}\rangle\langle\mathbf{g}|\\&=\sum_{k\in K}\sum_{\mathbf{g}\in G^{n}}\sum_{\Gamma\in\mathrm{Rep}(G)}{\textstyle \frac{d_{\Gamma}}{|G|}}\mathrm{Tr}[\Gamma(k^{-1})\Gamma(g_{q_{1}}\cdots g_{q_{j}})]|\mathbf{g}\rangle\langle\mathbf{g}|\\
&=\sum_{\Gamma\in\mathrm{Rep}(G)}{\textstyle \frac{d_{\Gamma}}{|G|/|K|}}\mathrm{Tr}(\Pi^{K}Z_{q_{1}}^{\Gamma}\cdots Z_{q_{j}}^{\Gamma})~.
\end{align}
\end{subequations}
Above, we use the group delta function \cite[Eq.~(21)]{albert2021spinchainsdefectsquantum} in the second step, the definition \(\Pi^{K}=\frac{1}{|K|}\sum_{k\in K}\Gamma(k)\) of the projection of the \(G\)-irrep \(\Gamma\) onto the trivial \(K\)-irrep in the third step, and define ``Tr'' to be the trace over internal irrep indices \(\alpha\beta\).
A reverse formula is also possible for certain \(Z^K\)'s \cite[Eq.~(29b)]{albert2021spinchainsdefectsquantum}.
An advantage of using this presentation is that it is clear there is a
matrix product operator representation of the $Z$-checks with \(\alpha\beta\) being the bond indices \cite{Fechisin_2025}, making them
``almost" transversal. 
Abelian \(G\) have only one-dimensional irreps, in which case \(Z^{\Gamma}\) become unitary \(Z\)-type Pauli operators, while \(Z^K\) remain projections onto certain qudit states.

In a physics context, it is often useful to interpret the group-CSS codespace $\mathcal{C}$ as the ground state subspace associated to a Hamiltonian $H$ (cf. \cite[Eq.~(137)]{Albert_2020} for the group GKP case and \cite[Def. 3.5]{gottesman2024surviving}),
\begin{equation}
        \mathcal H = \sum_{S\in \mathcal{S}_X\cup \mathcal{S}_Z} \left(1-\frac{S+S^\dagger}{2}\right) \label{eq:codeH}~,
    \end{equation}
where we have defined $\mathcal H$ such that its ground state energy is 0.
This can be proven as follows. 
If $|\psi\rangle \in \mathcal{C}$, clearly $\frac{1}{2}(S+S^\dagger)|\psi\rangle = |\psi\rangle$.  
Moreover, we can easily see that the maximal eigenvalue of $S+S^\dagger$ is 1 for each check. 
Therefore, if $|\psi\rangle \in \mathcal{C}$, it minimizes each term in $H$ individually, which is the definition of it being is a frustration-free ground state.

\subsection{Spontaneous symmetry breaking}\label{Sec: SSB}

As a warm-up example, consider $n$ qudits arranged in a one-dimensional chain with periodic boundary conditions, described by the Hamiltonian \cite{drouffe1979lattice,munk2018dyonic}
\begin{equation}\label{eq:ising}
   \mathcal H = - \frac{1}{|G|} \sum_{i=1}^{n-1}\sum_{\Gamma \in \mathrm{Rep}(G)}d_\Gamma \mathrm{Tr}\left(Z_i^{\Gamma\dagger} Z^\Gamma_{i+1}\right).
\end{equation}
This is a member of the family of \(G \downarrow H\) symmetry-breaking models for any subgroup \(H \leq G\) and is a lattice version of the principal chiral model Hamiltonian (in which case \(Z^{\Gamma}\) become group-valued fields) \cite{albert2021spinchainsdefectsquantum}.
The general quantum Ising-type model---called the flux ladder \cite{munk2018dyonic,albert2021spinchainsdefectsquantum}---has a dual model \cite{warman2025categorical,chung2025spontaneously} under a generalized group-qudit Kramers-Wannier transformation \cite{lootens2023dualities,Fechisin_2025} that hosts \(\text{Rep}(G)\)-symmetry breaking.

This model corresponds to a quantum double code on the following CW complex:  we have two ghost vertices $a$ and $b$ and every edge starts at $a$ and ends at $b$.   A 2-cell is glued between edges $i$ and $i+1$ for each $i$.   The CW complex is simply connected as it is topologically equivalent to a disk (ignoring the ghost vertices).  A similar interpretation of this model as a gauge theory for Abelian \(G\) can be found in earlier work \cite{Milsted_2016}.  

Notice that since there are two ghost vertices, the model has a global $G\times G$ symmetry.   Moreover, Proposition \ref{prop:codespace_rough} tells us that the codespace, i.e. the ground states of $H$ in \eqref{eq:ising}, is $\mathbb{C}[G]$.  In this example, it is easy enough to write down the ground states explicitly: \begin{equation}
    \mathcal{C} = \mathrm{span}\lbrace |g,g,\ldots, g\rangle \rbrace_{g\in G}.
\end{equation}
This is a ``group qudit repetition code", as indeed the checks simply specify that $g_i^{-1}g_{i+1}=1$ for each neighboring pair of qudits.   We further notice that the global $G\times G$ symmetry that we identified through the presence of two ghost 0-cells is spontaneously broken to $G$.  While one does not need the CW complex picture to understand the ground states of this Hamiltonian, it is elegant because it allows us to immediately discover that the same symmetry-breaking pattern arises in a more non-trivial model, which we introduce next.

\begin{figure}
\centering
\def\svgwidth{.75\textwidth}
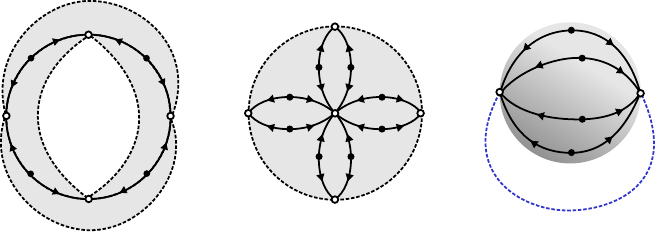
\caption{Example of spontaneous symmetry breaking of $G \times G$ symmetry in a 1D lattice system consisting of 4 qudits with periodic boundary conditions. The gray shaded regions illustrate 2-cells. The dotted lines show where ghost cells are identified with each other, as forced by the $Z$-checks (in the sense of Prop.~\ref{prop:roughQDonCW}). Solid dots correspond to $X$-checks, arrows show the locations of qudits, and open dots show ghost cells. \textbf{Left:} Illustration of the complex corresponding to $\mathcal H_{\mathrm{SSB}}$ [Eq.~\eqref{Eq: SSB}]. \textbf{Center:}  The complex looks like a disk with a ghost cell corresponding to the even sublattice at the center and the ghost cell corresponding to the odd sublattice on the boundary. A disk with its boundary identified to a point is a sphere. \textbf{Right:} The complex in the center is homeomorphic to a sphere with ghost cells at antipodal points. The blue dotted line illustrates the complex $\Sigma_{\mathrm {SSB}}/Y \cong S^1 \vee S^2$, with fundamental group $\pi_1(\Sigma_{\mathrm{SSB}}/Y) \cong \Z$, where $Y$ is the set of ghost cells.}
\label{Fig: SSB}
\end{figure}

Now consider a 1D chain with periodic boundaries and a $G$-qudit on each vertex, with the following Hamiltonian \cite{Fechisin_2025}:
\begin{align}
\mathcal H_{\mathrm{SSB}}=-\frac{1}{|G|}\sum_{i\in\mathrm{odd}}\left(\sum_{\Gamma\in\mathrm{Rep}(G)}d_{\Gamma}\mathrm{Tr}\left[Z_{i}^{\Gamma\dag}Z_{i+1}^{\Gamma}Z_{i+3}^{\Gamma\dag}Z_{i+2}^{\Gamma}\right]+\sum_{g\in G}\overrightarrow{X}_{i}^{g}\overrightarrow{X}_{i+1}^{g}\right).
\label{Eq: SSB}
\end{align}
The $G\times G$ symmetry is generated by  the following pair of symmetry operators:   \begin{equation}
A_{\mathrm{even/odd}}^g  := \prod_{i\in \text{even/odd}}\overleftarrow{X}^g_i .  
\end{equation}
Note that this symmetry is invertible as the above operators are unitary.
One finds $|G|$ degenerate ground states labeled by a single group element:
\begin{align} \label{eq:Gfolddegenerate}
\mathcal{C} = \mathrm{span}\left\lbrace |h_{\mathrm{L}}\rangle :=\sum_{g_i\in G}\ket{g_1, g_1h, g_3, g_3h, \cdots} \right\rbrace_{h \in G}.
\end{align}
One can see that $A_{\mathrm{even}}^g |h_{\mathrm{L}}\rangle := |(hg^{-1})_{\mathrm{L}}\rangle $ while $A_{\mathrm{odd}}^g |h_{\mathrm{L}}\rangle := |(gh)_{\mathrm{L}}\rangle $.

As before, we clearly notice that $\mathcal H_{\mathrm{SSB}}$ [Eq.~\eqref{Eq: SSB}] takes the form of a (rough) quantum double code, so we can look for the corresponding CW complex. As illustrated in Fig.~\ref{Fig: SSB}, we put the $G$-qudits on the edges, and $X$ checks in $\mathcal H_{\text{SSB}}$ are depicted by black vertices. We notice that for each qudit, there is only one $X$ check. To ensure compatibility with the $Z$-checks, we can add exactly two ghost 0-cells, one associated with all the open ends of odd-numbered qudits, and the other with even-numbered qudits, as illustrated by blue empty dots in Fig.~\ref{Fig: SSB}. Notice that, as in the SPT examples, the global $G\times G$ symmetry is associated with multiplication by group elements on each of the two ghost vertices independently. 
One can then see that the resulting CW complex $\Sigma_{\mathrm{SSB}}$ is homeomorphic to a sphere with ghost cells at antipodal points. Taking the quotient of $\Sigma_{\mathrm{SSB}}$ with the set of ghost cells $Y$ produces a complex homeomorphic to $S^1 \vee S^2$, with fundamental group $\pi_1(\Sigma_{\mathrm{SSB}}/Y) \cong \Z$; hence $\mathrm{Hom}(\pi_1(\Sigma_{\mathrm{SSB}}/Y),G)\cong \mathrm{Hom}(\Z,G) = G$,
which explains the $|G|$-fold degenerate ground state \eqref{eq:Gfolddegenerate}.

More generally, let us consider any model whose Hamiltonian takes the ``quantum double" form as in \eqref{Eq: SSB}.  If this model has a $G^R$ global symmetry for some integer $R\ge 1$, this means that we can glue exactly $R$ independent ghost 0-cells to the corresponding CW complex.   Proposition \ref{prop:codespace_rough} then shows us that the resulting ground states will be at least $|G|^{R-1}$-fold degenerate.  All examples of SPT phases and SSB phases in this section have and will saturate this general bound.

\subsection{Non-invertible symmetry-protected topological phases}
\label{Sec: non-inv SPT}

We first discuss the 1D non-Abelian $G$-cluster state discussed in Ref. \cite{Fechisin_2025} as an example of non-invertible symmetry-protected topological (SPT) states.   Given a 1D chain with periodic boundaries and a $G$-qudit on each vertex, the 1D non-Abelian $G$-cluster state $\ket{\mathcal{C}}$ is the unique ground state of the following stabilizer Hamiltonian, 
\begin{align}
\mathcal H_\mathcal{C}=-\frac{1}{|G|}\sum_{i\in\mathrm{odd}}\left(\sum_{\Gamma\in\mathrm{Rep}(G)}d_\Gamma \mathrm{Tr}\left[Z_{i}^{\Gamma\dag} Z_{i+1}^\Gamma Z_{i+2}^\Gamma\right]+\sum_{g\in G}\overleftarrow{X}_{i+1}^g\overrightarrow{X}_{i+2}^{g}\overrightarrow{X}_{i+3}^g\right).
\label{Eq: 1D cluster}
\end{align}
The $G$-cluster state $\ket{\mathcal{C}}$ is a \emph{non-invertible} SPT state, whose symmetry is $G\times\mathrm{Rep}(G)$, with the following symmetry operators,
\begin{align}
A_g:=\prod_{i\in\mathrm{odd}}\overleftarrow{X}_{i}^g,~~B_\Gamma:=\mathrm{Tr}\left[\prod_{i\in\mathrm{even}}Z_{i}^\Gamma\right].
\end{align}
The symmetry corresponding to $B_\Gamma$ is non-invertible \cite{Fechisin_2025} (Remark \ref{rmk:noninvertible}).  

\begin{figure}
\centering
\def\svgwidth{.6\textwidth}
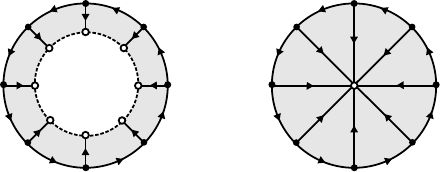
\caption{1D non-Abelian $G$-cluster state $\ket{\mathcal{C}}$ embedded onto a CW complex with ghost checks shown by open circles. Each black dot depicts an $X$ check, plaquettes are shaded gray, and qudits are associated to 1-cells. \textbf{Left:} CW complex representation of $\mathcal H_{\mathcal C}$ [Eq.~\eqref{Eq: 1D cluster}]. The dotted lines show where ghost checks are identified to a single ghost vertex. The resemblance to the boundary in Fig.~\ref{fig:rough_smooth_bdy} is apparent. \textbf{Right:} After identifying the removed checks with a single ghost cell, the complex $\Sigma_{\mathcal C}$ is homeomorphic to a 2-disk, which is contractible, implying that $\ket{\mathcal C}$ is the unique ground state.} 
\label{Fig: 1D cluster}
\end{figure}

The Hamiltonian $\mathcal H_{\mathcal C}$ is also that of a quantum double code. Indeed, we will now show that $\ket{\mathcal C}$ is the unique codeword of a quantum double code on a connected 2D CW complex. Using the construction of Prop.~\ref{prop:roughQDonCW}, the corresponding CW complex $\Sigma_{\mathcal C}$ is illustrated in Fig.~\ref{Fig: 1D cluster}, and we can see that $\Sigma_{\mathcal C} \cong D^2$ (the 2-disk).
Since $D^2$ is contractible, the fundamental group of $\Sigma_{\mathcal C}$ is trivial, so the non-Abelian $G$-cluster state $\ket{\mathcal C}$ is the unique ground state of this model.

We can also neatly see the relationship between these non-Abelian cluster states and edge theories in quantum double models with boundaries.  
Locally, the CW complex in the left panel of Fig.~\ref{Fig: 1D cluster} is identical to the rough boundary depicted in Fig.~\ref{fig:rough_smooth_bdy}.
This ``bulk-boundary'' relationship between the rough quantum double model and the non-Abelian cluster state can also be understood in the context of the recently proposed \textit{symmetry topological field theory (symTFT)} \cite{Ji_2020, Gaiotto_2021, Apruzzi_2023, freed2024topologicalsymmetryquantumfield, Bhardwaj_2025, huang2024fermionicquantumcriticalitylens}.

Next, we consider a more sophisticated example: the 2D non-Abelian $G$-cluster state on a Lieb lattice.
We put a $G$-qudit on each vertex and each edge, and initialize them in the following product state,
\begin{align}
\ket{\mathcal{C}_2^0}=\ket{+}\otimes\ket{1}\otimes\ket{+}\otimes\ket{1}\otimes\cdots,
\label{Eq: 2D product}
\end{align}
where the vertex qudits are initialized in the state $\ket{+}\equiv \sum_g\ket{g}$, and the edge qudits are initialized in the state $\ket{1}$, with \(1\) the identity group element. 

We define the generalized CNOT gate for $G$-qudits. Consider a two-qudit state $\ket{g_i, g_j}$, two generalized CNOT gates are defined as
\begin{align}
\overrightarrow{CX}_{(i,j)}\ket{g_i, g_j}:=\ket{g_i, g_ig_j},~~\overleftarrow{CX}_{(i,j)}\ket{g_i, g_j}=\ket{g_i g_j^{-1}, g_j}. 
\end{align}
Then we apply these generalized CNOT gates for every nearest pair of qudits on the product state \eqref{Eq: 2D product}, using the vertex qudits to control the edge qudits. The resulting 2D state $\ket{\mathcal{C}_2}=\prod_{v}\prod_{e\ni v}\overrightarrow{CX}_{(v,e)}\ket{\mathcal{C}_2^0}$ is the unique ground state of the following stabilizer Hamiltonian: 
\begin{align}
\begin{tikzpicture}[]
\tikzstyle{sergio}=[rectangle,draw=none]
\draw[very thick] (0,-1) -- (0,1);
\draw[very thick] (-1,0) -- (1,0);
\path (0,0) node [style=sergio]{\color{red}$\overrightarrow{X}_g$};
\path (0,-1) node [style=sergio]{\color{red}$\overrightarrow{X}_g$};
\path (1,0) node [style=sergio]{\color{red}$\overrightarrow{X}_g$};
\path (-1,0) node [style=sergio]{\color{red}$\overleftarrow{X}_g$};
\path (0,1) node [style=sergio]{\color{red}$\overleftarrow{X}_g$};
\path (-1.8,-0.1) node [style=sergio]{\Large$\sum\limits_{g\in G}$};
\path (-2.5,-0) node [style=sergio]{\Bigg(};
\path (-4,-0.1) node [style=sergio]{\Large $\mathcal H=-\frac{1}{|G|}\sum\limits_{i\in v}$};
\path (1.8,-0) node [style=sergio]{\Large$+$};
\path (2.9,-0.1) node [style=sergio]{\Large$\sum\limits_{\Gamma\in\mathrm{Rep}(G)}$};
\draw[very thick] (5,-1) -- (5,1);
\draw[very thick] (4.8,-1) -- (5.2,-1);
\draw[very thick] (4.8,1) -- (5.2,1);
\path (5,1) node [style=sergio]{\color{red}$Z_{\Gamma}^\dag$};
\path (5,0) node [style=sergio]{\color{red}$Z_{\Gamma}$};
\path (5,-1) node [style=sergio]{\color{red}$Z_{\Gamma}$};
\path (4.25,0) node [style=sergio]{$d_\Gamma$~Tr};
\path (6,0) node [style=sergio]{$+~d_\Gamma$~Tr};
\draw[very thick] (7,0) -- (9,0);
\draw[very thick] (7,0.2) -- (7,-0.2);
\draw[very thick] (9,0.2) -- (9,-0.2);
\path (7,0) node [style=sergio]{\color{red}$Z_{\Gamma}^\dag$};
\path (8,0) node [style=sergio]{\color{red}$Z_{\Gamma}$};
\path (9,0) node [style=sergio]{\color{red}$Z_{\Gamma}$};
\path (9.5,-0) node [style=sergio]{\Bigg)};
\end{tikzpicture}
\label{Eq: 2D cluster}
\end{align}
where the sum is taken over all vertices. In Appendix \ref{App: 2D cluster}, we give an explicit tensor network representation of the ground state of this Hamiltonian. The 2D non-Abelian cluster state is a non-invertible SPT state protected by a global $G$ symmetry and a 1-form $\mathrm{Rep}(G)$ symmetry, with the following symmetry operators, 
\begin{align}
A_g:=\prod_{i\in v}\overleftarrow{X}_{i}^g,\qquad\qquad B_\Gamma:=\mathrm{Tr}\left[\prod_{i\in \mathrm{loop}}Z_{i}^\Gamma\right],
\label{Eq: 2D cluster symmetry}
\end{align}
where the global $G$ symmetry operator $A_g$ is a product of $X$ over all vertex qudits, and the 1-form $\mathrm{Rep}(G)$ symmetry operator $B_\Gamma$ is a product of $Z$ over all edge qudits on arbitrary closed loops of the graph.

\begin{figure}
\centering
\def\svgwidth{.65\textwidth}
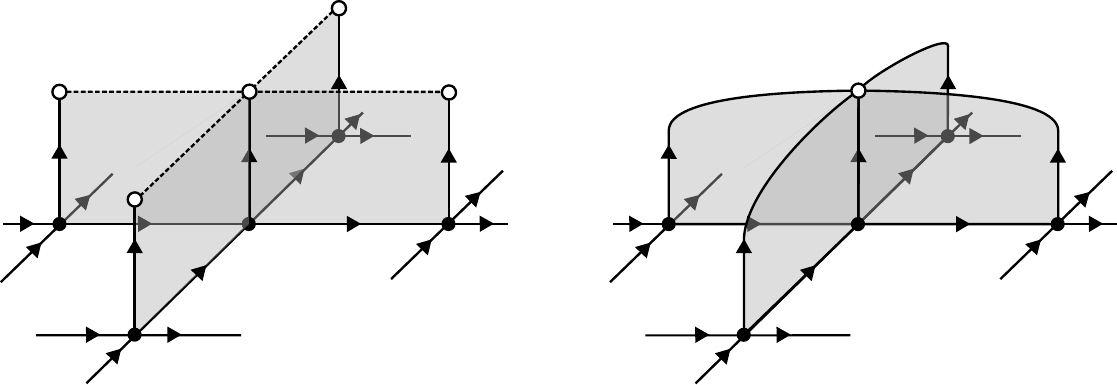
\caption{2D non-Abelian $G$-cluster state $\ket{\mathcal{C}_2}$ embedded onto a CW complex with a ghost check. \textbf{Left:} Each black dot depicts an $X$ check,  each arrow depicts a $G$-qudit on an oriented 1-cell, and each gray plaquette depicts a $Z$ check. 
The open circles represent ghost $X$-checks, and the dotted lines represent where these checks are identified together.
\textbf{Right:} After properly identifying the ghost checks to one ghost vertex, we can see the $\mathrm{Rep}(G)$ 1-form symmetry manifest, because all closed loops on the plaquettes of the lattice become topologically trivial.}
\label{Fig: 2D cluster}
\end{figure}

\begin{figure}
\begin{tikzpicture}[use Hobby shortcut,scale=0.7]
\tikzstyle{sergio}=[rectangle,draw=none]
\draw[very thick, color=red] (-1,1) -- (1,1) -- (0,0) -- (-2,0) -- cycle;
\filldraw[fill=black, draw=black] (1,4) circle (3pt);
\draw[very thick] (1,1) -- (1,4);
\draw[very thick] (-1,1) .. (-0.9,2) .. (1,4);
\draw[very thick] (0,0) .. (0.05,1) .. (1,4);
\draw[very thick] (-2,0) .. (-1.9,1) .. (1,4);
\path (2,2) node [style=sergio]{\huge$\Rightarrow$};
\draw[very thick, color=red] (4,1) -- (6,1) -- (5,0) -- (3,0) -- cycle;
\draw[very thick] (4,1) -- (6,1);
\draw[very thick, color=red] (1+5,1) -- (1+5,4);
\draw[very thick, color=red] (-1+5,1) .. (-0.9+5,2) .. (1+5,4);
\draw[very thick] (0+5,0) .. (0.05+5,1) .. (1+5,4);
\draw[very thick] (-2+5,0) .. (-1.9+5,1) .. (1+5,4);
\filldraw[fill=black, draw=black] (6,4) circle (3pt);
\path (7,2) node [style=sergio]{\huge$\Rightarrow$};
\draw[very thick, color=red] (9,1) -- (11,1) -- (10,0) -- (8,0) -- cycle;
\draw[very thick] (9,1) -- (11,1);
\draw[very thick] (10,0) -- (11,1);
\draw[very thick, color=black] (1+10,1) -- (1+10,4);
\draw[very thick, color=red] (-1+10,1) .. (-0.9+10,2) .. (1+10,4);
\draw[very thick, color=red] (0+10,0) .. (0.05+10,1) .. (1+10,4);
\draw[very thick, color=black] (-2+10,0) .. (-1.9+10,1) .. (1+10,4);
\filldraw[fill=black, draw=black] (11,4) circle (3pt);
\path (12,2) node [style=sergio]{\huge$\Rightarrow$};
\draw[very thick, color=black] (9+5,1) -- (11+5,1) -- (10+5,0) -- (8+5,0) -- cycle;
\draw[very thick, color=red] (14,1) -- (13,0);
\draw[very thick, color=black] (1+15,1) -- (1+15,4);
\draw[very thick, color=red] (-1+15,1) .. (-0.9+15,2) .. (1+15,4);
\draw[very thick] (0+15,0) .. (0.05+15,1) .. (1+15,4);
\draw[very thick, color=red] (-2+15,0) .. (-1.9+15,1) .. (1+15,4);
\filldraw[fill=black, draw=black] (16,4) circle (3pt);
\path (17,2) node [style=sergio]{\huge$\Rightarrow$};
\draw[very thick, color=black] (9+10,1) -- (11+10,1) -- (10+10,0) -- (8+10,0) -- cycle;
\draw[very thick, color=black] (14+5,1) -- (13+5,0);
\draw[very thick, color=black] (1+20,1) -- (1+20,4);
\draw[very thick, color=black] (-1+20,1) .. (-0.9+20,2) .. (1+20,4);
\draw[very thick] (0+20,0) .. (0.05+20,1) .. (1+20,4);
\draw[very thick, color=black] (-2+20,0) .. (-1.9+20,1) .. (1+20,4);
\filldraw[fill=red, draw=red] (21,4) circle (3pt);
\end{tikzpicture}
\caption{Illustration of the 1-form $\mathrm{Rep}(G)$ non-invertible symmetry as the contractibility of loops corresponding to the symmetry operators in Eq.~\eqref{Eq: 2D cluster symmetry} on the CW complex.}
\label{fig:1formSSB}
\end{figure}
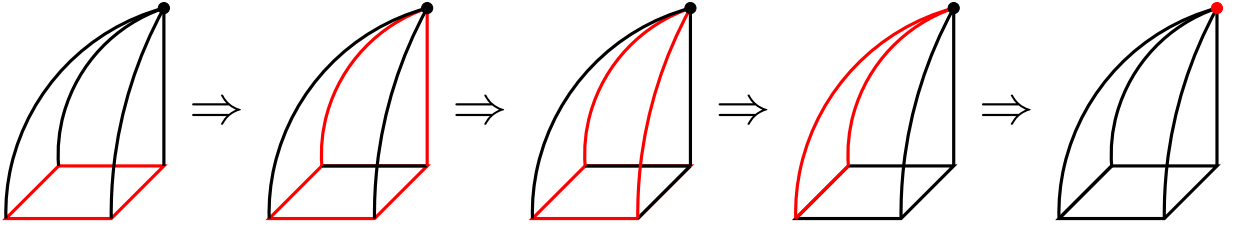

Now we embed $\ket{\mathcal{C}_2}$ onto a CW complex.
This is obtained from the original Lieb lattice by
extending the qudits placed on the vertices of the Lieb lattice vertically, and associating each vertex of the Lieb lattice with a vertex of the CW complex, as
illustrated in Fig.~\ref{Fig: 2D cluster}.
Each qudit remaining on the horizontal plane is involved in exactly two $X$ checks, but each vertical qudit is only invovled in one. Therefore, the new vertical edges terminate in ghost checks, which are identified by the plaquettes into a single ghost vertex to for the CW complex $\Sigma_2$, as shown in the right panel of Fig.~\ref{Fig: 2D cluster}. Similarly, Fig.~\ref{fig:1formSSB} illustrates the 1-form non-invertible symmetry as the set of contractible loops on the square lattice formed by the edges of the original Lieb lattice.

Similar to the 1D non-Abelian cluster state, 
we obtain that the fundamental group of $\Sigma_2$ is trivial.  This is easiest to see by noting that $\Sigma_2$ is the (topological) cone of the 1D CW complex corresponding to the square lattice graph, and every cone is a simply connected space.
Equivalently, $\Sigma_2$ is the $2$-skeleton of $D^3$ (the 3-disk), and the fundamental group can only depend on the $2$-skeleton of the complex, and therefore $\pi_1(\Sigma_2) \cong \pi_1(D^3) = \{1\}$.
Thus, the 2D non-Abelian $G$-cluster state can also be treated as a quantum double model embedded onto the CW complex $\Sigma_2$ with a ghost 0-cell, with a unique ground state.

\section{Rigidity of CSS codes for non-Abelian simple groups}\label{sec:rigid}

In the previous section,  we explored a basic extension of a well-known construction of group-qudit codes:  the quantum double model.   Remarkably, we will prove in this section that under seemingly mild assumptions, for certain non-Abelian $G$, these are the \emph{only} group qudit CSS codes that can be built.

To motivate this result, firstly, let us emphasize that clearly any \emph{qubit} CSS code, which we have already seen can be interpreted as a group qudit code with $G=\mathbb{Z}_2$, can be embedded into many other $G$ (cf. \cite[Sec.~VIII.B]{faist2020continuous}).
\begin{eg}\label{eg:silly}
    Let $G$ be a group with a non-trivial element $h\in G$ such that $h^2=1$, so that the subgroup $\langle h\rangle \cong \mathbb{Z}_2$.   
    Now, given 4 qudits, consider $Z$-checks of the form $Z^{\langle h\rangle}_1,\ldots, Z^{\langle h\rangle}_{4},\ldots Z^1_{1234}$ and a single $X$-check,  $X^h_{1234,\emptyset}$, following the notation of \eqref{eq:Xg_shorthand}.   This code is equivalent to the standard 4-qubit CSS code on $\mathbb{Z}_2$, with stabilizer group $X_1X_2X_3X_4$ and $Z_1Z_2Z_3Z_4$.
\end{eg}

This simple example can clearly be generalized to any CSS code on $\mathbb{Z}_2$.   To avoid this and similar qubit-based constructions, it makes sense to demand some additional structure.  
A modest proposal is to study $G$-covariant codes, which we define as:

\begin{defn}[$G$-covariant code] \label{def:covariant}
  A group CSS code on $G$ is covariant if and only if
  $X_{\mathrm{tot}}^g$ is a code automorphism for all $g \in G$, where
  \begin{equation}
        X^g_{\mathrm{tot}} := \prod_{i=1}^n \overrightarrow{X}^g_i\overleftarrow{X}^g_i .
    \end{equation}
By code automorphism, we mean that $X_{\mathrm{tot}}^g\mathcal{S}X_{\mathrm{tot}}^{g^{-1}} = \mathcal{S}$,
    where $\mathcal S$ is the set of stabilizers of the code.
\end{defn}

Note that the property of covariance still encompasses all Abelian CSS codes: If $G$ is Abelian, then $X^g_{\mathrm{tot}}$ is the identity operator.   However, the $G$-covariant constraint is a non-trivial constraint on a group CSS code over non-Abelian $G$. Intuitively, the $G$-covariance condition is intended to force us to ``use" the full structure of the group $G$ when designing the code. While quantum double models are $G$-covariant, Example \ref{eg:silly} is not in general covariant.\footnote{It could be made $G$-covariant if there is a $\mathbb{Z}_2$ subgroup in the center of $G$.} Indeed, our main result Theorem \ref{thm:classify} will show that this $G$-covariance condition can be quite restrictive: subject to reasonable assumptions on the \(Z\)-check weight $w_Z$ and distance $d_Z$, we find that the only family of group CSS codes that are well-defined for an arbitrary group $G$ are quantum double codes on CW complexes.   Before we can state this result, we need to introduce some useful terminology from combinatorial group theory \cite{cohen1989combinatorial}: 

\begin{defn}[Freely trivial words]
\label{def:freely-trivial}
  A \textit{group word} is a function $f:G^{r} \to G^w$ (as sets), where each component of
  $f(g_1, \dots, g_r)$ is $g_i^{\pm 1}$ for some $i$, and where $w$ is called the weight of $f$. For example, $f(g) = (g^{-1}, g,
  g)$ has $r=1$ and $w=3$.  
  A group word $f$ may be \textit{evaluated} in a given $G$ by taking the product, \(g_1 g_2 \cdots g_r\), of the
  components of $f(g_1,\dots, g_r)$.
  We say that $f$ is \textit{freely trivial} if $f$ evaluates to $1$ for all inputs in $G^r$ when
  $G$ is the free group of rank $w$. 
\end{defn}
  
By the properties of free groups \cite{cohen1989combinatorial}, all freely trivial words
are related to $1$ by a sequence of elementary reductions, i.e. canceling
adjacent inverses. This implies that freely trivial $Z$-constraints are equivalent to the $Z$-checks of a quantum double (Lemma \ref{lem:Zcheckform}).

\begin{defn}[Group law]\label{def:grouplaw}
    A group law is a function $f(g_1,\ldots, g_r)$ which evaluates to 1 for every $(g_1,\ldots, g_r)\in G$.  The smallest law weight, denoted as $\alpha(G)$, is the smallest possible weight of a group law for $G$.
\end{defn}

Clearly $\alpha(G)$ is finite for any finite group $G$, because $f(a)=a^{|G|}$ is not freely trivial but is a group law.
Non-trivially however, $\alpha(G)$ can be arbitrarily smaller in groups with large $|G|$: we are aware of examples for $\mathrm{PSL}_q(\mathbb{F}_r)$ and $\mathrm{A}_n$ \cite{cohen1989combinatorial}, and provide proofs for these in Appendix~\ref{app:group-laws}.

\begin{thm}\label{thm:classify}
Let $G$ be non-Abelian and simple (the commutator subgroup $[G, G]=G$) with smallest law weight $\alpha$.  Define \begin{equation}
    \beta := \left\lfloor \frac{\alpha-1}{2}\right\rfloor. \label{eq:beta_def}
\end{equation} Then every $G$-covariant CSS code with $d_Z>\beta$ and maximal
$Z$-check weight $w_Z\le \beta$ is, after a suitable change of generating set
for $\mathcal{S}_X$ and $\mathcal{S}_Z$, a quantum double code
(Definition \ref{def:roughQD}).
\end{thm}

\begin{proof}
The proof is presented in Appendix \ref{app:main}.
\end{proof}

One may hope that the group CSS framework was too restrictive, and that a different generalization of CSS codes to $G$-valued qudits would allow for more sophisticated codes.  A hint that this might be possible is the absence of duality between $X$/$Z$-type checks (Remark \ref{rmk:XZdual}).  Unfortunately, the main obstruction to building a good group CSS code is the $X$-checks (Definition \ref{def:Xcheck}), which were the nicer of the two generalizations of CSS codes (these checks are still transversal).

\begin{cor}
    Given any quantum stabilizer code stabilized by $X$-checks as given in Definition \ref{def:Xcheck}, the code distance is bounded by \eqref{Eq: dZ bound} if $G$ is non-Abelian and simple and the code is $G$-covariant.\label{cor:dZ_simple_covariant}
\end{cor}

\begin{proof}
    The assumptions above are already enough to force the $X$-checks of a $G$-covariant code take the form of a quantum double code (see the discussion after Lemma \ref{lem:equivalentthm}), so the proof of Theorem \ref{thm:dZ bound} applies.
\end{proof}

\section{Outlook}\label{sec:outlook}
We generalize the standard definition of CSS codes on qubits (more generally, on modular qudits $\mathbb{Z}_m$) 
to qudits associated with general groups $G$, yielding group-CSS codes. 
Our construction includes quantum double models as a special case, but it is not limited to them.   

We precisely characterize how quantum double codes can be placed on general 2D CW complexes, which go beyond two-dimensional manifolds due to the presence of junctions.
Surprisingly, we show that, under certain mild assumptions (Theorem \ref{thm:classify}), \emph{every} non-trivial group qudit code is a quantum double model defined on a non-manifold structure called a CW complex.

The CW framework allows one to overcome restrictions on code parameters arising from Euclidean and hyperbolic geometries, yielding a constant-rate code with a \(Z\)-distance (\(X\)-distance) that is linear (logarithmic) in the number of qubits \(n\).
Non-manifold CW complexes, together with our manifold-inspired tools to understand them, open up the possibility of more interesting and useful fault-tolerant gate sets (cf. \cite{warman2026}).

The CW framework also encapsulates several other many-body models.
In particular, we show that non-invertible SPT and SSB states, even in one dimension, are naturally understood as quantum double models on special CW complexes.  This resolves an open question raised in Ref.~\cite{Fechisin_2025} and explains these novel states with non-invertible symmetries in the language of quantum error-correcting codes.  Our framework also makes it transparent that these SPT states are edge states of (rough) quantum double codes, as was described in Refs.~\cite{QD_boundary, Cong1, Cong2, Cong3, warman2025categorical, schafernameki2023ictplecturesnoninvertiblegeneralized, shao2024whatsundonetasilectures, Luo_2024, thorngren2019fusioncategorysymmetryi, thorngren2021fusioncategorysymmetryii, Choi_2022, Bhardwaj_2022, zhang2023anomalies11dcategoricalsymmetries}.  

We also extend techniques used to calculate logical dimension and operators from surface codes to general quantum doubles. 
For typical surface codes (with or without boundaries), this logical structure comes about from non-contractible loops as well as defects between different types of boundaries.
We show that this intuition holds for CW quantum doubles and develop a formalism to explicitly determine the logical structure for any combination of group and CW complex.
In particular, our logical operator recipe unifies transversal logical gates, ribbon operators, and 1-form symmetries under one roof.

We conclude with some discussion of alternative ideas for how to find new group qudit codes and potentially intriguing physics that might arise in this setting.

\begin{enumerate}[1.]
    \item 
    Our code bounds on CW quantum doubles prohibit asymptotically good codes of that type.
    Our rigidity theorem further suggests that constructing such codes --- let alone those of LDPC-type --- requires special types of \(G\).
    It is similarly difficult to construct classical good codes over groups \cite{biglieri1995construction,interlando1996group}, and obtaining such constructions may require us to leave generalized group-qudit stabilizer formalisms altogether (cf. \cite{sahebi2012asymptotically} for the classical case).

    \item It would be interesting to see how much of the stabilizer-like formalism studied here can be salvaged when one further generalizes from group-based codes to category-theoretic ones \cite{levin2005string,koenig2010quantum}, which \textit{can} realize universal computation fault tolerantly \cite{freedman2002modular,zhu2020quantum}.

    \item Consider an arbitrary quantum code on $\mathbb{Z}_m$; these include asymptotically good codes \cite{panteleev2022asymptotically,leverrier2022quantum,dinur2023good}.  Consider the sign automorphism $\varphi \in \mathrm{Aut}(\mathbb{Z}_m)$ obeying $\varphi(x)=-x$.  Gauging this global automorphism leads to an expander code with non-invertible symmetries: a logical operator now takes the (schematic) form of $X^m+X^{-m}$ which is non-unitary.  This gauging procedure can be trivially incorporated in the language of group CSS codes as follows.  Consider the dihedral group $\mathrm{D}_{2m} :=\langle r,s | s^2,srsr,r^m\rangle $ with subgroup $\mathbb{Z}_m=\langle r\rangle$.   We can use $Z$-checks of the form $Z^{\langle r\rangle}_{g_i}$ for each qudit $g_i$, after which point the starting code is already a group qudit CSS code (Corollary \ref{cor:Zm}).  To gauge the automorphisms we simply add $X$-check $\prod_{i=1}^n \overrightarrow{X}^s_i\overleftarrow{X}^s_i$. By introducing additional qudits restricted to lie in the subgroup $\langle s\rangle$, this last check may be replaceable by small checks acting on additional qudits.

\item Consider a group with a non-trivial abelianization $G_{\mathrm{ab}} = G/[G,G]$.  Take a quantum double code on a 2D CW complex with $X$-checks restricted to multiplication in $[G,G]$, and include multiplication by non-trivial elements of $G_{\mathrm{ab}}$ only on combinations of the double's $X$-checks.   Notice that any check of the form $Z^{[G,G]}_{q_1\cdots q_p}$ is equivalent to a check of the form $\varphi(q_1)+\cdots \varphi(q_p)=0$ with $\varphi : G\rightarrow G_{\mathrm{ab}}$ the homomorphism with $\ker(\varphi)=[G,G]$.  We can then add a rather general Abelian group CSS code on $G_{\mathrm{ab}}$ ``on top of" the quantum double code on the full group $G$, so long as the $Z$-checks of the quantum double code are also checks of the abelian code.  More generally, taking a group with an interesting central series $\cdots \unlhd [G,G]\unlhd G$ may lead to a ``layered" quantum double model.  Potentially similar ideas involving quantum double codes mixing $[G,G]$ and $G$-valued qudits can also be found in the context of color codes \cite{Brell_2015} and fracton codes \cite{tantivasadakarn2021non}.   A simple example of a group where such codes may be interesting are the dihedral groups $\mathrm{D}_{2m}$, where $\mathbb{Z}_m=[\mathrm{D}_{2m},\mathrm{D}_{2m}]$; this enables us to leverage known expander codes on $\mathbb{Z}_m$.\footnote{Classical codes of this kind for prime \(p\) have been constructed \cite{sipser2002expander,zemor2002expander,RaoCodingTheoryLecture6}\cite[Theorem~7.16]{jeronimo2023fast}, and we should be able to embed them into \(\mathbb{Z}_{m=p^l}\) for any power \(l\) via a strategy similar to that of  Example~\ref{eg:silly}.}   

\item There is a natural way to ``twist" group CSS codes by modifying the $X$-type checks to include multiplication by a diagonal unitary in the $Z$-eigenbasis.  
We hope to describe this further in a future paper; the construction is more general than specifying a 3-cocycle, as in a twisted quantum double on a manifold \cite{Hu_2013, Hu_2017, CuiTwisted}. 
The recent constructions in  \cite{cupsgatescohomology,  Lin:2024uhb, tiedinknots, Zhu:2026vec, vedhika_twisted} can be understood as twisted group CSS codes with group $G=\mathbb{Z}_2$, for example. 

    \item 2-group gauge theories \cite{Bullivant_2017, Pfeiffer_2003, Baez_2010} defined using a crossed module with two
      finite groups should naturally lead to quantum codes on 3D CW complexes.
      Our initial investigation suggests that expander 2-group codes based on non-Abelian groups still have a string-like logical sector, suggesting that similar constraints to Theorem \ref{thm:dZ bound} apply,  at least to one sector of logical operators.

      \item We are able to re-formulate quantum-double and more general group-based codes in terms of double cosets.
    Another example of such double-coset codes are the diatomic molecular codes \cite[Sec. VI]{Albert_2020}.
    These two examples happen to be defined on symmetric spaces, the space \((G\times G) / G\)~\cite{bachoc2012invariant} in our case, and the sphere in the molecular case.
    It would be interesting to see how far these frameworks can be pushed to define codes and phases of matter on other symmetric and more general homogeneous spaces.
\end{enumerate}

\emph{Note added.}--- As we were writing this manuscript, the preprint Ref.~\cite{ohyama2026parameterized} appeared, which introduced the 2D cluster-state Hamiltonian from Sec.~\ref{Sec: non-inv SPT} along with the PEPS representation of its ground state.

\section*{Acknowledgments}
We thank 
Meng Cheng, Nicolas Delfosse, Tyler Ellison, Michael Gullans, Weizhen Jia, Henry Lamm, Chong Wang, Dominic Williamson, and Yichen Xu for useful discussions. 
This work was supported by the
Department of Defense through the National Defense Science \& Engineering Graduate Fellowship Program (BTM), by the
Department of Energy under Quantum Pathfinder Grant DE-SC0024324 (JHZ, AL), and by Air Force Office of Scientific Research Grant FA9550-24-1-0120 (AL).

\begin{appendix}
\renewcommand{\thesubsection}{\thesection.\arabic{subsection}}
\renewcommand{\thesubsubsection}{\thesubsection.\arabic{subsubsection}}
\renewcommand{\theequation}{\thesection.\arabic{equation}}

\section{Group CSS codes with Abelian $G$ are group GKP codes}\label{app:abelian}

\begin{prop}\label{prop:groupGKPabelian}
    Every group CSS code with Abelian $G$ on $n$ qudits is a group GKP code with group $G^n$.  Moreover, the Hamiltonian \eqref{eq:codeH} is a sum of commuting terms.
\end{prop}
\begin{proof}
It is convenient to use addition as the group operation on an Abelian group $G$; accordingly we denote $\mathbf{0}$ as the identity operator for this proof.

Let us first analyze the $X$-checks. Notice that $\overrightarrow{X}^g=\overleftarrow{X}^{g^{-1}}$ for any $g$ in an Abelian group $G$, so without loss of generality, we take the $X$-checks to only contain left multiplication. 
Such checks, \(X^{\mathbf{h}, \mathbf{0}}\) with \(\mathbf{h} \in \diagX\) form a  representation of the \(X\)-check subgroup.

The $Z$-checks can similarly be re-written for Abelian $G$ in a nicer way. First we define $Z_\alpha:= \mathbb{I}(g_1^{a_{\alpha1}}\cdots g_n^{a_{\alpha n}}\in K_\alpha)$, where we have used the fact that $G$ is Abelian to write the group word in an ordered form without loss of generality.
Then, the stabilizer condition is equivalent to an equation whose solutions define another subgroup: 
\begin{equation} \label{eq:aalpha}
    Z_\alpha |\mathbf{g}\rangle =|\mathbf{g}\rangle \iff a_{\alpha 1}g_1+\cdots + a_{\alpha n} g_n \in K_{\alpha} \iff \mathbf{g}\in L_\alpha,
\end{equation} 
where in the last step we have defined the ``global'' subgroup $L_\alpha \le G^n$ by the set of solutions to this equation.   The set of solutions indeed form a group (see Remark \ref{rmk:abeliancodelinear}).  
We can then define the subgroup $K$ in Definition \ref{def:GKP} as the intersection \begin{equation}
    K := \bigcap_{\alpha } L_\alpha \leq G^n~.
\end{equation}
This yields the \(Z\)-check subgroup associated with the group GKP code.

All that is left now is to show that \(H \leq K\).
The frustration-free condition on Hamiltonian $H$ in \eqref{eq:codeH} -- equivalently the condition that the codespace is stabilized by all of our checks -- simplifies to checking compatibility of $X$ and $Z$ checks. 
Since $G$ is Abelian, the $X$-checks now all commute.  
An $X$-check corresponding to $\mathbf{h}\in H$ takes a codeword to another codeword if and only if for every $\alpha$, 
\begin{equation}
    a_{\alpha 1}h_1+\cdots a_{\alpha n}h_n \in K_\alpha ,  \label{eq:alpha_h}
\end{equation}
which is equivalent to $\mathbf{h}\in K$, or $H\le K$.   This establishes that the group CSS code is a group GKP code on $G^n$.  It is also clear that when \eqref{eq:alpha_h} holds, 
the $X$-checks and $Z$-checks commute as operators, meaning that 
Hamiltonian \eqref{eq:codeH} is a sum of commuting terms (keeping in mind that \(G\) is Abelian).
\end{proof}

\begin{cor}\label{cor:Zm}
    When $G=\mathbb{Z}_m^l$, the subgroups $H$ and $K$ in the group GKP construction can be expressed as parity check matrices $\mathsf{H}_X \in \mathbb{Z}_m^{ln_X\times ln}$ and $\mathsf{H}_Z \in \mathbb{Z}_m^{ln_Z\times ln}$ obeying $\mathsf{H}_X\mathsf{H}_Z^{\mathsf{T}}=0$.  The codespace \begin{equation}
        \mathcal{C} := \mathbb{C}[\mathrm{Ker}(\mathsf{H}_Z)/\mathrm{Im}(\mathsf{H}_X^{\mathsf{T}})].
    \end{equation}
\end{cor}
\begin{proof}
    This is a well-established result in the literature, but we will go through it to see what is special about the groups $G=\mathbb{Z}_m^l$.   The essential case is $l=1$; for $l>1$ we can identically interpret the code as a group CSS code on $G=\mathbb{Z}_m$ with $ln$ qubits, where each of the $n_Z$ $Z$-checks is replaced with $l$ identical-looking checks, one for each ``copy" of the $\mathbb{Z}_m$.  

    Suppose that we have some generating set of group words $\mathbf{h}\in H$.  Then we can build the parity-check matrix $\mathsf{H}_X$ by defining it as follows: \begin{equation}
       \text{generator } \mathbf{h}_\beta =: ((\mathsf{H}_X)_{\beta 1}, \ldots (\mathsf{H}_X)_{\beta n})\in \mathbb{Z}_m^n.
    \end{equation} Similarly, the parity-check matrix $\mathsf{H}_Z$ is built by setting \begin{equation}
        (\mathsf{H}_Z)_{\alpha j} = \frac{m}{c_\alpha} a_{\alpha j} \label{eq:HZcalpha}
    \end{equation} where $a_{\alpha j}$ was given in \eqref{eq:aalpha}, and $c_\alpha$ is a divisor of $m$ that we now fix.  Every subgroup of $\mathbb{Z}_m$ is equivalent to $c\mathbb{Z}_m$ for an integer divisor $c$ of $m$.  Hence, the constant $c_\alpha$ in \eqref{eq:HZcalpha} is fixed by the choice of subgroup $H_\alpha =: c_\alpha \mathbb{Z}_m$ in the check, as $g\in c\mathbb{Z}_m$ if and only if $(m/c)\times g = 0\in\mathbb{Z}_m$.

    A short calculation shows that $\mathsf{H}_X\mathsf{H}_Z^{\mathsf{T}}=0$ is equivalent to conditions \eqref{eq:aalpha} and \eqref{eq:alpha_h}.
\end{proof}

\begin{rmk}\label{rmk:XZdual}
    In the special case where $G=\mathbb{Z}_m$, we could use a ``Hadamard" transformation to exchange the roles of $\mathsf{H}_X$ and $\mathsf{H}_Z$.  More generally, in this case we can identify the codespace as either $\mathrm{Ker}(\mathsf{H}_X)/\mathrm{Im}(\mathsf{H}_Z^{\mathsf{T}})$ or $\mathrm{Ker}(\mathsf{H}_Z)/\mathrm{Im}(\mathsf{H}_X^{\mathsf{T}})$.  
    
    If we take $G=\mathbb{Z}_m^l$, we can already see that the notion of duality above between $X$-checks and $Z$-checks is lost.  
    Take the example \(\mathbb{Z}_{2}^{2}=\{(\alpha,\beta)\,|\,\alpha,\beta\in\mathbb{Z}_{2}\}\) with binary addition being group ``multiplication''.
    Consider the code on two \(\mathbb{Z}_2^2\)-qudits with a single stabilizer $\overrightarrow{X}_{1}^{(1,1)}\overrightarrow{X}_{2}^{(0,1)}$.  
    To interpret this code in a dual basis with a single $Z$-check, we need to write this check as 
    \begin{equation}\label{eq:rmk-constraint}
        \phi(\alpha_{1},\beta_{1})+\theta(\alpha_{2},\beta_{2})=\alpha_{1}+\beta_{1}+\beta_{2}=0~, 
    \end{equation}
    where $\phi(\alpha,\beta)=\alpha+\beta$ and $\theta(\alpha,\beta)=\beta$.
    However, this is impossible using \(Z\)-checks from Def.~\ref{def:Zcheck} since such checks only test whether the group element is in some subgroup \(K\), \(\mathbb{I}((a_{1}\alpha_{1}+a_{2}\alpha_{2},a_{1}\beta_{1}+a_{2}\beta_{2}) \in K)\), and do not mix the binary coordinates \(\alpha\) and \(\beta\).
    This is because constraints of the form~\eqref{eq:rmk-constraint} are homomorphisms from $\mathbb{Z}_2^2\rightarrow \mathbb{Z}_2$, which are not covered by the framework of group CSS codes introduced in Definition \ref{def:CSS}.

    Of course, one can unblock qudits and consider a group CSS code on $G=\mathbb{Z}_m^l$ as a group CSS code on the group $\mathbb{Z}_m$ with $n \times \ell$ qudits instead, in which case the notion of duality would be restored.
    In that case, each coordinate \(\alpha_1, \beta_1, \alpha_2, \beta_2\) would
    be its own group element, and our \(Z\)-checks would accommodate their
    addition modulo two. This generalizes to arbitrary finite Abelian groups $A$
    if we allow for local unblocking of qudits and increasing the local dimension.  By the classification theorem for finite Abelian groups, $A \cong \Z_{q_1}^{\ell_1} \times \dots \Z_{q_r}^{\ell_r}$ for $q_1, \ldots, q_r$ which are powers of not necessarily distinct primes $p_1, \dots, p_n$. Of course, then the onsite Hilbert spaces do not correspond to a single group $G$. Leveraging the fact that $\Z_{q_i} \leq \Z_{q_1 \dots q_n}$ for each $i$, we can consider this code as a group-CSS code [Def.~\ref{def:CSS}] on $n \times (\ell_1 + \ldots +\ell_r)$ qudits valued in $\Z_{q_1\ldots q_r}$ in the same way described for products of cyclic groups above.
\end{rmk}

\section{Proof of Theorem \ref{thm:optimal}}\label{app:optimal}
In this appendix we show how to construct a quantum double code with asymptotically optimal properties, and prove Theorem \ref{thm:optimal}.    Strictly speaking, the CW complex we choose will depend on the group $G$, although at least one such complex with code parameters \eqref{eq:goodquantumdoubleproperties} will exist for any finite $G$.

In the text that follows, suppose that there exists a non-identity $g\in G$, and a prime $p\in\mathbb{Z}$, such that $g^p=1$.   In every non-trivial finite $G$ there is at least one $p$ for which this criterion holds.  We choose one such $p$ and hold it fixed for the rest of the proof.

\subsection{Constructing a good CW complex}

We start with an asymptotically good classical $[n_0,k_0,d_0]_p$ expander code over $\mathbb{Z}_p$  with
$k_0,d_0=\mathrm{\Theta}(n_0)$ \cite{sipser2002expander,zemor2002expander,RaoCodingTheoryLecture6}\cite[Theorem~7.16]{jeronimo2023fast}.

\begin{thm}[\cite{sipser2002expander,zemor2002expander,RaoCodingTheoryLecture6,jeronimo2023fast}]
    Let $\mathsf{H} \in \mathbb{Z}_p^{m\times n}$ denote the parity-check matrix of a classical $[n_0,k_0,d_0]_p$ linear code with entries $\mathsf{H}_{\alpha j}$ for $\alpha\in\lbrace 1,\ldots, m\rbrace $ and $j\in\lbrace 1,\ldots , n\rbrace$.   There exist $\mu,\nu \in (0,1)$ along with a family of  codes with arbitrarily large $n_0$ obeying \begin{subequations}\label{eq:goodcodeproperties}
        \begin{align}
            k_0 &\ge \mu \cdot n_0, \\
            d_0 &\ge \nu\cdot n_0.
        \end{align}
    \end{subequations}
\end{thm}

We use this code to attach 2-cells to a CW
complex $\Sigma$ whose 1-skeleton $\Sigma_1$ is carefully chosen to ensure that
$d_Z = \mathrm{\Theta}(\log n)$, while reducing $d_X$ (from the classical code)
by at most a constant factor.   Conceptually, the skeleton $\Sigma_1$ will
consist of a spanning tree $T$ together with an additional set of edges $M$ that
form the cycles, with $|T| = \mathrm{\Theta}(|M|)$, and girth
$g=\mathrm{\Theta}(\log|M|)$.  A sketch of $T$ and $M$ is provided in Figure \ref{fig:perfect_matching}.     The idea is that after gauge fixing the spanning tree $T$, the edges in $M$ form the qudits of a classical expander code on $\mathbb{Z}_p$.

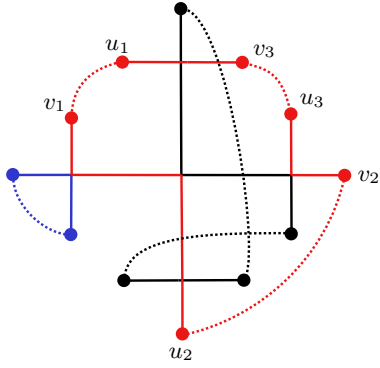
\begin{figure}
  \centering
  \def\svgwidth{.3\textwidth}
\begingroup%
  \makeatletter%
  \providecommand\color[2][]{%
    \errmessage{(Inkscape) Color is used for the text in Inkscape, but the package 'color.sty' is not loaded}%
    \renewcommand\color[2][]{}%
  }%
  \providecommand\transparent[1]{%
    \errmessage{(Inkscape) Transparency is used (non-zero) for the text in Inkscape, but the package 'transparent.sty' is not loaded}%
    \renewcommand\transparent[1]{}%
  }%
  \providecommand\rotatebox[2]{#2}%
  \newcommand*\fsize{\dimexpr\f@size pt\relax}%
  \newcommand*\lineheight[1]{\fontsize{\fsize}{#1\fsize}\selectfont}%
  \ifx\svgwidth\undefined%
    \setlength{\unitlength}{249.49551079bp}%
    \ifx\svgscale\undefined%
      \relax%
    \else%
      \setlength{\unitlength}{\unitlength * \real{\svgscale}}%
    \fi%
  \else%
    \setlength{\unitlength}{\svgwidth}%
  \fi%
  \global\let\svgwidth\undefined%
  \global\let\svgscale\undefined%
  \makeatother%
  \begin{picture}(1,0.8875707)%
    \lineheight{1}%
    \setlength\tabcolsep{0pt}%
    \put(0,0){\includegraphics[width=\unitlength,page=1]{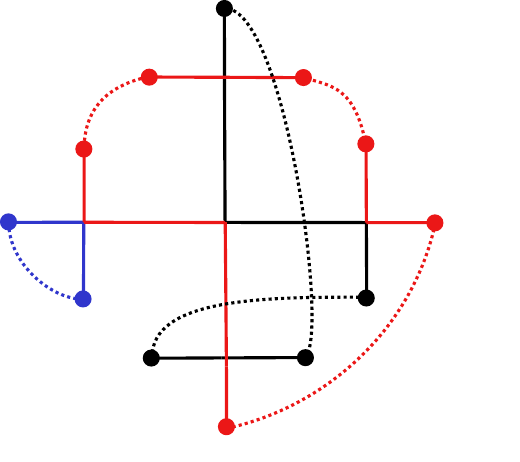}}%
    \put(0.24503508,0.7693109){\color[rgb]{0,0,0}\makebox(0,0)[lt]{\lineheight{1.10000002}\smash{\begin{tabular}[t]{l}$u_1$\end{tabular}}}}%
    \put(0.09201165,0.62377797){\color[rgb]{0,0,0}\makebox(0,0)[lt]{\lineheight{1.10000002}\smash{\begin{tabular}[t]{l}$v_1$\end{tabular}}}}%
    \put(0.40186824,0.00793598){\color[rgb]{0,0,0}\makebox(0,0)[lt]{\lineheight{1.10000002}\smash{\begin{tabular}[t]{l}$u_2$\end{tabular}}}}%
    \put(0.8689836,0.44561172){\color[rgb]{0,0,0}\makebox(0,0)[lt]{\lineheight{1.10000002}\smash{\begin{tabular}[t]{l}$v_2$\end{tabular}}}}%
    \put(0.72616952,0.63540472){\color[rgb]{0,0,0}\makebox(0,0)[lt]{\lineheight{1.10000002}\smash{\begin{tabular}[t]{l}$u_3$\end{tabular}}}}%
    \put(0.6104699,0.75722508){\color[rgb]{0,0,0}\makebox(0,0)[lt]{\lineheight{1.10000002}\smash{\begin{tabular}[t]{l}$v_3$\end{tabular}}}}%
  \end{picture}%
\endgroup%

  \caption{Example of $T \cup M$ with $a = 4$ and $R = 2$. The dots represent leaves in
  $T$, the solid lines represent edges in $T$, and the dotted lines represent
  edges in $M$. The cycle $\Gamma = u_1 \to v_1 \to u_2 \to v_2 \to u_3 \to
  v_3$ is illustrated in red. A cycle of small girth created by the matching is illustrated in blue. These leaves will be removed during the pruning.}
  \label{fig:perfect_matching}
\end{figure}

Naively it is trivial to find $\Sigma_1$, as there exist many expander graphs with girth $\mathrm{\Theta}(\log n)$ -- for example one can take a random $k$-regular graph and prune a subextensive number of edges to remove all short cycles, with high probability \cite{shortcycles}.  
But there is a potential issue with the naive idea: in order to bound $d_X$, we need to ensure we cannot flip the holonomies around too many non-contractible cycles by acting with $X$-operators on qudits that we added to form the spanning tree -- or at least, if we do this, to ensure that we cannot build a codeword in this way.   If we could glue each 2-cell to $\mathrm{\Theta}(1)$ edges in $\Sigma_1$ this would be fairly direct; unfortunately we do not have such a construction.  We will then need a way to ensure that $d_X = \mathrm{\Theta}(n)$ even when some qubits may be involved in a large number of $Z$-checks at large $n$.

To see how to overcome the potential obstacle described above, let us first state the following lemma, which is proved in Appendix \ref{app:optimallemma}.

\begin{lem}\label{lem:A}
    Pick an even integer $a \ge 4$.   There exists an undirected graph/1D CW complex $\Sigma_1 = T\cup M$ where $T$ is a tree of depth $R$ where each parent has $\le a$ children.  $M$ consists of a set of edges between pairs of leaves, and every leaf in $T$ is contained in at most one edge of $M$.   For $R$-independent constants $0<\alpha,\beta,\gamma <\infty$, the girth \begin{equation}
        g(\Sigma_1)\ge \alpha R \label{eq:gSigma1}
    \end{equation}
    and the number of edges in $\Sigma_1$ obeys \begin{equation}
       \left(\frac{a}{a-1} + \frac{1}{2}\right)  a^R \left(1 - R \gamma a^{-\beta R}\right) < n < \left(\frac{a}{a-1} + \frac{1}{2}\right) a^R.  \label{eq:lemmaAedges}
    \end{equation}
    Hence $g(\Sigma_1)=\mathrm{\Theta}(\log n)$, and $T$ represents a spanning tree for $\Sigma_1$.  
    
    Let $\mathcal{L}$ denote the number of leaves of $T$.  Then the number of independent cycles in the graph is $\frac{1}{2}\mathcal{L}$.  We also have \begin{equation}
        \mathcal{L} \ge a^R \left(1 - \gamma a^{-\beta R}\right). \label{eq:lemmaAcycles}
    \end{equation}
    This number also bounds the number of vertices that are leaves of $T$.

    Lastly, the total number of edges in the graph obeys the inequalities \begin{equation}
        \mathcal{L} < n < \frac{2a}{a-1} a^R. \label{eq:aRn}
    \end{equation}
    Therefore, the number of independent cycles is $\mathrm{\Theta}(n)$.
\end{lem}

To see why this lemma is useful, let us now explain how we will use the $\Sigma_1$ it provides to build a ``good" quantum double code. 
We wish to interpret the code by gauge fixing $T$, as depicted by the solid lines in Fig.~\ref{fig:perfect_matching}, which is a spanning tree of the graph $T \cup M$.  After this there is a one-to-one correspondence betwen non-trivial holonomies on $\Sigma_1$ and edges in $M$.  For convenience, given edge $e\in M$, let us denote $g_e$ to be the holonomy around the non-trivial loop on $\Sigma_1$ passing only through edge $e$ in $M$, and otherwise returning to a fixed vertex in $T$ (e.g. the root of the tree).  With this definition in mind, we state:

\begin{lem}
    Let $n_0:=|M|$ denote the number of edges in $M$, in $\Sigma_1$ constructed above.  Let $\mathcal{S}$ denote the set of all  vertices in $T$ that are the parents of leaves.  
    For each $S\in \mathcal{S}$, define a set $S^\prime \subseteq \text{children}(S)$ such that 
    \begin{equation}
        \gcd(|S^\prime|,p)=1. \label{eq:relativeprimeSprime}
    \end{equation} 
    Let $\mathcal{S}^\prime$ be the set of all such $S^\prime$ (there is one $S^\prime$ for each leaf-parent in the tree).  
    Let $\mathsf{H}$ denote the parity-check matrix of a classical code obeying \eqref{eq:goodcodeproperties} with parameters $\mu$ and $\nu$.   Choose an even integer \begin{equation}
        a = 2\left\lceil \frac{1}{\mu}\right\rceil. \label{eq:pickabig}
    \end{equation}
    Build the CW complex $\Sigma$ obtained by gluing 2-cells to $\Sigma_1$ according to the following \(Z\)-check constraints:
    \begin{subequations}
        \begin{align}
            g_e^p&=1\;\;\; \forall \;\;\; e\in M,  \label{eq:ppowercheck}\\
            g_eg_{e^\prime}g_e^{-1}g_{e^\prime}^{-1} &=1 \;\;\; \forall \;\;\; e,e^\prime \in M, \\
            \prod_{e=1}^n g_e^{\mathsf{H}_{\alpha e}} &= 1 \;\;\; \forall \;\;\; \alpha \in \lbrace 1,\ldots, m\rbrace, \label{eq:classicalcodecheck}\\
            \prod_{e\in S^\prime}g_e &= 1 \;\;\; \forall \;\;\; S^\prime\in \mathcal{S}^\prime. \label{eq:Sprimecheck}
        \end{align}
    \end{subequations}
     Then, any $X$-logical operator must act on $\mathrm{\Theta}(n)$ edges that connect to the leaves of $T$.
\end{lem}
\begin{proof}
Firstly, by the first two checks, we can only change the holonomies $g_e$ by some subgroup $\mathbb{Z}_p^r\le G$ for some positive integer $r\ge 1$.  Moreover, any $X$-logical that we find must correctly flip the holonomies on the leaves according to an $X$-logical of the classical code, in order to obey \eqref{eq:classicalcodecheck}.

    Secondly, suppose that we had an $X$-logical that acted only on edges internal to the spanning tree $T\subset \Sigma_1$.   Let $f\in T$ be one such edge.  Acting with e.g. $\overrightarrow{X}^g_f$ -- if it is part of a logical -- will change the holonomies $g_e$ for every edge $e$ whose cycle passes through $f$.  Critically, the set of $e\in M$ obeying this are a union of sets $S^\prime$.   The $Z$-check \eqref{eq:Sprimecheck} will then be violated due to the commutativity checks together with \eqref{eq:relativeprimeSprime}.  This means that any logical operator must have at least one $\overrightarrow{X}^g_e$ for an edge $e$ that connects to a leaf, for every $S^\prime$ that overlaps with the qudits contained in the $X$-logical.  The distance $d_0 \ge \nu n_0$ and therefore \begin{equation}
        |\lbrace S^\prime \in\mathcal{S}^\prime : \text{any logical of classical code overlaps with $S^\prime$}\rbrace | \ge \frac{\nu n_0}{a} = \mathrm{\Theta}(n).
    \end{equation}
The last statement follows because $\mu, \nu,a$ are all O(1) constants and because of \eqref{eq:aRn}.  

Not every logical of the classical code is guaranteed to be consistent with \eqref{eq:Sprimecheck}, as we will discuss more shortly.  Nevertheless, any $X$-logical of this CW code must correspond to an $X$-logical of the original classical code.  Ultimately, we conclude that any $X$-logical must act on at least $\mathrm{\Theta}(n)$ qudits.
\end{proof}

To finish the proof of Theorem \ref{thm:optimal} we now simply need to confirm that $k=\mathrm{\Theta}(n)$.   Consider one fixed subgroup $H\le G$ with $H\cong \mathbb{Z}_p$ and restrict only to holonomies $g_e\in H$.  In the absence of the $S^\prime$ checks \eqref{eq:Sprimecheck} we would have $k_0=\mu n_0$ logicals.  From the perspective of a code on the field $\mathbb{Z}_p$, the $S^\prime$ checks can be understood as adding $|\mathcal{S}^\prime|$ new rows to the parity-check matrix $\mathsf{H}$, which could at most decrease $k_0 \rightarrow k_0-|\mathcal{S}^\prime|$ by the rank-nullity theorem.  Since for large enough $n_0$,
\begin{equation}
    |\mathcal{S}^\prime| \le \frac{3}{2}\frac{n_0}{a}<\frac{3\mu n_0}{4}\le \frac{3k_0}{4}~,
\end{equation}
we conclude that the quantum-double codespace \(\mathcal{C}\) is roughly at least as large as that generated by this $\mathbb{Z}_p$ subgroup, \begin{equation}
    \dim(\mathcal{C}) \ge \frac{p^{k_0-|\mathcal{S}^\prime|}}{p-1}~,
\end{equation}
up to division by $(p-1)$, which accounts for the possible reduction due to the quotient by $\mathrm{Ad}_G$~\eqref{eq:codespacesmooth}.   
We conclude that the code log-dimension is bounded as
\begin{equation}
    k = \log_{|G|}\dim(\mathcal{C}) \ge \left(\frac{k_0}{4}-1\right) \log_{|G|}p  = \mathrm{\Theta}(n)~.
\end{equation}

\subsection{Proof of Lemma \ref{lem:A}}\label{app:optimallemma}
This section will be dedicated to constructing a CW complex with no small cycles. Explicit constructions for such graphs exist (see Ref.~\cite{margulis}), but our construction will be more technically convenient for our purposes.

We prove this existence result using a probabilistic method. 
We will first choose a random
matching $M$ of leaves in a regular tree $T$, then we will remove all leaves
which lie within a small cycle. We show that, on average, the fraction of small cycles --- and therefore of pruned leaves --- is vanishingly small, which will establish the lemma.
\begin{lem}
 Let $a \geq 4$ be an even number, and let $T$ be the depth-$R$ tree where every
parent has exactly $a$ children. The number of leaves is $a^R$, and since $a$ is even, there exist perfect matchings of the leaves in the tree.
Given such a matching,
let us denote $u\sim v$ if leaf $u$ is matched with a leaf $v$, and let $M$
represent the set of additional edges directly connecting \(u\) and \(v\) for any $u \sim v$.
Then
for any $\alpha > 0$ there exist constants $\gamma$ and $\beta$ independent
of $R$ and $\alpha$ such that
\begin{align}
  \mathbb E_M[\#\{u \in T : u \in \text{cycle of length} < \alpha R\}] < \gamma a^{\alpha\beta R}
\end{align}
where $\mathbb E_M$ denotes the expectation value over all perfect matchings.
\end{lem}

\begin{proof}
Let $\mathsf{d}(u,v)$ denote the distance between leaves $u$ and $v$ on $T$;
this corresponds to twice the depth of the join (lowest common parent) of the two leaves, so $\mathsf d(u,v)$ is
always even. Let $L_M$ be the set of non-backtracking loops in $T \cup M$ of
length less than or equal to $2q$, where
\begin{equation}
  q \equiv \left\lfloor  \frac{\alpha R}{2} \right\rfloor~.
\end{equation}
By non-backtracking, we mean loops
which cannot be decomposed into loops of smaller length. 
Nothing is lost by
restricting our attention to $L_M$, because if a vertex lies in a backtracking
loop of length $\ell$, then it must lie in a non-backtracking loop of length $< \ell$.

By linearity, we can express the expectation value as a sum of probabilities of each vertex \(u\) to belong in a non-backtracking loop for a given matching \(M\).
Defining the expected probabilities \(\mathbb{P}_{M}\) to be expectations of those probabilities over \(M\) yields
\begin{equation}\label{eq:expected_bad2}
\mathbb{E}_{M}[\#\{u\in\Gamma:\Gamma\in L_{M}\}]=\sum_{u\in T}\mathbb{E}_{M}\,\text{Pr}[u\in\Gamma:\Gamma\in L_{M}] \equiv \sum_{u\in T}\mathbb{P}_{M}[u\in\Gamma:\Gamma\in L_{M}]~.
\end{equation}
Therefore we consider a fixed vertex $u_1$ and bound the expected probability that it lies
within a non-backtracking loop. First, we observe the following union bound:
\begin{equation}
 \mathbb P_M[u_1 \in \Gamma: \Gamma \in L_M] \leq
 \sum_{m=1}^{\infty}m\mathbb P_M[\text{$u_1$ lies in $m$ loops $\Gamma \in L_M$}] = 
  \mathbb E_{M}[\#\{\Gamma \in L_M: u_1 \in \Gamma\}]
  \label{eq:expected_loops2}
\end{equation}

A loop $\Gamma \in L_M$ necessarily traverses leaves $u_1 \rightarrow v_1
\rightarrow u_2 \rightarrow \cdots \rightarrow u_m\rightarrow v_m$ for some
number of segments $m$, where $u_j\sim v_j$,
because otherwise the loop would retrace itself along the
tree (see Fig.~\ref{fig:perfect_matching} for an example). Furthermore, the non-backtracking property implies that the sequence $u_1,v_1
\dots, u_m, v_m$ uniquely labels $\Gamma$ up to cyclic permutation, because it is a property of trees that
there is a unique non-backtracking path (meaning that each vertex is visited at most once) 
through $T$ connecting any two vertices.
The length of $\Gamma$ is then
\begin{equation}\label{eq:shortcycle2}
        |\Gamma| = m+ \mathsf{d}(v_1,u_2)+\cdots + \mathsf{d}(v_m,u_1)~.
\end{equation}
Therefore we may
rewrite \eqref{eq:expected_loops2} by summing over the possible loops $\Gamma$ labeled by
$u_1 \to v_1 \to \dots u_m \to v_m$ (suppressed in the notation) and the probability that $\Gamma$ forms a loop
in $L_M$:
\begin{equation}
  \mathbb E_M[\#\{\Gamma \in L_M: u_1 \in \Gamma\}] =
  \sum_{m=1}^{q}\sum_{u_1, v_1, \dots, u_m, v_m}\mathbb I[|\Gamma| \leq 2q]\mathbb P_M[\{u_i \sim v_{i}: 1\leq i \leq m\}]~,
\label{eq:union_bound}
\end{equation}
where \eqref{eq:shortcycle2} implies that $|\Gamma|$ is independent of $M$.
Furthermore, Since the number of matchings on $a^R$ vertices is $(a^R-1)!!$, the number of matchings with $u_i \sim v_i$ fixed is $(a^R-2m-1)!!$. Since all matchings are equally likely,
\begin{equation}
\mathbb P_M(\{u_i \sim v_{i}: 1\leq i \leq m\}) = \frac{(a^R - 2m -1)!!}{(a^R-1)!!}.
\end{equation}

Next, we introduce the variables $2J_i \equiv \mathsf d(v_i, u_{i+1})$ for $i \leq
m$ and $2J_m \equiv \mathsf d(v_m, u_1)$. We can reorganize the remaining sum in Eq.~\eqref{eq:union_bound} using these variables:
\begin{equation}
\sum_{u_1, v_1, \dots, u_m, v_m}\mathbb I[|\Gamma| \leq 2q] = 
\sum_{J_1, \dots, J_m} N(J_1, \dots, J_m)\mathbb I[m+2\sum_{i=1}^mJ_i \leq 2q]
\label{eq:expected_cycle_number2}
\end{equation}
where $N(J_1, \dots, J_m)$ is the number of choices of $v_1, u_2, v_2, \dots,
u_m, v_m$ consistent with $\mathsf d(v_i ,u_{i+1}) = 2J_i$ for $i \leq m$ and
$\mathsf d(v_m, u_1) = 2J_m$. We can bound this number by
\begin{align}
  N(J_1, \dots, J_m) \leq \frac{(a^R-1)!!}{(a^R-2m+1)!!}a^{\sum_{i=1}^m J_i},
\end{align}
because each choice of $u_i, v_i$ eliminates two choices for $v_{i+1}$, and
given any leaf $u$ there are $a^{J}$ leaves $v$ with $\mathsf d(u, v) \leq J$.
The choices for $v_m$ are further restricted because they must lie at a
distance $2J_m$ from $u_1$, which is fixed.

Now let $P(\ell)$ be the number of choices of $J_1, \dots, J_m$ such that
$\sum_{i=1}^m J_i = \ell$.  
The generating function of $P$ is
\begin{align}
  \sum_{\ell = 0}^\infty x^\ell P(\ell) =
  \sum_{\ell = 0}^\infty x^\ell \sum_{J_1, \dots, J_m} \mathbb I\qty(\sum_{i = 1}^mJ_m = \ell) = \qty(\sum_{j = 0}^{\infty}x^j)^m = (1-x)^{-m}
\end{align}
By the generalized binomial theorem, we have
\begin{align}
  P(\ell) = (-1)^\ell\binom{-m}{\ell} = \binom{m+\ell}{m} \leq \frac{(2q)^m}{m!}~,
\end{align}
where the inequality holds whenever $\ell + m \leq 2q$, which is always the case.

Substituting all of the above into \eqref{eq:expected_cycle_number2}, we find
\begin{align}
\sum_{J_1, \dots, J_m} N(J_1, \dots, J_m)\mathbb I\qty[m+2\sum_{i=1}^mJ_i \leq 2q] &= 
\sum_{\ell = m}^{q-m/2}\sum_{J_1, \dots, J_m} N(J_1, \dots, J_m)\mathbb I\qty[\sum_{i=1}^mJ_i = \ell] \notag\\
 &\leq \frac{(a^R-1)!!}{(a^R-2m+1)!!}\sum_{\ell = m}^{q - m/2}a^{\ell}P(\ell) \notag \\
  &\leq  \frac{4(a^R-1)!!}{3(a^R-2m+1)!!}\frac{(2q)^m}{m!}a^{q-m/2},
\end{align}
where we used the assumption that $a \geq 4$ to conclude that $\frac{a}{a-1}
\leq \frac{4}{3}$ when summing the series. Now substituting this result back into
\eqref{eq:union_bound}, we find
\begin{align}
\sum_{m=1}^{q}\sum_{u_1, v_1, \dots, u_m, v_m}\mathbb I[|\Gamma| \leq 2q+1]\mathbb P_M(\{u_i \sim v_{i}: 1\leq i \leq m\}) &\leq
\frac{4}{3}\sum_{m=1}^{q}\frac{(a^R - 2m-1)!!}{(a^R-2m+1)!!}a^{q-m/2}\frac{(2q)^m}{m!} \notag\\
&\leq \frac{4a^{q}}{3(a^R-2q+1)}\e^{2q/\sqrt{a}} \notag\\
&\leq \frac{4}{3}\qty(1-\frac{\alpha}{\e\log(a)})^{-1}a^{-[1-\alpha(1+\frac{2}{\sqrt{a}\log(a)})]R}.
\end{align}
Finally, putting this estimate back into \eqref{eq:expected_bad2}, we find
\begin{align}
 \mathbb E_M[\#\{u \in \Gamma:\Gamma \in L_M\}] \leq \frac{4}{3}\qty(1-\frac{\alpha}{\e\log(a)})^{-1}a^{\alpha(1+\frac{2}{\sqrt{a}\log(a)})R}.
\end{align}
This finishes the proof of the lemma.
\end{proof}

Lemma \ref{lem:A} follows from the fact that there must exist at least one
perfect matching $M$ where the number of pruned leaves is less than or equal to
the average. An example of $T$ and $M$ are illustrated in Fig.~\ref{fig:perfect_matching}.

\section{Tensor network representation of the 2D non-Abelian cluster state}
\label{App: 2D cluster}

In Sec. \ref{Sec: non-inv SPT}, we demonstrated that 1D and 2D non-Abelian cluster states with non-invertible symmetries can also be unified as generalized quantum double codes with ghost checks. In this Appendix, we give a projected entangled-pair state (PEPS) representation of the 2D non-Abelian cluster state (while the tensor network representation of the 1D non-Abelian cluster state as a matrix product state (MPS) has been proposed in Ref. \cite{Fechisin_2025}). 

\begin{figure}
\begin{tikzpicture}[scale=1, line cap=round, line join=round]

\tikzset{
  edge/.style={thick},
  nodev/.style={circle, draw, fill=white, minimum size=12pt, inner sep=0pt},
  centerv/.style={circle, draw, fill=black,                  pattern color=black, minimum size=12pt, inner sep=0pt}
}

\coordinate (C)  at (0,0);
\coordinate (L)  at (-1.6,0);
\coordinate (R)  at ( 1.6,0);
\coordinate (DL) at (-0.9,-1.2);
\coordinate (UR) at ( 0.9, 1.2);

\draw[edge] (-2.2,0) -- (2.2,0);                 
\draw[edge] (-1.2,-1.6) -- (1.2,1.6);            
\draw[edge] (C) -- ++(0,0.55);
\draw[edge] (L) -- ++(0,0.55);
\draw[edge] (R) -- ++(0,0.55);
\draw[edge] (DL) -- ++(0,0.55);
\draw[edge] (UR) -- ++(0,0.55);

\node[nodev]   at (L)  {};
\node[nodev]   at (R)  {};
\node[nodev]   at (DL) {};
\node[nodev]   at (UR) {};
\node[centerv] at (C)  {};

\node (anchor) at (4.7,1.4) {}; 

\node[yshift=-0.5cm, anchor=west,align=left] at (anchor) {%
$\bullet\quad \tinyhub\;=\;\displaystyle\sum_{g}
  (g_{s},g_{r},g_{u},g_{d}|\otimes\lvert g\rangle_{p}$.
};

\node[below=2cm of anchor,anchor=west,align=left] {%
$\bullet\quad \tinytwoleg\;=\;\displaystyle\sum_{g}\,
  (gh,h| \otimes |g\rangle$.
};

\end{tikzpicture}
\caption{Projected entangled-pair state (PEPS) of 2D non-Abelian $G$-cluster model \eqref{Eq: 2D cluster} on a Lieb lattice. Here we use the solid ball to label the control qudits and the empty ball to label the target qudits, respectively, with their explicit analytical expression in the right panel.}
\label{Fig: 2D cluster tensor}
\end{figure}
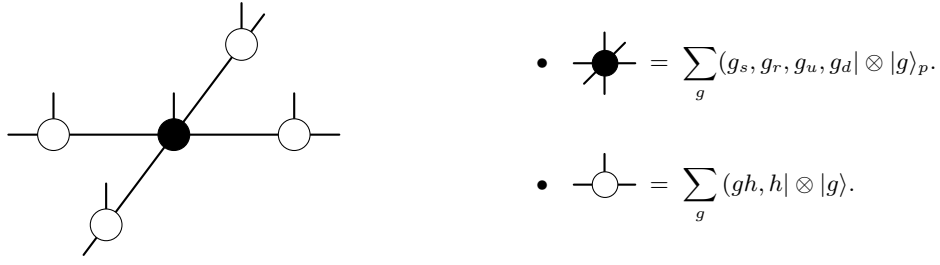

Given the 2D non-Abelian $G$-cluster model with the Hamiltonian \eqref{Eq: 2D cluster}, we first construct the local tensor of $X_g$ and $Z_\Gamma$ operators,
\begin{align}
\overrightarrow{X}_g
\;=\;
\sum_{h}\,\lvert gh\rangle\!\langle h\rvert
\;=\;
\Xrighticon
\quad,\quad
\overleftarrow{X}_g
\;=\;
\sum_{h}\,\lvert h g^{-1}\rangle\!\langle h\rvert
\;=\;
\Xlefticon
\quad,\quad
Z_\Gamma=\sum_{g}\Gamma(g)\otimes\ket{g}\bra{g}=\Zicon,
\label{Eq: local tensor X Z}
\end{align}
where the virtual indices of $Z_\Gamma$ are matrix elements of the representation $\Gamma\in\mathrm{Rep}(G)$. Then the Hamiltonian \eqref{Eq: 2D cluster} can be rephrased in terms of the following projected entangled-pair operators (PEPO), namely
\begin{align}
\Zstabilizer
\quad, \quad
\Xstabilizer
\quad.
\end{align}
It is straightforward to verify that the projected entangled-pair state, as illustrated in Fig.~\ref{Fig: 2D cluster tensor}, is the unique ground state of the stabilizer Hamiltonian \eqref{Eq: 2D cluster}, as well as the global $G$ symmetry and 1-form non-invertible $\mathrm{Rep}(G)$ symmetry \eqref{Eq: 2D cluster symmetry}.

\section{Group laws}
\label{app:group-laws}

Here, we establish the smallest group-law weight explicitly for two families of simple
groups, $\mathrm A_n$ and $\PSL_2(\Z_p)$, to which our main results then apply. 

\begin{lem}
If $n \geq 5$, then $\alpha(\textnormal{A}_n) > n-3$.
\end{lem}

\begin{proof}
  We will follow the proof of residual finiteness of free groups found on
  \cite[pg.~11]{cohen1989combinatorial}. 
  Let $\F_m = \langle  x_1, \dots,
  x_m \rangle$ be a free group of rank $m$, and let
  \begin{align}
    s = y_1y_2 \dots y_n
  \end{align}
  be a reduced word in $\F_m$, where $y_i = x_{\eta(i)}^{\epsilon_i}$ for $\epsilon_i =
  \pm 1$ and an indexing function $\eta:[n] \to [m]$, and suppose without loss of generality that $\epsilon_n = 1$.
  Then for each generator $x_i$, let $\phi(x_i): [n+1] \to [n+1]$ be a map such that $\phi(x_j): i
  \mapsto i + \epsilon_i$ for all $i$ such that $\eta(i) = j$ and $\phi(x_j):
  i \to i$ otherwise. If $\phi(x_j)$ fails to be injective,
  then $\eta(i) = \eta(i+1) = j$ and $\epsilon_i = -\epsilon_{i+1}$,
  contradicting the assumption that $s$ was reduced. A similar consideration
  shows that $\phi(x_j)$ is well-defined.
  Thus, each $\phi(x_i)$
  may be extended to a permutation in $\mathrm S_{n+1}$. Since
  $[\phi^{-1}(s)](n+1) \neq n+1$ by construction, clearly $\phi(s)$ is a
  non-identity permutation, and since $\phi$ is defined on the generators,
  $\phi:\F_m \to \mathrm S_{n+1}$ is a homomorphism.

  To extend $\phi$ to $\mathrm A_{n+3}$, we can simply add two extra labels
  and transpose them whenever $\phi(x_i)$ is odd, ensuring that $\phi(x_i)$ is
  even for each $i$. 

  To see how this implies the claim, suppose that we have a group law $f(y_1, \dots, y_m)$ in $\mathrm
  A_{n}$ with weight $w \leq n-3$. Treating the codomain of $f$ as a free group, we may
  assume that $f$ is freely reduced (because this only decreases $w$). By the
  argument above, we can then construct $\phi:
  \F_{m} \to \mathrm A_{n}$ where $\phi(s) \neq 1$ for some $s$ in the image of
  $f$. This implies that $\phi \circ f 
  \neq 1$, so $f$ was not a group law, exhibiting a contradiction. This shows
  that $\alpha(\mathrm A_{n}) > n-3$, as was to be found.
\end{proof}

\begin{lem}
If $p$ is a prime number, then $\alpha(\PSL_2(\Z_p)) \geq \log^2_\alpha(p)/4$, where $\alpha = 1+\sqrt{2}$.
\end{lem}

\begin{proof}
  Consider the elements
  \begin{align}
    a &= \mqty(1 & 2 \\ 0 & 1) & b &= a^{\mathsf{T}}
  \end{align}
  in $\PSL_2(\Z)$. These elements together generate a free
  subgroup  $ \F_2 \leq \PSL_2(\Z)$, which can be proven via a ``ping-pong''
  argument, making it virtually free 
  (i.e., it has a free subgroup of finite index).
  The
  commutator subgroup of $\F_2$ is free because all subgroups of a free group
  are free, and infinite rank because it is generated by the elements $[a^i,
  b^j]$ for all nonzero $i, j \in \Z$.

  Let $f(x_1, \dots, x_m)$ be a group law in $\PSL_2(\Z_p)$. 
  Then the elements $[a^i, b^j]$ for $|i|, |j| \leq M$ generate a free subgroup $\F_m$ of $\PSL_2(\Z)$. The number of generators of this subgroup is $m = 4M^2 - 1$, corresponding to all the choices of $i,j$  with $(i,j) \neq (0,0)$, so we have $M = \sqrt{m-1}/2$.
  If any $[a^i,b^j]$ is identity after reduction mod $p$, then we know
  \begin{align}
  p \leq \Vert [a^i, b^j]\Vert \leq \Vert a\Vert^{i}\Vert a^{-1}\Vert^i \Vert b \Vert^{j}\Vert b^{-1}\Vert^j  = \alpha^{4M},
 \end{align} 
 where $\alpha = 1+\sqrt{2}$.
 Therefore if $3^{4M} < p$, then $f(x_1, \dots, x_m)$ will be a nontrivial word in $\F_2 \leq \PSL_2(\Z)$ which does not
  reduce mod $p$. Solving this, we find
  \begin{align}
    4M \leq 2\sqrt{m} < \log_\alpha(p),
  \end{align}
  implying this holds for all
  \begin{equation}
    m < \frac{\log_\alpha^2(p)}{4}.
  \end{equation}
  This establishes the claim.
\end{proof}

The orders of both groups scale very unfavorably with the minimum law distance in both these cases. Specifically, $|\PSL_2(\Z_p)| = p(p^2-1)/2$ and $|A_n| = n!/2$. This implies that weight-4 $Z$-checks (such as the surface code) would require $n = 7$ and thus $|A_n| = 2520$ or $p = 44$ and thus $|\PSL_2(\Z_p)| = 19635$ before our theorem guarantees that all the group-CSS codes on either of these two groups must be quantum doubles on a CW complex. This leaves open a lot of room for identifying better group-CSS codes. Our results guide this search by showing that we must leverage the structure of a specific choice of $G$, in particular normal subgroups and small group laws.

\section{Proof of Theorem \ref{thm:classify}}\label{app:main}
The proof proceeds in three steps: \begin{enumerate}
    \item Show that every $Z$-check projects onto solutions of a group word equation (Lemma \ref{lem:groupword}).
    \item Show that the $X$-checks take the form of \eqref{eq:Xspecialform} (Lemma \ref{lem:equivalentthm} and following discussion).
    \item Show that every $Z$-check can be taken to be a group word corresponding to the gluing rules of a 2-cell on a CW complex.  More precisely, we will show that any $Z$-check of the code can be decomposed into smaller $Z$-checks each of which take this form.
\end{enumerate}

\begin{lem}\label{lem:groupword}
    Under the conditions of Theorem \ref{thm:classify}, every $Z$-check is of the form $Z^{\lbrace 1\rbrace}_{q_1\cdots q_p}$, namely a projector onto the solutions of some group word $g_{q_1}^{a_1}\cdots g_{q_p}^{a_p}=1$ where $a_i \in \{-1, 1\}$.
\end{lem}

\begin{proof} Consider any $Z$-check $Z^H_{q_1\cdots q_p}$.
Def.~\ref{def:covariant} implies that $X_{\mathrm{tot}}^g Z_{q_1 \cdots q_p}^K
X_{\mathrm{tot}}^{g^{-1}} \in \mathcal S_Z$.
Therefore, if the product $g_{q_1}\cdots g_{q_p} \in K$, then $g (g_{q_1}\cdots g_{q_p})g^{-1} \in K$ as well. Without loss of generality, we may choose $K$ to be the smallest possible subgroup.  We know that this subgroup has the property that there is a set of generators $h_a$ (the allowed values of $g_{q_1}\cdots g_{q_p}$) such that $gh_ag^{-1} \in K$.   Therefore if $K=\langle h_a\rangle$, $gKg^{-1}=K$, or $K\unlhd G$, namely $K$ is a normal subgroup of $G$.  But $G$ is non-Abelian and simple, so $K=\lbrace 1\rbrace$ or $K=G$.  The operator $Z^G_{q_1\cdots q_p}$ would be the identity operator and does not need to be included in $\mathcal{S}_Z$. 
In other words, covariant codes over simple groups do not admit \(Z\)-checks associated with nontrivial subgroups.
\end{proof}

Next, we will prove that the $X$-stabilizers take the form of
\eqref{eq:Xg_shorthand}.  Recall from the double-coset construction that the
$X$-stabilizers form a representation of the \(X\)-check subgroup $\diagX\le
G^{2n}$~\eqref{eq:x-check-group}.  Translating \eqref{eq:Xg_shorthand} into the double-coset language, we
need to show
\begin{equation} 
    \diagX = \left\langle(g^{\mathbb{I}(1\in B_\alpha)},\ldots, g^{\mathbb{I}(2n\in B_\alpha)}) \right\rangle_{g\in G, \alpha\in\lbrace 1,\ldots,  p\rbrace}~,\label{eq:Xspecialform}
\end{equation}
where $B_{\alpha} \subseteq \{1, \dots, 2n\}$ is a subset of double-coset
indices upon which the $X$-check labeled by $\alpha$ acts non-trivially. The sets $B_\alpha$ will
eventually become the vertices of the CW double associated  to the QD code.
We will show \eqref{eq:Xspecialform} in two steps: first by establishing that $\diagX$ is isomorphic to $G^L$ for some integer $1\le L\le 2n$ and generated by ``diagonal" subgroups, up to a ``twist", and then by showing that $G$-covariance forbids the twist.   The first step is captured by the following lemma.

\begin{lem}\label{lem:equivalentthm}
    Let $\diagX \le G^{L}$ with $G$ non-Abelian simple.  Define $\pi_j:G^{L}\rightarrow G$ to be a projector onto the $j^{\mathrm{th}}$ component.  Suppose that for each $j$, $\pi_j(\diagX)=G$.  Then $\diagX$ is \begin{equation}\label{eq:thetagenerate}
    \diagX = \left\langle \left(\theta_1(g)^{\mathbb{I}(1\in B_\alpha)},\ldots, \theta_{L}(g)^{\mathbb{I}(L\in B_\alpha)}\right)  \right\rangle_{\stackrel{\alpha = 1 \dots p}{g \in G}}
    \end{equation}
for some set of automorphisms $\theta_j\in\mathrm{Aut}(G)$, and $B_1\sqcup
\cdots \sqcup B_p = \lbrace 1,\ldots, L\rbrace$.  
\end{lem}

The essence of this lemma is as follows.  Obviously, there are many potential subgroups $\diagX \le G^{2n}$.  However, the key idea is that if each ``entry" of $\mathbf{g}=(g_1,\ldots, g_{2n})\in \diagX$ gets to take any value in $G$ for some element $\mathbf{g}\in\diagX$, then $\diagX$ is isomorphic to a group of the form $(g_1,\ldots, g_1,g_2,\ldots, g_2,\ldots, g_p)$, which is in turn isomorphic to $G^p$.  The nature of all possible isomorphisms is captured by \eqref{eq:thetagenerate}.  A quick intuitive argument for the lemma is as follows:  suppose we had checks of the form  $X_1^g X_2^g$ and $X_2^g X_3^g$ for every $g$.  Then, by considering the check obtained from their group commutator, 
\begin{equation}
    X_1^g X_2^g\cdot X_2^h X_3^h \cdot X_1^{g^{-1}} X_2^{g^{-1}}\cdot X_2^{h^{-1}} X_3^{h^{-1}} = X_2^{ghg^{-1}h^{-1}}~,
\end{equation}
we see that we also generate single-qudit checks of the form $X_2^g$.
These in turn can be multiplied by the first two kinds of checks to get $X_1^g$ and $X_3^g$ checks.   Repeatedly applying this trick and using that $[G,G]=G$ will lead to \eqref{eq:thetagenerate}.

Before proving the lemma, let us note that
since the code is $G$-covariant, for each $j \in \{1, \dots, 2n\}$, either there must
exist some $X$-check which contains $X^g_j$ for each $g \in G$ or no $X$-check acts non-trivially on index $j$, in which case there is
a ghost vertex. Therefore $\pi_j(\diagX) = G$ whenever there is an $X$-check containing $X^g_j$ for any non-identity $g$. Going forward, let $B_0 \subseteq \{1, \dots, 2n\}$ be the subset of $2n - L$ indices with no action by any $X$-check.

\begin{proof}[Proof of Lemma~\ref{lem:equivalentthm}]
First let us prove this fact for the case $L=2$.  Here we have $\diagX\le
G\times G$ with $\pi_1(\diagX)=\pi_2(\diagX)=G$.  Let $K_{1,2} =
\ker_{\diagX}(\pi_{1,2})$.  Using Goursat's Lemma, we know that
$K_{1,2}\unlhd G$ and $G/K_1 \cong G/K_2$.  Since $G$ is simple its only normal
subgroups are 1 and $G$. If $K_1=K_2=G$, then we know that for every $g\in G$,
$(g,1)\in \diagX$ and $(1,g)\in \diagX$.  Then $\diagX=G\times
G$. If instead $K_1=K_2=1$, then we know that $(1,1)\in \diagX$ is the only non-trivial element with identity in either factor.  If there were two elements $(g,h)$ and $(g,k)$ with the same $g$, then $(1, hk^{-1}) \in \diagX$, which would contradict our previous assertion.  At the same time  there is one element with $g$ in either the first or second entry since $\pi_{1}(\diagX) = \pi_{1}(\diagX) =G$. 
    Thus there must be a bijection $\theta : G\rightarrow G$ such that each element of $\diagX$ takes the form $(g,\theta(g))$.  Since $\diagX$ is a subgroup, $\theta$ must in fact be an automorphism.  Hence $\diagX\cong G$ and is generated by the form \eqref{eq:thetagenerate}.

Now let us return to the general case.   For every pair $\lbrace i,j\rbrace
\subset \lbrace 1,\ldots, 2n\rbrace$, define $\pi_{ij}:G^{2n}\rightarrow G^2$ by
$\pi_{ij} = (\pi_i,\pi_j)$.  Applying the result from the special case $L=2$,
$\pi_{ij}(\diagX)$ is either isomorphic to $G\times G$ or to $G$.  We
propose the following equivalence relation on $\lbrace 1,\ldots, L\rbrace$:
$i\sim j$ if and only if $\pi_{ij}(\diagX)\cong G$.  To check that it is
indeed an equivalence relation, we need to check that $j\sim k$ if $i\sim j$ and
$i\sim k$.   Indeed, if $i\sim j$ and $i\sim k$, then assuming $\pi_i(h) = g$,
for some $h\in \diagX$ and $g\in G$, we know $\pi_j(h) = \theta_j(g)$
and $\pi_k(h) = \theta_k(g)$.  
But this implies that when $\pi_j(h)=g$,
$\pi_k(h) = \theta_k(\theta_j^{-1}(g))= (\theta_k \circ
\theta_j^{-1})(g)$, and we see that $j\sim k$ since $\theta_k \circ
\theta_j^{-1}\in\mathrm{Aut}(G)$. Apply the equivalence relation above to $\lbrace 1,\ldots ,L\rbrace$ and define each equivalence class as a set $B_\alpha$ ($\alpha=1,\ldots, p$ if we find $p$ distinct equivalence classes).  This explicitly shows \eqref{eq:thetagenerate}.
\end{proof}

Now let us show that we must take each $\theta_j$ in \eqref{eq:thetagenerate} to be the trivial automorphism.  This will establish \eqref{eq:Xspecialform} and complete the second step of the proof.   To show this, suppose that there exist two elements\footnote{The labels chosen in this paragraph are only for convenience/explicitness; the argument is clearly general.  We also note that in each partition $B_\alpha$, we may always take one of the automorphisms to be identity without loss of generality.} $\lbrace 1,2\rbrace\in B_1$ with $\theta_1$ the identity automorphism but $\theta_2=\theta$ a non-trivial automorphism.  Then there exists an $X$-check for the code of the form $X^g_{B_1}X^g_{\mathrm{tot}}X^{g^{-1}}_{B_1}X^{g^{-1}}_{\mathrm{tot}}$ which acts non-trivially on a proper subset $B_1^\prime \subset B_1$ for which $1\notin B_1^\prime$.  This is a contradiction with \eqref{eq:thetagenerate}, as we see that $\diagX$ now has a larger set of generators.  The only resolution is that $B_1$ was a combination of multiple distinct types of $X$-checks, and that $B_1$ itself can be partitioned into multiple subsets in which, by Lemma \ref{lem:equivalentthm}, we can separately generate diagonal subgroups.   However, that would be a contradiction because then $\diagX$ would have been a larger subgroup than was given in \eqref{eq:thetagenerate}. Put simply, only when we choose all $\theta_j$ to be identity is it the case that we cannot generate smaller $X$-checks by commuting through $X^g_{\mathrm{tot}}$.  This establishes \eqref{eq:Xspecialform}.

The final step of the proof of Theorem \ref{thm:classify} is to show that we can, without loss of generality,
take the group word equations in the $Z$-checks to be the checks of a CW complex
quantum double -- or write the checks in an equivalent form that does have this
property.
Consider an arbitrary $Z$-check of weight $k$ enforcing $g_{q_1} \dots g_{q_k} =
1$. To get started, we reduce to the case where there are no ghost vertices. We
observe that
\begin{equation}
X_{\mathrm{tot}}^h\prod_{\alpha}X_{B_\alpha}^{h^{-1}} = X_{B_0}^h .
\end{equation}
Since each $Z \in \mathcal S_Z$
commutes with $X_{\mathrm{tot}}^g$ by assumption, this implies that $X_{B_0}^h$
is a logical operator. Therefore, consider an enlarged stabilizer group which
includes $X_{B_0}^h$. We know without loss of generality that there exists a
codeword generated by the orbits of the $X$-checks on the group word $g_{q_1} =
\ldots = g_{q_k} = 1$ by the assumption that the $X$-checks preserve $\mathcal
C_Z$, which always contains the tensor product of identity-labeled states, $\ket{\mathbf 1} = \ket{1}^{\otimes n}$.
Then we may find a product of $X$-checks $X_{B_1}^{a_1} \dots X_{B_p}^{a_p}$ such that every $q_1 \ldots q_k$
is acted upon by both left- and right-multiplication (possibly including $X_{B_0}^h$). 

Next, we see how the above \(X\)-check acts on the the \(Z\)-check constraint, $Z^{\{1\}}_{q_1, \dots,
  q_k}\ket{[\vb 1]} = \ket{[\vb 1]}$, for the codeword \(|[\mathbf{1}]\rangle\) constructed out of the identity-element orbit of the identity state [see Eq.~\eqref{eq:basis-codewords}].
  This constraint is equivalent to the weight-\(k\) group law $g_{q_{1}}g_{q_{2}}\cdots g_{q_{k}}=1$.
  Plugging in the two group elements arising from left- and right-multiplication of each coordinate by the \(X\)-checks $X_{B_1}^{a_1} \dots X_{B_p}^{a_p}$ yields a weight-\(2k\) group law, which in turn specifies the gluing rule for faces of the CW complex. 

\begin{lem}
  The code space stabilized by $\mathcal S_Z, \mathcal S_X$ is equivalent to one
  where each $Z$-check specifies the gluing rules of a 2-cell.
  In other words, each $Z$-check corresponds to a group-law
  \begin{equation}
    f(g_{q_1}, \dots, g_{q_k}) = (h_0, h_1^{-1}, h_1, h_2^{-1}, h_2, \dots, h_{m}^{-1}, h_{m}, h_0^{-1})
  \end{equation}
  where each $h_i = g_{q_j}$ or $g_{q_j}^{-1}$ for some $j$.
  \label{lem:Zcheckform}
\end{lem}

\begin{proof}
  Consider 1-dimensional CW complex $\Sigma$ built from the $X$-checks following
  Definition \ref{def:generalizedQDsmooth}, with a single ghost vertex $B_0$ representing
  the logical operator $X_{B_0}^h$ if $B_0$ is non-empty. Given any $Z$-check with a corresponding
  group-constraint $g_{q_1}g_{q_2}\cdots g_{q_k} = 1$,
 $2k < \alpha(G)$ by
  assumption, the corresponding group word $f$ must be freely trivial 
 [Def.~\ref{def:freely-trivial}].
  By the properties of a free group,
  this means it is related to 1 by a sequence of elementary cancellations of
  adjacent inverses. This pairing of inverses corresponds to a 
  non-crossing partition. The following is an example showing how to associate a
  freely trivial law $f$ to a non-crossing partition:
  \begin{equation}
    f(a,b,c,d, e, g, h) = \wick{(\c4 g, \c3 h^{-1}, \c2 a, \c1 b^{-1}, \c1 b, \c2 a^{-1}, \c2 c, \c1 d^{-1}, \c1 d, \c1 e^{-1}, \c1 e, \c2 c^{-1}, \c3 h, \c4 g^{-1})}
    \label{eq:examplegrouplaw}
  \end{equation}
 To simplify this picture, there is also a standard way to associate such a
 partition to a tree: every time one of the brackets is entered, the tree
 descends a level, and every time one is exited, the tree ascends one level.
 Using this prescription, the tree associated to $f$ is shown in the left panel
 of Fig.~\ref{fig:gluing_rules_example}.

  \begin{figure}
  \centering
  \def\svgwidth{.7\textwidth}
  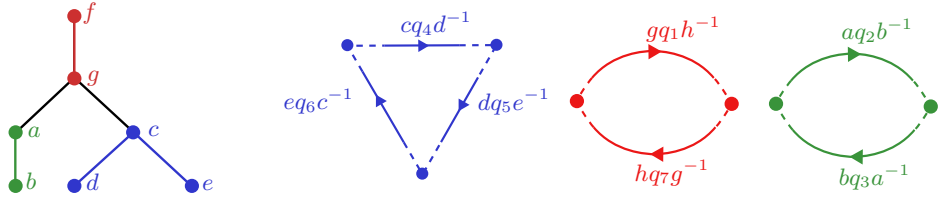
  \caption{Example showing the association of the group law in
    \eqref{eq:examplegrouplaw} to small cycles in CW complex, each of which gets
    its own $Z$-check. On the left, the tree representation of the group law is
    shown, with even-depth subtrees colored in red, green, and blue. On the
    right, the cycles corresponding to these subtrees are shown. The solid lines
    show physical qudits, and the dotted lines illustrate how the action of the
    $X-$checks glues the edges to vertices.}
   \label{fig:gluing_rules_example}
\end{figure}

Our claim is that we can add $Z$-checks to prune leaves from the tree without
changing the codespace until the tree is exhausted. First, any leaf with an odd
depth (starting from a single node at depth 1) corresponds to a pairing
beginning at an odd index and ending at an even index.
Such a pairing represents an edge which is attached to the same vertex at either
side, or a length-1 loop. Since $d_Z > 1$, we may add a $Z$-check and prune it
from the tree, such that every remaining leaf has even depth. Clearly the added
$Z$-check corresponds to gluing a 2-cell to this small loop.
  
  Now, every remaining leaf corresponds to a partition which begins on
  an even index and ends on an odd index, so its parent $p$ must begin on an odd
  index and end on even index. Therefore the leftmost child $n_{\ell}$ of $p$ acts on the edge
  $q_{i}$ by right-multiplication, and the rightmost child $n_r$ of $p$ acts on the edge
  $q_{f}$ by left-multiplication. Then $p$ acts on $q_i$ by left multiplication
  and on $q_f$ by right-multiplication, whence $q_i \to q_j$ is a loop in
  $\Sigma$. 

  This is illustrated by an example in Fig.~\ref{fig:gluing_rules_example}.
  Since $k < d_Z$ by hypothesis, this loop must have trivial holonomy
  in all codewords. Therefore, we may add an additional $Z$-check and prune the
  subtree rooted at $p$ from the graph.  This process is iterated until the
  remainder is either empty or is a disjoint union of depth-2 trees.
  $f$ has the form of Lem.~\ref{lem:Zcheckform} if and only if it
  can be represented as a depth-2 tree, so all the new checks have this form,
  and we are done.
\end{proof}

Together, these lemmas show that under mild conditions, every covariant code on a
non-Abelian simple group is a generalized quantum double on a 2D CW complex,
where our bounds on the code parameters and construction of logical operators
apply. This implies that there is no purely geometric specification of
stabilizers for a non-Abelian code which can be used to produce good quantum
codes, and a geometry which relies on specific group structure would be needed to overcome these bounds.

\end{appendix}

\bibliography{thebib}
\end{document}